\documentclass[11pt]{amsart}
\usepackage[usenames,dvipsnames]{color}
\usepackage{amsmath,amsthm,amsfonts,amssymb,multicol,amscd,amsbsy,dsfont}
\usepackage{xcolor}
\usepackage{graphicx}
\usepackage{booktabs}
\usepackage{array}
\usepackage{subfigure}
\usepackage{enumerate}
\usepackage{enumitem}
\usepackage{url}
\usepackage[nodayofweek]{datetime}
\usepackage[english]{babel}
\usepackage{mathtools}
\usepackage[toc,page]{appendix}
\usepackage{verbatim} 
\usepackage{amsthm}
\usepackage{enumerate}
\usepackage{bbm} 
\usepackage[normalem]{ulem} 
\usepackage{bm}
\usepackage{hyperref}

\usepackage[margin=1.5in]{geometry}

\newcommand{\be}{\begin}
\newcommand{\e}{\end}
\newcommand{\beq}{\begin{equation}}
\newcommand{\eeq}{\end{equation}}
\newcommand{\beqs}{\begin{equation*}}
\newcommand{\eeqs}{\end{equation*}}
\newcommand{\bal}{\begin{align}}
\newcommand{\eal}{\end{align}}
\newcommand{\bals}{\begin{align*}}
\newcommand{\eals}{\end{align*}}
\newcommand{\nn}{\nonumber}

\newcommand{\ol}{\overline}

\newcommand{\la}{\left\langle}
\newcommand{\ra}{\right\rangle}

\renewcommand{\l}{\left}
\renewcommand{\r}{\right}

\newcommand{\T}{\mathbb{T}}
\newcommand{\N}{\mathbb{N}}
\newcommand{\Q}{\mathbb{Q}}
\renewcommand{\d}{\mathrm{d}} 

\newcommand{\mf}{\mathfrak}
\newcommand{\tor}{\T}

\newcommand{\set}[1]{\mathbb{#1}}

\newcommand{\curly}[1]{\mathcal{#1}}

\newcommand{\R}{\set{R}}
\newcommand{\C}{\set{C}}
\newcommand{\Z}{\set{Z}}

\newcommand{\eps}{\epsilon}

\newcommand{\al}{\alpha}
\newcommand{\de}{\delta}

\newcommand{\vp}{\varphi}

\newcommand{\ind}{\mathbbm{1}}		
\newcommand{\pr}{\mathbb{P}}		

\newcommand{\supp}{\,\mathrm{supp}\,}

\newcommand{\del}{\partial}

\newcommand{\spec}{\mathrm{spec}\,}
\def\one{\mathds{1}}
\renewcommand\ln{\log}  

\newcommand{\Tr}{\mathrm{Tr}}	

\newtheorem{thm}{Theorem}[section]

\newtheorem{cor}[thm]{Corollary}

\theoremstyle{definition}
\newtheorem{defn}[thm]{Definition}

\numberwithin{equation}{section}

\theoremstyle{remark}
\newtheorem{rmk}[thm]{Remark}

\def\dotuline{\bgroup
  \ifdim\ULdepth=\maxdimen  
   \settodepth\ULdepth{(j}\advance\ULdepth.4pt\fi
  \markoverwith{\begingroup
  \advance\ULdepth0.08ex
  \lower\ULdepth\hbox{\kern.15em .\kern.1em}%
  \endgroup}\ULon}

\def\dashuline{\bgroup
  \ifdim\ULdepth=\maxdimen  
   \settodepth\ULdepth{(j}\advance\ULdepth.4pt\fi
  \markoverwith{\kern.15em
  \vtop{\kern\ULdepth \hrule width .3em}%
  \kern.15em}\ULon}
\allowdisplaybreaks

\def\nn{\nonumber}

\def\oddsum{\quad o\!\!\!\!\!\sum}
\def\evensum{\quad e\!\!\!\!\!\sum}
\def\oddsumm{\quad o\!\!\!\!\!\!\!\!\!\!\!\!\sum}
\def\evensumm{\quad e\!\!\!\!\!\!\!\!\!\!\!\sum}

\begin{document}

\title[Effective approach for the skew shift]{Effective multi-scale approach to the Schr\"{o}dinger cocycle over a skew shift base} 

\author{R.\ Han, M.\ Lemm, W.\ Schlag}

\begin{abstract}
We prove a conditional theorem on the positivity of the Lyapunov exponent for a Schr\"{o}dinger cocycle over a skew shift base with a cosine potential and the golden ratio as frequency. For coupling below $1$, which is the threshold for Herman's subharmonicity trick, we formulate three conditions on the Lyapunov exponent in a finite but large volume and on the associated large deviation estimates at that scale. 
Our main results demonstrate that these finite-size conditions imply the positivity of the infinite volume Lyapunov exponent. This paper shows that it is possible to make the techniques developed for the study of Schr\"{o}dinger operators with deterministic potentials,  based on large deviation estimates and the avalanche principle, effective. 
\end{abstract}

\address{Institute for Advanced Study, School of Math, 1 Einstein Drive, Princeton, NJ 08540, U.S.A.}

\thanks{The authors thank the Institute for Advanced Study, Princeton, for its hospitality during the 2017-18 academic year.  The third  author was partially supported by the NSF, DMS-1500696. The authors thank James Maynard and Silvius Klein for helpful conversations, and Jean Bourgain and Thomas Spencer for their interest in this work. }

\maketitle

\section{Introduction}
The study of Lyapunov exponents occupies a central role in ergodic theory and dynamical systems. They arise in a multitude of distinct settings, such as diffemorphisms on a manifold, 
chaotic dynamics in nonlinear systems as exhibit by  the standard map, cocycles defined over some base, and the theory of localization. Perhaps the most fundamental question about Lyapunov exponents relates to their simplicity. 
Or more quantitatively, to the gaps between them. In the case of $SL_2(\R)$ cocycles this amounts to the question of positivity of the top Lyapunov exponent. Another much studied property of these exponents concerns 
their continuity relative to external parameters. For a beautiful introduction to this field see the textbook~\cite{Viana}. 

This paper studies Schr\"{o}dinger cocycles 
\begin{align*}
(x,v) \in X\times\R^2 &\mapsto (Tx,A_\lambda(x,E)v)\\
A_\lambda(x,E) &= \left[\begin{matrix}
                   \lambda f(x) - E & -1 \\ 1 & 0
                  \end{matrix} \right] \in SL_2(\R)
\end{align*}
where $(X,\mu,T)$ is some ergodic system, $\lambda, E\in \R$ and $f:X\to\R$ is measurable. These cocycles arise in the spectral analysis of the operators
\[
 \big(H_{\lambda,x}\psi\big)_n = \psi_{n+1}+\psi_{n-1} + \lambda f(T^n x) \psi_n, \quad n\in \Z
\]
Indeed, solutions of $H_{\lambda,x}\psi = E\psi$ are given by 
\begin{align}
 \binom{\psi_{n+1}}{\psi_n} &= M_n(x;\lambda, E) \binom{\psi_{1}}{\psi_{0}}, \label{AM}\\ 
M_n(x;\lambda, E) &=  \prod_{j=n}^1 A_\lambda(T^j x,E), \quad n\ge 1   \nn
 \end{align}
The growth of solutions to~\eqref{AM} $\mu$-a.e.\ in~$x$ is governed by the Lyapunov exponent 
\[
 L(\lambda,E) = \lim_{n\to\infty} n^{-1} \int_X \log \|M_n(x;\lambda,E)\|\, \mu(dx)
\]
which always exists by subadditivity. By unimodularity of the matrices, $L(\lambda,E)\ge0$. The main issue is then to determine strict positivity. We remark that by classical
ergodic theory (F\"{u}rstenberg-Kesten theorem, Kingman's subadditive ergodic theorem~\cite{Viana}), 
\[
  n^{-1}   \log \|M_n(x;\lambda,E)\| \to L(\lambda,E) \qquad \mu\text{\ \ a.s.}
\]
as $n\to\infty$. F\"{u}rstenberg's theorem~\cite{Fur}, shows that $L>0$ for all $\lambda, E$ for $T$ the Bernoulli shift and $\mu$ a nontrivial probability distribution. Herman's subharmonicity argument~\cite{Her},
which is recalled in Section~\ref{sec:Herman}, shows that $L(\lambda,E)\ge \log\lambda>0$ if $\lambda>1$, $X=\T$, $f(x)=2\cos(2\pi x)$, and $Tx=x+\omega$ a rotation (for general analytic $f$ and large $\lambda$ see~\cite{SoSp}) . On the other hand, 
one has $L(\lambda,E)=0$ for all $0<\lambda<1$ and $E\in \spec(H_{\lambda,x})$. The latter is the spectrum of the Harper or almost Mathieu operator 
\[
 \big(H_{\lambda,x}\psi\big)_n = \psi_{n+1}+\psi_{n-1} +2\lambda \cos(2\pi (x+n\omega)) \psi_n, \quad n\in \Z
\]
which does not depend on $x$ (assuming $\omega$ irrational). In particular, $L(\lambda,0)=0$, cf.~\cite{BelSim,Dam}. 

In contrast to the Harper operator, its analog over the skew-shift base is conjectured to exhibit positive Lyapunov exponents for all $\lambda>0$ and~$E$. To be specific, let $X=\tor^2$,
$T(x,y)=(x+y,y+\omega)$, where $\omega$ is irrational (or Diophantine). Iterating $T$ yields
\beq\label{Mnskew}
\begin{aligned}
M_n(x,y;\lambda, E) &=  \prod_{j=n}^1 \left[\begin{matrix}
                   2\lambda f(x+jy+j(j-1)\omega/2, y+j\omega) - E & -1 \\ 1 & 0
                  \end{matrix} \right] 
\end{aligned}
\eeq
The presence of $j^2\omega/2$ in these matrices appears to be the origin of the conjectured exponential growth of the norm of these matrices for all $E$ (assuming $\partial_x f(x,y)\not\equiv 0$, with $f$ analytic). 
In fact, the distribution of the fractional parts of $\{j^2\omega\}_{j=1}^N$ is known to be ``random'' in some sense as $N\to\infty$ for generic $\omega$, see the Poissonian 
conjecture in~\cite{RSZ}, as well as~\cite{MS,DRH}. Note that this is in stark contrast to the distribution of $\{j\omega\}_{j=1}^N$. 

However, not only is this randomness property in and of itself delicate (see some negative results to this effect in~\cite{RSZ}), but how to use it in the context of~\eqref{Mnskew} is 
entirely unclear. As far as rigorous results are concerned, Bourgain~\cite{B2} proved that for all $\lambda>0$ there exists a set of $\omega\in\tor$ with positive measure 
(which decreases to~$0$ as $\lambda\to0$), 
so that the operator
\[
 (H \psi)_n = \psi_{n+1} + \psi_{n-1} + \lambda\cos\big( n(n-1)\omega/2\big) \psi_n
\]
exhibits point spectrum whose closure has positive measure. This was the first result of its kind which showed that for small $\lambda$ the skew shift leads to completely different behavior than the shift, i.e.,
potentials $\cos( n \omega)$.

In this paper we present an effective multi-scale machinery aiming at positivity of the Lyapunov exponent for the matrices 
\beq\label{Mnskew2}
M_n(x,y;\lambda, E) =  \prod_{j=n}^1 \left[\begin{matrix}
                   2\lambda \cos\big(2\pi(x+jy+j(j-1)\omega/2)\big) - E & -1 \\ 1 & 0 
                  \end{matrix} \right] 
\eeq
uniformly in $E$, and in the range $0<\lambda\le 1$. We fix $\omega$ to be the golden ratio. By the aforementioned estimate by Herman, one has $L(\lambda,E)\ge\log\lambda>0$ for $\lambda>1$. So only $\lambda\le 1$ is of interest here. 
The basis of our analysis is the inductive argument from~\cite{BGS}, which established Anderson localization for large $\lambda$ for the skew shift model, at the expense of removing a small
set (in measure) of frequencies $\omega$ and phases $(x,y)$ (the largeness of $\lambda$ depended on the smallness of the measure of excluded parameters). The proof in~\cite{BGS} is not effective,
and it was not possibly to explicitly determine the size of admissible $\lambda$ in relation to the other parameters. 

To formulate our main results, recall the finite-volume Lyapunov exponents
\beq\nn
L_N(\lambda, E):=\int_{\T^2} \frac{1}{N} \log{\|M_N(x,y;\lambda, E)\|}\, dx dy,
\eeq
and their limits $L=\lim_{N\to\infty}L_{N}$. 
We quantify the failure of the F\"{u}rstenberg-Kesten theorem via the following level sets $\mathcal{B}_N$, 
\beq\nn
\mathcal{B}_N:=\left\lbrace (x,y)\in \T^2:\ \l| \frac{1}{N}\log{\|M_N(x,y;\lambda, E)\|}-L_N(\lambda, E) \r|>\frac{1}{10} L_N(\lambda, E) \right\rbrace. 
\eeq

The machinery developed in this paper establishes a method for checking the positivity of the Lyapunov exponent $L(\lambda,E)$ by verifying information on a finite, initial scale. We could have formulated a very general ``finite-size criterion'' which establishes $L(\lambda,E)>0$ under appropriate assumptions on the initial scale and for appropriate values of various other parameters. Instead, we have opted to present three representative theorems that can be obtained from the machinery developed in this paper by making specific choices.

These representative theorems differ by the precise assumptions (i)-(iii) made at the initial scale. We comment on this further after the first theorem. Moreover, the various other parameters appearing in our proof are identical in all three cases. These parameters are only chosen in the final part of the proof, Section \ref{sec:numbers}, so they can easily be modified.

For a Borel set,  $|\cdot|$ denotes the Lebesgue measure. 

\begin{thm}\label{thm:main}
Consider the skew-shift cocycle given by \eqref{Mnskew2}.
Let $\omega$ be the golden ratio and let $\lambda \in [1/2, 1]$.
Let $N_0:=2\times {10^{37}}$.
Assume that for some energy $E\in [-2-2\lambda, 2+2\lambda]$ the following hold:
\begin{enumerate}[label=(\roman*).]
\item $L_{N_0}(\lambda, E)\geq 2\times 10^{-4}$,
\item $L_{N_0}(\lambda, E)-L_{2N_0}(\lambda, E)\leq L_{N_0}(\lambda, E)/8$,
\item $\max{(|\mathcal{B}_{N_0}|, |\mathcal{B}_{2N_0}|)}\leq N_0^{-21}$.
\end{enumerate}
Then we have 
\[L(\lambda, E)\geq \frac{1}{2} L_{N_0}(\lambda, E)>0.\]
\end{thm}

Before we give the two alternative theorems, we comment on conditions (i)-(iii).

\be{rmk}\label{rmk:main}
\begin{enumerate}[label=(\roman*).]
\item
First, one might expect that $L(\lambda,E)>c\lambda^{2}$ holds for small $\lambda$, by analogy with the Figotin-Pastur  asymptotics in the random case.  Numerical experimentation suggests that  is indeed the case for our model with $c>10^{-2}$ (with a generous margin of error). Therefore, we would expect to have $L_{N_0}(\lambda, E)\geq 2\times 10^{-3}$. Condition~(i) was chosen to  allow for an even wider  margin. We remark that we can lower the number $2\times 10^{-4}$ to basically any positive constant, at the expense of increasing $N_{0}$. 
\item
Condition~(ii) is known to hold if the Lyapunov exponent is positive and $N_{0}$ is large enough. Indeed, it follows from the methods in~\cite{GS} that $$L_{N_0}(\lambda, E)-L_{2N_0}(\lambda, E)\leq cL_{N_0}(\lambda, E)/N_{0}$$ with some absolute constant $c\sim1$, see Section~\ref{sec:MS} below for the details. Given the size of $N_{0}$, condition~(ii) is indeed asking for very little. 
\item
Finally, condition (iii) is some weak form of a {\em large deviation estimate} as in~\cite{BG, GS, BGS}. In fact, analogy with these references suggests that a bound of the form  $|\mathcal{B}_{N}| < \exp(-N^{\frac{1}{10}}) $ should hold for large $N$ (and perhaps a much stronger bound, say with $N^{\frac12}$ or larger). For (iii) to hold in this case would then require $N>8\cdot 10^{31}$, which is within our range.  It is important to note that condition (iii) differs strongly from (i) and~(ii). Indeed, while the latter conditions are intimately related to the $L>0$, (iii) is not. For the shift dynamics it is known that the large deviation estimates hold apriori, i.e., without any reference to the positivity of the Lyapunov exponent, see~\cite{BG, GS, B1}. For the skew shift, however, such apriori derivations are currently not known. Rather, we rely an inductive procedure that uses lower bounds on $L_n$, the Lyapunov exponents in finite volume. 
\end{enumerate}
\e{rmk}

We now state two further representative theorems. These alternative finite-size criteria both involve much smaller initial scales $N_0$, at the price of having a more restrictive assumption (iii) on the measure of the set $\curly{B}_{N_0}$. 

\begin{thm}\label{thm:main2}
Consider the skew-shift cocycle given by \eqref{Mnskew2}.
Let $\omega$ be the golden ratio and let $\lambda \in [1/2, 1]$.
Let $N_0:=3\times 10^5$. Assume that for some energy $E\in [-2-2\lambda, 2+2\lambda]$ the following hold:
\begin{enumerate}[label=(\roman*).]
\item $L_{N_0}(\lambda, E)\geq 2\times 10^{-4}$,
\item $L_{N_0}(\lambda, E)-L_{2N_0}(\lambda, E)\leq L_{N_0}(\lambda, E)/8$,
\item $\max{(|\mathcal{B}_{N_0}|, |\mathcal{B}_{2N_0}|)}\leq N_0^{-141}$.
\end{enumerate}
Then we have 
\[L(\lambda, E)\geq \frac{1}{2} L_{N_0}(\lambda, E)>0.\]
\end{thm}

The upper bound in assumption (iii) is more restrictive than in Theorem \ref{thm:main2}. Importantly, it is still polynomial in nature. Hence, in view of Remark \ref{rmk:main} (iii), it may hold depending on the precise kind of exponential decay that is presumably exhibited by the true $|\mathcal{B}_{N}|$.

In the next representative result, Theorem \ref{thm:main3}, we strengthen assumption (i) somewhat (in a way that is compatible with the numerics described in Remark \ref{rmk:main} (i) above). This allows to reduce the initial scale even further,  to the value $N_0=3\times 10^4$, which may be amenable to numerical investigation. 

\begin{thm}\label{thm:main3}
Consider the skew-shift cocycle given by \eqref{Mnskew2}.
Let $\omega$ be the golden ratio and let $\lambda \in [1/2, 1]$.
Let $N_0:=3\times 10^4$. Assume that for some energy $E\in [-2-2\lambda, 2+2\lambda]$ the following hold:
\begin{enumerate}[label=(\roman*).]
\item $L_{N_0}(\lambda, E)\geq 2\times 10^{-3}$,
\item $L_{N_0}(\lambda, E)-L_{2N_0}(\lambda, E)\leq L_{N_0}(\lambda, E)/8$,
\item $\max{(|\mathcal{B}_{N_0}|, |\mathcal{B}_{2N_0}|)}\leq N_0^{-165}$.
\end{enumerate}
Then we have 
\[L(\lambda, E)\geq \frac{1}{2} L_{N_0}(\lambda, E)>0.\]
\end{thm}

Regarding assumption (iii), the comment made after Theorem \ref{thm:main2} still applies. In particular, the relatively small value of $N_0$ in this result elevates the problem of finding an {\bf analytical proof} of~(iii). Indeed, (i) and~(ii) are accessible numerically by a Figotin-Pastur expansion, for example, but it seems completely unreasonable to ask for a computer assisted proof of~(iii). 

\smallskip

The restriction $\lambda \in [1/2, 1]$ was chosen for convenience. In fact, our methods apply to any given interval of the form $[\lambda_{0},1]$, $\lambda_{0}>0$, albeit with increasing~$N_{0}$  as $\lambda_{0}\to0$.  Similarly, the golden ratio was chosen for simplicity. One can replace it by a class of Diophantine frequencies obeying an explicit Diophantine condition. 

\smallskip

  It remain to be seen what the true range of applicability of our methods are, and to which extent they can also be refined. It may be possible to verify assumptions (i) and (ii) of Theorem \ref{thm:main3} numerically. However, the measure estimates (iii) would seem the most delicate to  check reliably. 

\smallskip

The methods in this paper are an adaptation of those in \cite{GS,BGS,B1}. One of our motivations was to obtain an effective rendition of the techniques based on harmonic analysis (subharmonic functions, Riesz representation theorem, John-Nirenberg type estimates for $BMO$ functions) in combination with linear algebra and the geometry of matrix products (``avalanche principle'',~\cite{GS}). This had never been attempted before, but we show that it is possible to do so.

\section{Effective Riesz Representation}

It is of fundamental importance to the entire method to make the underlying potential theory effective. To this end it is most convenient to remain on the disk since other geometries will lead to complicated Green functions. The disk will suffice for our purposes, thanks to a variant of Herman's regularization \cite{Her}, which we present in Section \ref{sec:Herman}.

\begin{defn}\label{def:Poisson}
Given $R>0$, we write $D_R$ for the open disk of radius $R$ around the origin in $\C$. Let $z=re(\phi)$, we write $P_z(\theta)$ for the Poisson kernel
$$
P_z(\theta):=\frac{1-r^2}{1-2r\cos{(2\pi (\phi-\theta))}+r^2}.
$$
The following constants will be used throughout, with $1<R_{2}<R_{1}<R$: 
\beq\label{def:constantsB0B1B2}
\begin{aligned}
B_0(R,R_1,R_2):=&
\frac{1}{2\log(R/R_1)}\left(\frac{R_1+R_2}{R_1-R_2}\right)
\times
\be{cases}
\log R \qquad &\textnormal{if } R^2-R_1^2>R,\\
\log\l(\frac{R^2}{R^2-R_1^2}\r), \qquad &\textnormal{if } R^2-R_1^2<R.
\e{cases}\\
B_1(R,R_1,R_2):=&B_0(R,R_1,R_2) \frac{8R_2}{R_2^2-1},\\\
B_2(R,R_1,R_2):=&B_0(R,R_1,R_2)\\
&\frac{16\pi  (R_2^2-1)\sqrt{16 R_2^2-\l(\sqrt{R_2^4+34 R_2^2+1}-1-R_2^2\r)^2}}{\l(3R_2^2+3-\sqrt{R_2^4+34 R_2^2+1}\r)^2}
\end{aligned}
\eeq
as well as
\beq\label{def:constantsB3B4}
\begin{aligned}
B_3(R,R_1,R_2) := \sqrt{5B_2(R,R_1,R_2)+\frac{10\pi}{\log{R/{R_1}}}}.
\end{aligned}
\eeq
\end{defn}

The main result of this section is the following Riesz representation theorem for subharmonic functions. The essential feature here are the explicitly computable constants.  Recall that a subharmonic function in some domain $\Omega\subset\C$ is an upper semicontinuous function $u:\Omega\to\R\cup\{-\infty\}$ which satisfies the sub-mean value property in $\Omega$. 

\be{thm}\label{Riesz}
Let $1<R_2<R_1<R$ and let $v:\ol{D_R}\to \R\cup\{-\infty\}$ be a subharmonic function satisfying 
\beq
v(z)\leq B,\qquad v(0)=m.
\eeq
Then, for all $w\in D_{R_1}$, we have the Riesz representation
\beq
\label{eq:thmriesz}
\begin{aligned}
v(w)=&\int_{D_{R_1}} \log\l|z-w\r|\mu(\d z)+h(w),
\end{aligned}
\eeq
where
\be{enumerate}[label=(\roman*)]
\item $\mu$ is a positive measure satisfying the bound
\beq\label{est:muthm11}
\mu(D_{R_1})\leq \frac{B-m}{\log(R/R_1)},
\eeq
\item $h$ is harmonic on $D_{R_1}$ and satisfies the following bounds
\beq\label{est:hthm11}
\begin{aligned}
\min_{c\in \R}\max_{|w|\leq R_2}\ |h(w)-c|&\leq B_0(R,R_1,R_2)(B-m),\\
\l|\frac{d^k}{d\varphi^k}h(e(\varphi))\r| &\leq B_k(R,R_1,R_2)(B-m), \qquad (k=1,2),
\end{aligned}
\eeq
\e{enumerate}
\e{thm}

The proof of this theorem will occupy the rest of this section. 

\subsection*{Proof of Theorem \ref{Riesz}}
The basic idea of the proof is that the equation $\mu=\frac{1}{2\pi}\Delta v$ holds in the distributional sense, with $\mu$ a  positive measure.
Without loss of generality, we may assume that $v$ is smooth. If this is not the case, we convolve $v$ with a radial nonnegative mollifier. The submean property then guarantees monotone convergence. We skip these technical details.

\subsection{Riesz representation}

Rescaling the unit disk yields the Green function on any disk. 

\be{lm}[Green's function for the disk]\label{lm:green}
The function $G:D_R\times D_R\to \R$ given by
$$
G(z,w):=\frac{1}{2\pi }\log\l|\frac{R(z-w)}{R^2-z\ol{w}}\r|
$$
satisfies $\Delta_z G(z,w)=\delta_w$ and $G(z,w)=0$ when $|z|=R$.
\e{lm}
\be{proof}
To see this, notice that $G(z,w)=G_1(z/R,w/R)$ where $$G_1(z,w):=\frac{1}{2\pi }\log\l|\frac{z-w}{1-z\ol{w}}\r|$$ is the Green function of the unit disk.
\e{proof}

Let $w\in D_{R_1}$. By Green's second identity for the domain $D_R$, we have
$$
v(w)-\int_{D_R}G(z,w) \Delta v(z)\, \mathrm{Vol}(\d z)
= \int_{\del D_R} v(z) \frac{\del G}{\del n_z}(z,w) \, \sigma(\d z), 
$$
where $\mathrm{Vol}$ is the standard volume measure and $\sigma$ is the (unnormalized) arclength measure on the circle $\del D_R$. Since $v$ is smooth and subharmonic, $\Delta v$ is a non-negative, continuous function, call it $2\pi \mu$. Therefore

\begin{align}\label{eq:green}
v(w)=\int_{D_R}2\pi G(z,w)\, \mu(\d z)+h_0(w),
\end{align}
where
\begin{align}
h_0(w):=\int_{\del D_R} v(z) \frac{\del G}{\del n_z}(z,w) \, \sigma(\d z).
\end{align}

By Lemma \ref{lm:green}, we then have Riesz representation with the functions
\beq\label{eq:riesz}
v(w)=\int_{D_{R_1}}\log{|z-w|}\, \mu(\d z)+h(w),
\eeq
where
\beq\label{eq:hdefn}
h(w):=\int_{D_R\setminus D_{R_1}} \log{\l|\frac{R(z-w)}{R^2-z\overline{w}}\r|}\, \mu(\d z)+\int_{D_{R_1}} \log\l|\frac{R}{R^2-z\ol{w}}\r|\, \mu(\d z)+h_0(w)
\eeq

\be{lm}
$h(w)$ is harmonic in $D_{R_1}$.
\e{lm}

\be{proof} 
Write $w=re^{2\pi i \varphi}$. The first and second term in \eqref{eq:hdefn} are harmonic because they are real parts of analytic functions on $D_{R_1}$. For the third term, recall that $\frac{\del G}{\del n_z}(z,w)$ is  the Poisson kernel whence 
\beq\label{eq:h0=poisson}
h_0(w)=\int_{0}^1 v(Re(\theta)) P_{r/R}(\varphi-\theta)\, \d\theta.
\eeq
The Poisson kernel is harmonic in all of $D_R$ and this proves the lemma.
\e{proof}

\subsection{Control of the Riesz mass}
\be{lm}\label{lm:muest}
We have the following bound on the Riesz mass: 
\beq
\mu(D_{R_1})\leq \frac{B-m}{\log(R/R_1)}.
\eeq
\e{lm}

\be{proof}
Taking $w=0$ in \eqref{eq:green}, we see
$$
(\log{R/R_1}) \mu(D_{R_1})\leq \int_{D_R} \log{\frac{R}{|z|}}\, \mu(\d z)=h_0(0)-v(0)\leq B-m,
$$
in which we used 
\beq\label{eq:h0<B}
h_0(0)\leq B,
\eeq 
which comes from the maximum principle and the fact that $h_0(w)$ is the harmonic function on $D_R$ with boundary values $v({\del D_R})$ by \eqref{eq:h0=poisson}.
\e{proof}

\subsection{Control of the harmonic part}

We have the following estimate for the harmonic part.

\be{lm}\label{lm:hest}
Let $1<R_2<R_1<R$.
$$\min_{c\in \R}\max_{|w|\leq R_2}\ |h(w)-c|\leq B_0(R,R_1,R_2) (B-m),$$ with constant $B_0(R,R_1,R_2)$ given by~\eqref{def:constantsB0B1B2}.
\e{lm}
\be{proof}
We will first prove an upper bound and then use Harnack's inequality to conclude a lower bound.
From \eqref{eq:hdefn}, and  $G(z,w)\leq 0$ on $D_R\times D_R$,  
$$
h(w)\leq \int_{D_{R_1}}\log{\l| \frac{R}{R^2-z\overline{w}}\r|} \, \mu(\d z)+h_0(w).
$$
From \eqref{eq:h0<B}, we infer that for all $w\in D_{R_1}$,
$$
h(w)\leq \log{\l| \frac{R}{R^2-R_1^2}\r|} \mu(D_{R_1})+B.
$$
Now we distinguish cases. On the one hand, if $R<R^2-R_1^2$, then the logarithm is negative and \eqref{eq:hubb} implies
$$
h(w)\leq B.
$$
On the other hand, if $R>R^2-R_1^2$, then we use Lemma \ref{lm:muest} to obtain
$$
\begin{aligned}
h(w) \leq&\frac{1}{\log{R/{R_1}}}\left(B\log{\frac{R^2}{R_1(R^2-R_1^2)}}-m\log{\frac{R}{R^2-R_1^2}}\right)\\
=&\frac{m\log{R}-B\log{R_1}}{\log{R/{R_1}}}+\frac{\log{R^2/(R^2-R_1^2)}}{\log{R/{R_1}}}(B-m).
\end{aligned}
$$
Combining the two cases, we arrive at the upper bound 
\beq\label{eq:hubb}
\begin{aligned}
h(w) \leq\al,\qquad 
\al:=
\be{cases}
B,\qquad &\textnormal{if } R^2-R_1^2>R,\\
\frac{m\log{R}-B\log{R_1}}{\log{R/{R_1}}}+\frac{\log{R^2/(R^2-R_1^2)}}{\log{R/{R_1}}}(B-m),\qquad &\textnormal{if } R^2-R_1^2<R.
\e{cases}
\end{aligned}
\eeq
Consider the non-negative harmonic function $\alpha-h(w)$ on $D_{R_1}$. 
By Harnack's inequality,  
$$
\alpha -h(w)\leq \frac{R_1+|w|}{R_1-|w|}(\alpha -h(0)), 
$$
which implies the following lower bound
\beq\label{eq:hlbb1}
h(w)\geq \frac{R_1+|w|}{R_1-|w|} h(0)-\frac{2|w|}{R_1-|w|} \alpha.
\eeq
By \eqref{eq:riesz} with $w=0$ and \eqref{lm:muest}, we have
$$
h(0)=v(0)-\int_{D_{R_1}}\log{|z|}\, \mu (\d z)\geq m-\frac{\log{R_1}}{\log{R/{R_1}}}(B-m).
$$
Together with \eqref{eq:hlbb1}, this yields
\beq\label{eq:hlbb2}
h(w)\geq \frac{R_1+|w|}{R_1-|w|}\ \frac{m\log{R}-B\log{R_1}}{\log{R/{R_1}}}-\frac{2|w|}{R_1-|w|} \alpha.
\eeq
Based on \eqref{eq:hubb} and \eqref{eq:hlbb2}, we obtain
\begin{align*}
&\min_{c\in \R}\max_{|w|\leq R_2}\ |h(w)-c|\\
&\leq \frac{1}{2} \left( \alpha - \min_{|w|\leq R_2} \left(\frac{R_1+|w|}{R_1-|w|}\ \frac{m\log{R}-B\log{R_1}}{\log{R/{R_1}}}-\frac{2|w|}{R_1-|w|} \alpha \right) \right)\\
&=\frac{1}{2}\max_{|w|\leq R_2}\left(\frac{R_1+|w|}{R_1-|w|}\right)\left(\alpha-\frac{m\log{R}-B\log{R_1}}{\log{R/{R_1}}}\right)\\
&=\frac{1}{2}\left(\frac{R_1+R_2}{R_1-R_2}\right)\left(\alpha-\frac{m\log{R}-B\log{R_1}}{\log{R/{R_1}}}\right)\\
&= B_0(R,R_1,R_2)(B-m).
\end{align*}
This proves Lemma \ref{lm:hest}.
\e{proof}

\be{lm} 
For $k=1,2$, we have
$$
\l|\frac{d^k}{d\vp^k}h(e^{2\pi i \vp})\r|\leq B_k(R,R_1,R_2)(B-m)
$$
with constants $B_1, B_2$ given by \eqref{def:constantsB0B1B2}.
\e{lm}

\be{proof}
Since $h$ is harmonic in $D_{R_1}$, we have that for any constant $c$,
$$
h(e^{2\pi i \vp})-c=\int_0^1 (h(R_2 e^{2\pi i \theta})-c) P_{1/R_2}(\vp-\theta)\d\theta.
$$
We take a derivative in $\vp$ and estimate $h$ using Lemma \ref{lm:hest}. This gives
$$
\begin{aligned}
\l|\frac{d}{d\vp}h(e^{2\pi i \vp})\r|
\leq& \l(\min_{c\in \R} \max_{|w|\leq R_2} |h(w)-c|\r) \int_0^1\l|\frac{\partial}{\partial\vp}P_{1/R_2}(\vp-\theta)\r|\d\theta\\
\leq& B_0(R,R_1,R_2) (B-m)\int_0^1\l|\frac{\partial}{\partial\theta}P_{1/R_2}(\theta)\r|\d\theta.
\end{aligned}
$$
We recall that 
$$
P_{1/R_2}(\theta)=\frac{R_2^2-1}{R_2^2-2R_2 \cos(2\pi\theta)+1}
$$
and therefore
$$
\frac{\partial}{\partial\theta}P_{1/{R_2}}(\theta)
=-\frac{4\pi R_2(R_2^2-1) \sin{(2\pi \theta)}}{(R_2^2-2 R_2\cos{(2\pi \theta)}+1)^2}.
$$
Since $\sin(2\pi\theta)$ changes sign at $\theta=1/2$, we conclude that 
$$
\begin{aligned}
\int_0^1\l|\frac{\partial}{\partial\theta}P_{1/R_2}(\theta-\vp)\r|\d\theta
=& -\int_0^{1/2}\frac{\partial}{\partial\theta}P_{1/R_2}(\theta)\d\theta+\int_{1/2}^1\frac{\partial}{\partial\theta}P_{1/R_2}(\theta)\d\theta\\
=& 2P_{1/R_2}(0)-2P_{1/R_2}(1/2)\\
=& \frac{8R_2}{R_2^2-1}.
\end{aligned}
$$
This proves the claim for $k=1$.

For the second derivative, we argue similarly. We have
$$
\begin{aligned}
\l|\frac{d^2}{d\vp^2}h(e^{2\pi i \vp})\r|
\leq B_0(R,R_1,R_2) (B-m)\int_0^1\l|\frac{\partial^2}{\partial\theta^2}P_{1/R_2}(\theta)\r|\d\theta,
\end{aligned}
$$
where
$$
\begin{aligned}
\frac{\partial^2}{\partial\theta^2}P_{1/{R_2}}(\theta)
=\frac{-8\pi^2 R_2(R_2^2-1)(2R_2\cos^2{(2\pi\theta)}+(R_2^2+1)\cos{(2\pi\theta)}-4R_2)}{(R_2^2-2R_2\cos{(2\pi\theta)}+1)^3}.
\end{aligned}
$$
By symmetry, we may restrict our attention to $\theta\in [0,1/2]$ from now on. On that interval, the function $\frac{\partial^2}{\partial\theta^2}P_{1/{R_2}}$ has exactly one zero. Its location, call it $\theta_0\in [0,1/2]$, is given by
\beq\label{eq:theta0defn}
\theta_0=\frac{1}{2\pi}\arccos\l(\frac{\sqrt{R_2^4+34 R_2^2+1}-(R_2^{2}+1)}{4R_2}\r).
\eeq
It is easy to see that $\frac{\partial^2}{\partial\theta^2}P_{1/{R_2}}$ is negative on $[0,\theta_0)$, and hence positive on $(\theta_0,1/2]$. Therefore
$$
\begin{aligned}
\int_0^1\l|\frac{\partial^2}{\partial\theta^2}P_{1/R_2}(\theta)\r|\d\theta
=&-2\int_0^{\theta_0}\frac{\partial^2}{\partial\theta^2}P_{1/R_2}(\theta)\d\theta
+2\int_{\theta_0}^{1/2}\frac{\partial^2}{\partial\theta^2}P_{1/R_2}(\theta)\d\theta\\
=&2  \frac{\partial}{\partial\theta}P_{1/R_2}(0)-4\frac{\partial}{\partial\theta}P_{1/R_2}(\theta_0)+2  \frac{\partial}{\partial\theta}P_{1/R_2}(1/2)\\
=&\frac{16\pi R_2(R_2^2-1) \sin{(2\pi \theta_0)}}{(R_2^2-2 R_2\cos{(2\pi \theta_0)}+1)^2}.
\end{aligned}
$$
since $\frac{\partial}{\partial\theta}P_{1/R_2}(0)= \frac{\partial}{\partial\theta}P_{1/R_2}(1/2)=0$. 
When we evaluate the last expression using the definition \eqref{eq:theta0defn} of $\theta_0$, we obtain the quantity
$$
\frac{16\pi  (R_2^2-1)\sqrt{16 R_2^2-\l(1+R_2^2-\sqrt{R_2^4+34 R_2^2+1}\r)^2}}{\l(\sqrt{R_2^4+34 R_2^2+1}-3R_2^2-3\r)^2}.
$$
This proves the claim for $k=2$.
\e{proof}

\section{$\T^1$ Splitting Lemma}

For any $f\in L^1(\T)$, let $\la f\ra=\int_{\T}f(x)\, \d x$. For a function $f$ on $\C$, let us denote $f(e(x))$ by $f(x)$ for simplicity.
For a Borel set $U$, let $|U|$ be its Lebesgue measure.

\be{lm}\label{lm:T1Splitting}
Let $v$ be as in Theorem \ref{Riesz}. Assume that for some constant $c$
\beq
v(x)=v_1(x)+v_0(x) +c
\eeq
with  $\|v_1\|_{L^1(\T)}< \varepsilon_1$ and $\|v_0\|_{L^\infty(\T)}<\varepsilon_0$. Then we have
\beq\label{eq:expintegrable}
\begin{aligned}
\int_{\T} &\exp\big( \frac{\pi}{4}\delta_{0}^{-1}|v(x)-c| \big) \, \d x  \leq C_{0} \\
\delta_{0} &:= \frac{9}{2}\,\varepsilon_0+2B_3(R,R_1,R_2) \sqrt{\varepsilon_1 (B-m)} \\
C_{0} &:= 2\sqrt{2}\, \exp\Big(\pi\big [\frac{17}{144}+\frac{B_1}{16B_3^2}  \big]  \Big)
\end{aligned}
\eeq
with constants given by Definition~\ref{def:Poisson}. 
\e{lm}

As a corollary of the exponential integrability, we have the following estimate on the level sets from  Markov's inequality.

\be{cor}\label{cor:T1Splitting}
For any $\varepsilon_2>0$, we have
\beq\nn
|\{x\in \T: \ |v(x)|>\varepsilon_2\}|\leq 2\sqrt{2} \, \exp \Big(\frac{\pi}{4} \big [\frac{17}{144}+\frac{B_1}{16B_3^2}-\varepsilon_{2}\delta_{0}^{-1}   \big]  \Big)
\eeq
with $\delta_0$ as in \eqref{eq:expintegrable}. 
\e{cor}

Note that this level set estimate is only useful if $\varepsilon_{2}\gg \varepsilon_{0}$ and $\varepsilon_{2}^{2}\gg \varepsilon_1 (B-m)$.

\subsection*{Proof of Lemma \ref{lm:T1Splitting}}
For simplicity, we will denote $B_3(R,R_1,R_2)$ by $B_3$ throughout the proof.
We will first show the following special form of the Riesz representation,  valid only on the unit circle. The idea is simply to reflect the part of the disk outside the circle back inside of it. 

\be{lm}\label{lm:T1Riesz}
Let $1<R_2<R_1<R$ and $v$ be defined as in Theorem \ref{Riesz}. Then
there exists a positive measure $\tilde{\mu}$ and a harmonic function $\tilde{h}$ on $D_R$ such that
\beq\label{eq:T1Riesz}
v(e(\varphi))=\int_{\overline{D_1}}\log{|z-e(\varphi)|}\, \tilde{\mu}(\d z)+\tilde{h}(e(\varphi)),
\eeq
with the following estimates
\beq\label{est:tildemulemma}
\tilde{\mu}(\overline{D_1})\leq \frac{B-m}{\log{R/{R_1}}}.
\eeq
and $\tilde{h}$ satisfies the bound Lemma~\ref{lm:hest} on the circle as well as 
\beq\label{est:tildehlemma}
\l|\frac{d^k}{d \varphi^k}\tilde{h}(e(\varphi))\r|\leq B_k(R,R_1,R_2)(B-m),\ \ k=1,2,
\eeq
where $B_1,B_2$ are the same constants as those in Theorem \ref{Riesz}.
\e{lm}
\be{proof}
By Theorem \ref{Riesz}, we have
\beq\label{eq:reflec1}
v(w)=\int_{D_{R_1}}\log{|z-w|}\, \mu(\d z)+h(w),
\eeq
with $\mu(D_{R_1})\leq \frac{B-m}{\log{R/{R_1}}}$, and $\l|\frac{d^k}{d \varphi^k}h(e(\varphi))\r|\leq B_k(R,R_1,R_2)(B-m)$, $k=1,2$.

Let us define $\mu^*$ by reflection, i.e., 
\beq
\mu^*(E)=\mu (E^*),
\eeq
where
$$E^*=\{\overline{z^{-1}}:\ z\in E\}$$
for any measurable set $E\subset \C$. 
Then for any $|w|=1$,
\begin{align}\label{eq:tildemu}
  &\int_{D_{R_1}}\log{|z-w|}\, \mu(\d z)  \notag\\
=&\int_{\overline{D_1}}\log{|z-w|}\, \mu(\d z)+\int_{D_{R_1}\setminus \overline{D_1}}\log{|z-w|}\, \mu (\d z) \notag\\
=&\int_{\overline{D_1}}\log{|z-w|}\, \mu(\d z)+\int_{D_{R_1}\setminus \overline{D_1}}\log{|w-\overline{z^{-1}}|}\, \mu (\d z)+ \int_{D_{R_1}\setminus \overline{D_1}}\log{|z|}\, \mu (\d z) \notag \\
=&\int_{\overline{D_1}}\log{|z-w|}\, \mu(\d z)+\int_{D_1\setminus \overline{D_{1/R_1}}}\log{|w-z|}\, \mu^* (\d z)+ \int_{D_{R_1}\setminus \overline{D_1}}\log{|z|}\, \mu (\d z) \notag\\
=&\int_{\overline{D_1}}\log{|z-w|}\, \tilde{\mu}(\d z)+ \int_{D_{R_1}\setminus \overline{D_1}}\log{|z|}\, \mu (\d z),
\end{align}
where, for any $E\subset \overline{D_1}$,
\begin{align}\label{eq:mutildemu}
\tilde{\mu}(E)=\mu(E)+\mu^*(E\cap (D_1\setminus \overline{D_{1/R_1}}))=\mu(E)+\mu(E^*\cap (D_{R_1}\setminus \overline{D_1})).
\end{align}
By \eqref{eq:mutildemu}, it is clear that we have the following estimate for $\tilde{\mu}$
\begin{align}\label{eq:tildemuest}
\tilde{\mu}(\overline{D_1})= \mu(D_{R_1})\leq \frac{B-m}{\log(R/R_1)}.
\end{align}

By \eqref{eq:reflec1} and \eqref{eq:tildemu}, we have for $|w|=1$,
\begin{align}
v(w)=\int_{\overline{D_1}}\log{|z-w|}\, \tilde{\mu}(\d z)+h(w)+\int_{D_{R_1}\setminus \overline{D_1}}\log{|z|}\, \mu (\d z),
\end{align}
in which the third term is a constant. 
Let us take $\tilde{h}=h+\int_{D_{R_1}\setminus \overline{D_1}}\log{|z|}\, \mu (\d z)$. 
Since $\tilde{h}$ only differs from $h$ by a constant, the estimates on the derivatives still hold.
\e{proof}

The Riesz representation~\eqref{eq:T1Riesz} allows us to give upper bounds on the parameters $\varepsilon_{0}$ and $\varepsilon_{1}$ in Lemma~\ref{lm:T1Splitting} in terms of~$B-m$.  This will be relevant in the proof of that lemma. 

\begin{cor}
We may always assume in Lemma~\ref{lm:T1Splitting} that 
\beq\label{eq:eps01upper}
\varepsilon_{0} \le B_{0}(R,R_{1},R_{2})(B-m), \quad  \varepsilon_{1}\le \frac{13}{20} \frac{B-m}{\log(R/R_1)}
\eeq
Alternatively, we can assume that $\varepsilon_{0}=0$ and 
\beq\label{eq:eps1set}
\varepsilon_{1} = \Big(B_{0}(R,R_{1},R_{2})  +  \frac{13}{20\log(R/R_1)}    \Big)(B-m)
\eeq
\end{cor}
\begin{proof}
In view of \eqref{eq:T1Riesz} we set 
\[
v_{0}(\varphi):=\tilde{h}(e(\varphi)),\qquad   v_{1}(\varphi)=\int_{\overline{D_1}}\log{|z-e(\varphi)|}\, \tilde{\mu}(\d z)
\]
Then $\varepsilon_{0}$ is the constant from  Lemma~\ref{lm:hest} and we claim that 
\[
\varepsilon_{1}:= \| \log|1-e(\varphi)| \|_{L^{1}_\varphi} \| \tilde\mu\| 
\]
is an admissible choice. Indeed, 
\begin{align*}
\| v_{1}\|_{L^{1}(\T)} &\le \int_{\overline{D_1}} \big\| \log{|z-e(\varphi)|} \big\|_{L^{1}(\T_{\varphi})} \tilde{\mu}(\d z) \\
 & =  \int_{\overline{D_1}} \big\| \log{||z| -e(\varphi)|} \big\|_{L^{1}(\T_{\varphi})} \tilde{\mu}(\d z) \\
 &\le \max_{0\le r\le 1} \big\| \log{|r -e(\varphi)|} \big\|_{L^{1}(\T_{\varphi})}\, \| \tilde\mu\|
\end{align*}
Set $h(r) :=  \big\| \log{|r -e(\varphi)|} \big\|_{L^{1}(\T_{\varphi})}$ with $0\le r\le1$. In order to establish the claim, it suffices to verify that $h(r)$ is nondecreasing. 
First,
\begin{align*}
\int_{0}^{1} \log{|r -e(\varphi)|} \, d\varphi & = \int_{0}^{1} \log{|1 -re(\varphi)|} \, d\varphi =0
\end{align*}
since $\log{|1 -r\zeta|}$ is harmonic in $\zeta$ for $|\zeta|<1$ and any fixed $0\le r\le 1$. Therefore, if $0<r<1$ and $0<\varphi_{0}(r)<\frac12$ is the unique solution of $|r -e(\varphi_{0})| =1$, then    
\[
h(r) = 2\int_{\varphi_{0}(r)}^{1-\varphi_{0}(r)}  \log{|r -e(\varphi)|} \, d\varphi = \int_{\varphi_{0}(r)}^{1-\varphi_{0}(r)}  \log(1+r^{2} -2r\cos(2\pi\varphi))\, d\varphi
\]
Consequently, 
\begin{align*}
h'(r) &= \int_{\varphi_{0}(r)}^{1-\varphi_{0}(r)}  \frac{2(r-\cos(2\pi\varphi))}{1+r^{2} -2r\cos(2\pi\varphi)}\, d\varphi  \\
&\ge \int_{\varphi_{0}(r)}^{1-\varphi_{0}(r)}  \frac{r}{1+r^{2} -2r\cos(2\pi\varphi)}\, d\varphi  \ge 0
\end{align*}
In the second line we used that on the domain of integration $$|r -e(\varphi)|^{2}=1+r^{2}-2r\cos(2\pi\varphi) \ge 1$$ whence $2r-2\cos(2\pi\varphi) \ge r $.
Therefore, indeed $h(r)\le h(1)$,  justifying our choice of~$\varepsilon_{1}$ above. 
Finally, 
\begin{align*}
h(1) &= \| \log|1-e(\varphi)| \|_{L^{1}_\varphi} = -2 \int_{0}^{1} \min( \log|1-e(\varphi)|,0)\, d\varphi\\
&= -2 \int_{-\frac16}^{\frac16}  \log|1-e(\varphi)| \, d\varphi = -4 \int_{0}^{\frac16} \log(2\sin(\pi \varphi))\, d\varphi < 13/20. 
\end{align*}
and $\| \tilde\mu\| $ is controlled by~\eqref{est:tildemulemma}. 
\end{proof}

\begin{defn}  \label{def:epssmall}
Henceforth we impose the condition that 
\beq\label{eq:epsklein}
289\big(B_{0}(R,R_{1},R_{2})+\frac{13}{20\log R/R_{1}}\big) < B_3^{2}(R,R_{1},R_{2})
\eeq
where the constants are those from Definition~\ref{def:Poisson}. 
\end{defn}

Returning to the proof of Lemma \ref{lm:T1Splitting}, we denote the first term in \eqref{eq:T1Riesz} by $u$, viz. 
\beq\label{def:u}
u(x)=\int_{\overline{D_1}}\log{|z-e(x)|}\, \tilde{\mu}(\d z).
\eeq
Then $v=u+\tilde{h}$.
For any $f\in L^1(\T)$ with $\la f\ra=0$, the anti-derivative $D^{-1}f$ is uniquely defined as the absolutely continuous function 
\beq\label{def:anti-de}
(D^{-1}f)(t)=\int_{0}^t f(x)\, \d x + m(f),\quad \la D^{-1}f\ra=0,
\eeq
for arbitrary $t\in \T$. The constant $m(f)$ is chosen  to ensure the vanishing mean.  For  $e_{n}(x)=\exp(2\pi in x)$, one has   $D^{-1}(e_{n})=(2\pi in)^{-1}e_{n}$ for all $n\ne0$, whereas $D^{-1}e_{0}=0$.   In the distributional sense, $D^{-1}$ also applies to (complex) measures. For example, with $\delta_{0}$ now being the Dirac delta, 
$$D^{-1}(\delta_{0}-1)(x)=-(x+\frac12) \one_{[-\frac12<x<0]}(x) + (\frac12-x) \one_{[0<x<\frac12]}(x).$$ 
For any $z=|z|e(y)\in D_1$, one has the elementary relation
\beq
\begin{aligned}
\frac{d}{dx}\log{|e(x)-z|}
=&\frac{2\pi |z|\sin{(2\pi(x-y))}}{1-2|z|\cos{(2\pi(x-y))}+|z|^2}\\
=&\pi Q_{z}(x)\\
=&\pi (\mathcal{H}[P_{z}])(x),
\end{aligned}
\eeq
where $\mathcal{H}$ denotes the Hilbert transform and $Q_z$ is the standard notation for the conjugate function of the Poisson kernel. 
In particular,
\beq
\log{|e(x)-z|}=\pi (D^{-1}\mathcal{H}[P_z])(x)
\eeq
holds for any $z\in D_1$.
We thus have
\beq \label{eq:uHnu}
u(x)=(D^{-1}\mathcal{H}[\nu])(x),
\eeq
where 
$$
\frac{d \nu}{d x}(x) =\pi \int_{\overline{D_1}}P_z(x)\, \tilde{\mu}(\d z)
$$
is a positive measure, with 
\beq\label{est:nu}
\nu(\T)=\pi \tilde{\mu}(\overline{D_1}).
\eeq
Set 
\beq
\label{eq:epsviaeps1}
\epsilon:=\frac{1}{B_3} \sqrt{\frac{\varepsilon_1}{B-m}},
\eeq
 and define $J_{\epsilon}(x):=\frac{1}{2\epsilon}\one_{[-\epsilon, \epsilon]}(x)$ to be the box kernel. 
 Because of the upper bound in~\eqref{eq:eps1set} on $\varepsilon_{1}$ and~\eqref{eq:epsklein}, one has
 \[
 \varepsilon_{1} \le \big(B_{0}(R,R_{1},R_{2})+\frac{13}{20\log R/R_{1}} \big)(B-m) < \frac{B_3^2}{289}(B-m) 
 \]
 which ensures that $\epsilon<\frac{1}{17}$.  We will use this smallness property for the remainder of the proof. For example, it 
  guarantees that $J_{\epsilon}$ is in fact well-defined on the circle (less than $\frac12$ is enough here, but below we will need this sharper bound). 
 Then 
\begin{align*}
{v}
=&  {v}-J_{\epsilon}*{v}+J_{\epsilon}*v_1+J_{\epsilon}*v_0\notag\\
=& ({u}-J_{\epsilon}*{u})+(\tilde{h}-J_{\epsilon}*\tilde{h})+J_{\epsilon}*v_1+J_{\epsilon}*v_0.
\end{align*}
The last three terms have small $L^\infty$ norms, in the sense that
\begin{align*}
\left\lbrace
\begin{matrix}
&\|J_{\epsilon}*v_1\|_{L^\infty(\T)}\leq \|J_{\epsilon}\|_{L^\infty(\T)} \|v_1\|_{L^1(\T)}<\frac{\varepsilon_1}{2\epsilon}=\frac{B_3}{2}\sqrt{\varepsilon_1(B-m)},\\
\\
&\|J_{\epsilon}*v_0\|_{L^\infty(\T)}\leq \|J_{\epsilon}\|_{L^1(\T)} \|v_0\|_{L^\infty(\T)}<\varepsilon_0,\\
\\
&\|\tilde{h}-J_{\epsilon}*\tilde{h}\|_{L^\infty(\T)}\leq \frac{\epsilon}{2} \|\tilde{h}^\prime\|_{L^\infty(\T)}\leq \frac{1}{2} \epsilon B_1(B-m)=\frac{B_1}{2 B_3}\sqrt{\varepsilon_1(B-m)}.
\end{matrix}
\right.
\end{align*}
Hence 
\beq\label{eq:comparevu}
|({u}-J_{\epsilon}*{u})(x)|\geq |v(x)-C|-\varepsilon_0-\left(\frac{B_3}{2}+\frac{B_1}{2B_3}\right)\sqrt{\varepsilon_1(B-m)}.
\eeq
By \eqref{eq:uHnu}, we have
\begin{align}\label{eq:uHnu2}
(u-J_{\epsilon}*u)(x)
=&(D^{-1}\mathcal{H}[\nu-J_{\epsilon}*\nu])(x) \notag\\
=&\mathcal{H}[D^{-1}(\nu-J_{\epsilon}*\nu)](x).
\end{align}
Next, we control the pointwise size of the term in brackets in \eqref{eq:uHnu2}. Since the Hilbert transform eliminates constants, the integration constant in~\eqref{def:anti-de} drops out. 

\be{lm}\label{lm:LinftynorminsideH}
Modulo additive constants the function $D^{-1}(\nu-J_{\epsilon}*\nu)$ satisfies 
\beq 
\|D^{-1}(\nu-J_{\epsilon}*\nu)\|_{L^\infty(\T)}\leq \frac{9}{2}\varepsilon_0+2B_3\sqrt{\varepsilon_1(B-m)}
\eeq
\e{lm}
\be{proof}
We being with the observation that  (recall $\nu$ is a positive measure)
\begin{align}\label{eq:inftynorm1}
|(\nu-J_{\epsilon}*\nu)([a,b])|
=&|\int_{\T}(\one_{[a,b]}-\one_{[a,b]}*J_{\epsilon})(x) \, \nu (\d x)| \notag\\
\leq &  \sup_{\theta\in \T}\nu([\theta-\epsilon, \theta+\epsilon]),
\end{align}
uniformly in $[a,b]\subset\T$. If $b-a\ge 2\epsilon$, then 
on the one hand, 
\[
 |(\one_{[a,b]}-\one_{[a,b]}*J_{\epsilon})(x) | \le \frac12(\one_{[a-\epsilon,a+\epsilon]}+\one_{[b-\epsilon,b+\epsilon]})(x) 
\]
 On the other hand, if $b-a< 2\epsilon$ then by translation invariance it suffices to consider the symmetric expression
 \[
f(x):= \one_{[-d,d]}-\one_{[-d,d]}*J_{\epsilon}, \qquad 2d=b-a,\; d<\epsilon
 \]
 For any $0<d<\epsilon/2$ this function satisfies
 \begin{align*}
 |f| & \le \frac{d}{\epsilon} \one_{(-\epsilon-d,\epsilon+d)}+(1-2d/\epsilon)\one_{[-d,d]} \\
 &\le \frac{d}{\epsilon} \one_{(-\epsilon-d,\epsilon-d]}  + \frac{d}{\epsilon} \one_{(\epsilon-d,2\epsilon-d)}  +(1-2d/\epsilon)\one_{(-\epsilon,\epsilon)}
 \end{align*}
 whereas for $\epsilon/2<d<\epsilon$ one has
 \[
 |f| \le \frac12 \one_{(-\epsilon-d,\epsilon+d)} \le \frac12 \one_{(-\epsilon-d,\epsilon-d]}  + \frac12 \one_{(\epsilon-d,3\epsilon-d)}
 \]
In either case \eqref{eq:inftynorm1} holds. 

It therefore suffices to estimate $\sup_{\theta\in \T}\nu([\theta-\epsilon, \theta+\epsilon])$.
Next, we define an atom $\tau^\prime$ in the Hardy space $H^{1}(\T)$ as follows: 
\begin{align}\label{eq:tau'def}
\tau^\prime(x)=
\left\lbrace
\begin{matrix}
(x-(a-3\epsilon))/{\epsilon^2},\ \ \ a-3\epsilon\leq x\leq a-2\epsilon,\\
((a-\epsilon)-x)/{\epsilon^2},\ \ \ a-2\epsilon\leq x\leq a-\epsilon,\\
((a+\epsilon)-x)/{\epsilon^2},\ \ \ a+\epsilon\leq x\leq a+2\epsilon,\\
(x-(a+3\epsilon))/{\epsilon^2},\ \ \ a+2\epsilon\leq x\leq a+3\epsilon,\\
0,\ \ \ \text{otherwise}.
\end{matrix}
\right.
\end{align}
Note that this is well-defined on the circle since $\epsilon<\frac16$. 
By construction, $\la \tau^\prime\ra=0$. Set $\tau(x):=\int_{a-\frac{1}{2}}^x \tau^\prime(t)\,\d t$.
Moreover, $\tau\geq 0$, $\la \tau\ra=4\epsilon$, and $\tau(x)=1$ on $[a-\epsilon, a+\epsilon]$.
Thus 
\beq\label{eq:tau nu}
\nu([a-\epsilon, a+\epsilon])\leq \int_{\T}\tau(x) \, \nu(\d x) = \langle \tau,\nu\rangle.
\eeq
Let us consider
\beq\label{eq:inftynorm2}
\begin{aligned}
  &|(\tau-\la \tau\ra, \nu)|\\
=&|(\frac{d}{dx}\mathcal{H}[\tau], D^{-1}\mathcal{H}[\nu])|=|(\mathcal{H}[\tau^\prime], {u})|\\
 = &|(\mathcal{H}[\tau^\prime], v-\tilde{h} )|=  |(\mathcal{H}[\tau^\prime], v_{0}+v_{1}-\tilde{h} )| \\
\leq &|(\mathcal{H}[\tau^\prime], v_0)|+|(\mathcal{H}[\tau^\prime], v_1)|+|(\tau, \mathcal{H}[\tilde{h}^\prime])|\\
\leq &\|\mathcal{H}[\tau^\prime]\|_{L^1(\T)}\|v_0\|_{L^\infty(\T)}+\|\mathcal{H}[\tau^\prime]\|_{L^\infty(\T)}\|v_1\|_{L^1(\T)}+\|\tau\|_{L^1(\T)}
\|\mathcal{H}[\tilde{h}^\prime]\|_{L^\infty(\T)}\\
\leq &\varepsilon_0 \|\mathcal{H}[\tau^\prime]\|_{L^1(\T)}+\varepsilon_1 \|\mathcal{H}[\tau^\prime]\|_{L^\infty(\T)}+2\epsilon \|\frac{d^2}{dx^2}\tilde{h}\|_{L^\infty(\T)}
\end{aligned}
\eeq
In the last line, we used the following lemma on the third term. 

\be{lm}
For any $f\in C^{1}(\T)$ one has  $\| \mathcal{H}[f]\|_{\infty}\le \frac12 \|f'\|_{\infty}$.
\e{lm} 
\be{proof}
Since $\sin(\pi x)\ge 2x$ for all $0< x\le\frac{1}{2}$, one has 
\begin{align*}
 \| \mathcal{H}[f]\|_{\infty} &= \sup_{y\in \T} \Big| \int_{\T} \frac{f(x)-f(y)}{\sin(\pi(x-y))} \cos(\pi(x-y))\, dx\Big| \le \frac12 \|f'\|_{\infty}
\end{align*}
as claimed. 
\e{proof}

In order to bound the the other terms in the last line of~\eqref{eq:inftynorm2} we prove two lemmas.

\be{lm}\label{lm:L1norm}
Let $\tau'$  be defined by \eqref{eq:tau'def} and assume that $0<\eps\leq \frac{1}{17}$. Then we have
\beq
\|\curly{H}[\tau']\|_{L^1(\mathbb T)}\leq \frac{9}{2}.
\eeq
\e{lm}

\be{rmk}
The upper bound $\frac{1}{17}$ is a particular choice which we have found to be convenient in the last section of the paper. A more restrictive assumption on $\eps$ will slightly improve the bound; e.g., assuming $\eps<\frac{1}{36}$ yields the value $4.2$ for $\frac{9}{2}$. Such improvements are mainly due to the lower bound in \eqref{eq:sinslope} approaching $\pi$ as $\eps\to0$.
\e{rmk}

\emph{Proof.} By translation symmetry, we may assume $a=0$. We let $0<\eps\leq \frac{1}{17}$ and $I=[-r\eps,r\eps]$ with $r=3.6>3$. Notice that $r\eps<\frac{1}{2}$. We decompose $\|\curly{H}[\tau']\|_{L^1(\mathbb T)}$ into the two parts
$$
\|\curly{H}[\tau']\|_{L^1(\mathbb T)}=\|\curly{H}[\tau']\|_{L^1(I)}+\|\curly{H}[\tau']\|_{L^1(I^c)}.
$$
By Cauchy-Schwarz and the fact that the Hilbert transform is an isometry on $L^2(\mathbb T)$, we have
$$
\|\curly{H}[\tau']\|_{L^1(I)}\leq 
\sqrt{|I|}\|\curly{H}[\tau']\|_{L^2(I)} 
\leq \sqrt{|I|} \|\tau'\|_{L^2(I)} 
=\frac{2\sqrt{2}}{\sqrt{3}} \sqrt{r}.
$$
In the remainder of the proof, we bound $\|\curly{H}[\tau']\|_{L^1(I^c)}$. By symmetry, we have $\|\curly{H}[\tau']\|_{L^1(I^c)}=2\|\curly{H}[\tau']\|_{L^1([r\eps,\frac{1}{2}])}$. Hence, it suffices to consider the interval $[r\eps, \frac{1}{2}]$. For all $x\in [r\eps, 1/2]$, we have
\beq\label{eq:Htau'1}
\begin{aligned}
|\curly{H}[\tau'](x)|
=& \l|\int_{\supp\tau'} \tau'(y) \cot(\pi(x-y))\d y\r|\\
=& \l|\int_{\supp\tau'} \tau'(y) \l(\cot(\pi(x-y))-\cot(\pi x)\r)\d y\r|\\
=& \l|\int_{\supp\tau'} \tau'(y) \frac{\sin(\pi y)}{\sin(\pi(x-y))\sin(\pi x)}\d y\r|\\
\leq& \int_{\supp\tau'} |\tau'(y)| \frac{|\sin(\pi y)|}{\sin(\pi(x-y))\sin(\pi x)}\d y.
\end{aligned}
\eeq
Here we used that $x-y\in [(r-3)\eps,\frac{1}{2}+3\eps]\subset [0,1]$. To estimate this expression further, we decompose $\supp\tau'$ into the four intervals
$$
I_1:=[2\eps,3\eps],\qquad I_2:=[\eps,2\eps],\qquad I_3:=[-2\eps,-\eps],\qquad I_4:=[-3\eps,-2\eps].
$$
We write $\rho_j(x)\in I_j$ for the point in $I_j$ that is nearest to $x$ in the toroidal distance, i.e., $\|x-\rho_j(x)\|_{\mathbb T}=\mathrm{dist}_{\mathbb T}(x,I_j)$. (That point is not unique if $x$ is ``antipodal'' to the center of $I_j$; in this case we define $\rho_j$ as the right endpoint of $I_j$ for definiteness.) Notice that $\rho_j(x)$ is constant for $j\in \{1,2\}$ with 
\beq\label{eq:specifically0}
\rho_1(x)=\rho_{1,1}:=3\eps,\qquad \rho_2(x)=\rho_{2,1}:=2\epsilon,
\eeq
and piecewise constant for $j\in \{3,4\}$, i.e.,
$$
\rho_j(x)=\be{cases}
\rho_{j,1},\qquad \textnormal{ if } x\in[0,t_j],\\
\rho_{j,2},\qquad \textnormal{ if } x\in\l(t_j,\frac{1}{2}\r].
\e{cases}
$$
Specifically, we have
\beq\label{eq:specifically}
\begin{aligned}
\rho_3(x)=&
\be{cases}
-\eps,\;\,\qquad \textnormal{ if } x\in\l[0,\frac{1}{2}-\frac{3\eps}{2}\r],\\
-2\eps,\qquad \textnormal{ if } x\in\l(\frac{1}{2}-\frac{3\eps}{2},\frac{1}{2}\r],
\e{cases}\\
\rho_4(x)=&
\be{cases}
-2\eps,\qquad \textnormal{ if } x\in\l[0,\frac{1}{2}-\frac{5\eps}{2}\r],\\
-3\eps,\qquad \textnormal{ if } x\in\l(\frac{1}{2}-\frac{5\eps}{2},\frac{1}{2}\r],
\e{cases}
\end{aligned}
\eeq
From \eqref{eq:Htau'1} and $|\sin (\pi y)|\leq \pi |y|$, we obtain
$$
|\curly{H}[\tau'](x)|\leq \pi \sum_{j=1}^4 \frac{1}{\sin(\pi(x-\rho_j(x)))\sin(\pi x)} \int_{I_j}  |\tau'(y)| |y|\d y.
$$
We integrate both sides over $x\in [r\eps,\frac{1}{2}]$ and find
\beq\label{eq:Htau'2}
\begin{aligned}
\|\curly{H}[\tau']\|_{L^1([r\eps,\frac{1}{2}])}
\leq \pi \sum_{j=1}^4 f_j(\eps) \int_{I_j}  |\tau'(y)| |y|\d y,
\end{aligned}
\eeq
where we introduced the notation
$$
f_j(\eps):=\int_{r\eps}^{\frac{1}{2}} \frac{1}{\sin(\pi(x-\rho_j(x)))\sin(\pi x)}\d x,
$$
for $j\in \{1,2,3,4\}$. The following lemma gives a bound on these integrals.

\be{lm}\label{lm:j1234}
For $j\in \{1,2\}$, we have
\beq\label{eq:j12}
f_j(\eps)\leq \frac{1}{2.98\pi \rho_{j,1}} \log\l(1+\frac{\pi}{2.98}\,\frac{ \rho_{j,1}}{r\eps-\rho_{j,1}}\r).
\eeq
and for $j\in \{3,4\}$, we have
\beq\label{eq:j34}
f_j(\eps)\leq  \frac{2}{5\pi |\rho_{j,1}|} \log\l(1+\frac{2\pi}{5}\,\frac{|\rho_{j,1}|}{r\eps}\r)+0.21.
\eeq
\e{lm}

We postpone the proof of Lemma \ref{lm:j1234} for now. To continue the proof of Lemma \ref{lm:L1norm}, recall \eqref{eq:Htau'2}. We perform the integration in $y$ and find
\beq\label{eq:yintegral}
\begin{aligned}
 \int_{I_1}  |\tau'(y)| |y|\d y=& \int_{I_4}  |\tau'(y)| |y|\d y=\frac{7\eps}{6},\\
 \qquad \int_{I_2}  |\tau'(y)| |y|\d y=& \int_{I_3}  |\tau'(y)| |y|\d y=\frac{5 \eps}{6}.
 \end{aligned}
\eeq
Then we apply Lemma \ref{lm:j1234} and recall Definitions \eqref{eq:specifically0} and \eqref{eq:specifically} of $\rho_{j,1}$. This gives
\beq\label{eq:Htau'3}
\begin{aligned}
&\|\curly{H}[\tau']\|_{L^1([r\eps,\frac{1}{2}])}
\leq\pi \sum_{j=1}^4f_j(\eps) |\tau'(y)| |y|\d y\\
\leq& \frac{1}{17.88} \l(\frac{7}{3}\log\l(1+\frac{\pi}{2.98}\frac{3}{r-3}\r)
+\frac{5}{2}\log\l(1+\frac{\pi}{2.98}\frac{2}{r-2}\r)\r)\\
&
 +\frac{1}{15} \l(5\log\l(1+\frac{2\pi}{5}\frac{1}{r}\r)
+\frac{7}{2}\log\l(1+\frac{2\pi}{5}\frac{2}{r}\r)\r)
+(0.42)\pi \eps.
 \end{aligned}
\eeq
Notice that $(0.42)\pi \eps <0.08$. We write $\nu(r)$ for the expression in the last line. Altogether, we have shown that 
$$
\|\curly{H}[\tau']\|_{L^1(\mathbb T)}
\leq \frac{2\sqrt{2}}{\sqrt{3}} \sqrt{r} +2\nu(r)<4.5.
$$
In the second step, we evaluated the expression at $r=3.6$. This proves the main claim of Lemma \ref{lm:L1norm}. It remains to give the

\be{proof}[Proof of Lemma \ref{lm:j1234}]
We begin by observing that
\beq\label{eq:antideriv}
\frac{\d}{\d x}  \l(\frac{1}{ \sin x'}\log\l(\frac{\sin (x-x')}{\sin(x)}\r)\r) =\frac{1}{\sin x \sin(x-x')},
\eeq
whenever $x,x'$ are such that the logarithm is well-defined. 

Let $j\in\{1,2\}$, which implies that $\rho_{j}(x)=\rho_{j,1}>0$. By \eqref{eq:antideriv} and $\sin,\cos\leq 1$, we have
\beq\label{eq:j12pf}
\begin{aligned}
f_j(\eps)=&\frac{1}{\pi \sin(\pi\rho_{j,1})}\l[\log\l(\frac{\sin(\pi(x-\rho_{j,1}))}{\sin(\pi x)}\r)\r]_{x=r\eps}^{x=\frac{1}{2}}\\
\leq &\frac{1}{\pi \sin(\pi\rho_{j,1})} \log\l(\frac{\sin(\pi r\eps)}{\sin(\pi(r\eps-\rho_{j,1}))}\r)\\
\leq &\frac{1}{\pi \sin(\pi\rho_{j,1})} \log\l(1+\frac{\sin(\pi \rho_{j,1})}{\sin(\pi(r\eps-\rho_{j,1}))}\r)
\end{aligned}
\eeq
Next, we estimate the $\sin$'s by linear functions. While the upper bound $\sin (\pi x)\leq \pi x$ is valid for all $x$ (and is sharp for small $x$), a linear lower bound on $\sin (\pi x)$ depends directly on the allowed range of $x$ values. This is conveniently expressed via the quotient
$$
\inf_{x\in [0,3\eps]} \frac{\sin(\pi x)}{x}=\frac{\sin(3\pi \eps)}{3\eps}>2.98
$$
In the last step, we used that $\eps<\frac{1}{17}$. We may verify that all the arguments of $\sin(\pi\cdot)$ in the last line of \eqref{eq:j12pf} are located in the interval $[0,3\eps]$. Therefore
$$
\begin{aligned}
f_j(\eps)
\leq &\frac{1}{\pi \sin(\pi\rho_{j,1})} \log\l(1+\frac{\sin(\pi \rho_{j,1})}{\sin(\pi(r\eps-\rho_{j,1}))}\r)\\
\leq &\frac{1}{2.98\pi \rho_{j,1}} \log\l(1+\frac{\pi}{2.98}\,\frac{ \rho_{j,1}}{r\eps-\rho_{j,1}}\r).
\end{aligned}
$$
This proves \eqref{eq:j12}.\\

Next, let $j\in \{3,4\}$, so that $\rho_{j,1},\rho_{j,2}<0$. We have $r\eps<t_j$ by our assumptions on $r,\eps$ and \eqref{eq:antideriv} yields
$$
\begin{aligned}
f_j(\eps)\leq &\frac{1}{\pi \sin(\pi\rho_{j,1})} \l[\log\l(\frac{\sin(\pi(x-\rho_{j,1}))}{\sin(\pi x)}\r)\r]_{x=r\eps}^{x=t_j}\\
&+\l(\frac{1}{2}-t_j\r) \max_{x\in(t_j,\frac{1}{2}]} \frac{1}{\sin(\pi(x-\rho_{j,2}))\sin(\pi x)}.
\end{aligned}
$$
The second term is an error term (it vanishes as $\eps\to 0$). Indeed, recalling the definition of $t_j$ and $\rho_{j,2}$ from \eqref{eq:specifically}, we see that for $j\in\{3,4\}$, 
$$
\l(\frac{1}{2}-t_j\r) \frac{1}{\sin(\pi(\frac{1}{2}-\rho_{j,2}))\sin(\pi t_j)}
\leq \frac{5\eps}{2\cos^2(3\pi\eps)}\leq \frac{5\eps}{2(1-\frac{9}{2}\pi^2\eps^2)^2}\leq 0.21,
$$
where the last estimate used that $\eps\leq \frac{1}{17}$. Therefore, we have
$$
\begin{aligned}
&\int_{r\eps}^{\frac{1}{2}} \frac{1}{\sin(\pi(x-\rho_j(x)))\sin(\pi x)}\d x\\
\leq& \frac{1}{\pi \sin(\pi\rho_{j,1})} \l[\log\l(\frac{\sin(\pi(x-\rho_{j,1}))}{\sin(\pi x)}\r)\r]_{x=r\eps}^{x=t_j}+0.21\\
\leq& \frac{1}{\pi \sin(\pi|\rho_{j,1}|)} \log\l(\frac{\sin(\pi(r\eps-\rho_{j,1}))}{\sin(\pi r\eps)}\r)+0.21\\
\leq& \frac{2}{5\pi |\rho_{j,1}|} \log\l(1+\frac{2\pi}{5}\,\frac{|\rho_{j,1}|}{r\eps}\r)+0.21.
\end{aligned}
$$
In the second step, we used that $0\leq t_j<t_j-\rho_{j,1}\leq \frac{1}{2}$ and monotonicity properties of $\sin$. In the last step, we used $\cos\leq 1$ and 
\beq\label{eq:sinslope}
\inf_{x\in [0,(r+2)\eps]} \frac{\sin(\pi x)}{x}=\frac{\sin(\pi(r+2)\eps)}{(r+2)\eps}>\frac{5}{2}.
\eeq
This bound holds because $(r+2)\eps<\frac{1}{3}$, which may be verified from $r=3.6$ and $\eps\leq \frac{1}{17}$. This shows \eqref{eq:j34} and concludes the proof of Lemma \ref{lm:j1234}, and hence also of Lemma \ref{lm:L1norm}.
\e{proof}

Next, we control the pointwise size of $\mathcal{H}[\tau^\prime]$.

\be{lm}\label{lm:Linftynorm}
We have $$\|\mathcal{H}[\tau^\prime]\|_{L^\infty(\T)}\le \frac{5}{2\epsilon}$$
\e{lm}
\be{proof}
By translation invariance, we can set $a=0$. 
Let us first consider $$\int_{\mathrm{supp}(\tau^\prime)}\tau^\prime(x)\cot{\pi(x-y)}\, \d y.$$
If $x\notin \mathrm{supp}(\tau^\prime)$, then $\int_{\mathrm{supp}(\tau^\prime)}\tau^\prime(x)\cot{\pi(x-y)}\, \d y=0$. 
Thus, we can assume without loss of generality that $x\in [-3\epsilon, -\epsilon]$. 
On the one hand, 
\begin{align*}
\Big |\int_{[\epsilon, 3\epsilon]}\tau^\prime(x)\cot{\pi(x-y)}\, \d y \Big | &\leq 2\epsilon\,  |\tau'(x)| \sup_{y\in[\epsilon, 3\epsilon] }|\cot \pi (x-y)|\\
&\le 2 \sup_{u\in [2\epsilon,6\epsilon]} |\cot(\pi u)| \le 2 \cot{(2\pi \epsilon)}\leq \frac{1}{2\epsilon},
\end{align*}
since $6\epsilon \pi \le \pi - 2\pi \epsilon$ and $\sin(2\pi \epsilon)\ge 4\epsilon$ 
(recall that $\epsilon <\frac18$). 
On the other hand, for the negative support of $\tau'$, we can further assume by symmetry that $x\in [-3\epsilon, -2\epsilon]$. Thus, 
\beq 
\begin{aligned}
  &|\int_{[-3\epsilon, -\epsilon]}\tau^\prime(x)\cot{\pi(x-y)} \, \d y|\\
=&|\tau^\prime(x)\int_{2x+3\epsilon}^{-\epsilon} \cot{\pi (y-x)} \, \d y|
=|\tau^\prime(x)\int_{x+3\epsilon}^{-\epsilon-x}\cot{(\pi y)} \, \d y|\\
=&\frac{1}{\epsilon^2}(x+3\epsilon)\int_{x+3\epsilon}^{-\epsilon-x}\frac{1}{2y} \, \d y
=\frac{1}{\epsilon}\ \frac{x+3\epsilon}{2\epsilon}\ln{\left(\frac{2\epsilon}{x+3\epsilon}-1\right)}\\
\leq &\frac{1}{\epsilon} \sup_{t\in [0,1/2]} t\ln{\left(\frac{1}{t}-1\right)}<\frac{1}{4\epsilon}.
\end{aligned}
\eeq
Hence, overall we have 
\beq
|\int_{\mathrm{supp}(\tau^\prime)}\tau^\prime(x)\cot{\pi(x-y)}\, \d y|<\frac{3}{4\epsilon}.
\eeq
Now let us consider with $x\in [-3\epsilon,-\epsilon]$, 
\beq
\begin{aligned}
  &|\mathcal{H}[\tau^\prime](x)-\int_{\mathrm{supp}(\tau^\prime)}\tau^\prime(x)\cot{\pi(x-y)}\, \d y|\\
=&\Big |\int_{\mathrm{supp}(\tau^\prime)}(\tau^\prime(y)-\tau^\prime(x))\,\cot{\pi (x-y)}  \, \d y \Big|\\
\leq &\int_{-3\epsilon}^{-\epsilon}\left|\frac{\tau^\prime(y)-\tau^\prime(x)}{\sin{\pi(y-x)}}\right| \, \d y + \int_{\epsilon}^{3\epsilon}\frac{|\tau^\prime(y)|+|\tau^\prime(x)|}{|\sin{\pi(y-x)}|}  \, \d y \\
\leq &\int_{-3\epsilon}^{-\epsilon} \frac{|x-y|\epsilon^{-2}}{2|x-y|} \, dy +\big (\frac{1}{4}+ \frac12\log2\big)\epsilon^{-1}
\le \frac{1.6}{\epsilon}.
\end{aligned}
\eeq
In summary, 
\beq
\|\mathcal{H}[\tau^\prime]\|_{L^\infty(\T)}<\frac{5}{2\epsilon}
\eeq
as claimed. 
\e{proof}

Combining \eqref{est:nu},  \eqref{eq:tau nu},  \eqref{eq:inftynorm2}, with Lemmas \ref{lm:muest}, \ref{lm:L1norm} and \ref{lm:Linftynorm}, we have
\begin{align*}
(\tau, \nu) \leq& \frac{9}{2} \varepsilon_0+5\frac{\varepsilon_1}{2\epsilon}+2\epsilon B_2(B-m)+\la \tau\ra \nu(\T) \\
\leq &\frac{9}{2}\varepsilon_0+5\frac{\varepsilon_1}{2\epsilon}+\epsilon \left(2B_2+\frac{4\pi}{\log{R/{R_1}}}\right)(B-m)\\
=&\frac{9}{2}\varepsilon_0+2B_3\sqrt{\varepsilon_1(B-m)},
\end{align*}
Note that our choice~\eqref{eq:epsviaeps1} of $\epsilon$ minimizes the contribution of $\varepsilon_1$.
Finally,  \eqref{eq:inftynorm1} concludes the proof of Lemma \ref{lm:LinftynorminsideH}.
\e{proof}

In order to prove the exponential integrability of $v-c$, and thus complete the proof of Lemma~\ref{lm:T1Splitting}, we invoke the following classical result about the Hilbert transform of bounded functions on the circle, see for example~\cite{Katz}. 

\be{lm}\label{lm:expint}
Let $f$ be a real-valued function on $\T$ such that $|f|\leq 1$. Then for any $0\leq \alpha<\frac{1}{2}\pi$,
$$\int_{\T} e^{\alpha |\mathcal{H}[f](x)|}\ \d x\leq \frac{2}{\cos{\alpha}}=2\sec\alpha.$$
\e{lm}

Applying Lemma \ref{lm:expint} to $f=(D^{-1}(\nu-J_{\epsilon}*\nu))/\|(D^{-1}(\nu-J_{\epsilon}*\nu))\|_{L^\infty(\T)}$, by \eqref{eq:uHnu2}, we have
\beq\label{est:expint}
\int_{\T} \exp\big( \beta |(u-J_{\epsilon}*u)(x)|\big) \,  \d x\leq {2}\,{\sec{\alpha}},
\eeq 
where $\alpha=\beta (\frac{9}{2}\varepsilon_0+2B_3\sqrt{\varepsilon_1(B-m)})<\frac{\pi}{2}$. 
Taking $\alpha=\frac{\pi}{4}$ in \eqref{est:expint}, then \eqref{eq:comparevu} yields (absorbing the constant $c$ into~$v$ for simpicity)
\beq\begin{aligned}\label{est:expint2}
&\int_{\T} \exp\Big(\frac{\pi |v(x)|}{18\varepsilon_0+8B_3\sqrt{\varepsilon_1 (B-m)}}\Big)\;\d x\\
& \leq 2\sqrt{2}\exp\Big(\frac{\pi \varepsilon_0+\pi (B_3+B_1/B_3)\sqrt{\varepsilon_1(B-m)}/2}{18\varepsilon_0+8B_3\sqrt{\varepsilon_1(B-m)}}\Big)\\
&\le 2\sqrt{2}\, \exp\Big(\pi\big [\frac{17}{144}+\frac{B_1}{16 B_3^2}\big]  \Big), 
\end{aligned}\eeq
which concludes the proof of Lemma~\ref{lm:T1Splitting}. 
$\hfill{} \Box$

\subsection*{Decay of the Fourier coefficients}

We conclude this section with an important decay estimate on the Fourier coefficients of the subharmonic function~$v$. 
This lemma will be used in Sec.\ref{sec:longsum}.

\be{lm}\label{lm:Fourier}
Let $v$ be as in Theorem \ref{Riesz}, then the Fourier coefficients of $v$ satisfy
$$|\hat{v}(k)|\leq \frac{C(R,R_1,R_2)}{|k|}(B-m),\ \ \text{for any}\ k\neq 0 ,$$
in which 
\beq\label{def:constantCRR1R2}
C(R,R_1,R_2)=\frac{1}{2 \log{ R/{R_1}}}+\frac{1}{2\pi} B_1(R,R_1,R_2).
\eeq
\e{lm}
\be{proof}
For any $k\neq 0$, we have
\beq\label{eq:vFourier}
|\hat{v}(k)|=\left|\int_{\T}(u(x)+\tilde{h}(x))e^{-2\pi i k x}\d  x\right|\leq \frac{1}{2\pi |k|} \left(|\widehat{u^\prime}(k)|+|\widehat{\tilde{h}^\prime}(k)| \right).
\eeq
By Lemma \ref{lm:T1Riesz}, we have $|\tilde{h}^\prime(x)|\leq  B_1(R,R_1,R_2)(B-m)$, hence 
\beq\label{eq:hFourier}
|\widehat{\tilde{h}^\prime}(k)|\leq B_1(R,R_1,R_2)(B-m).
\eeq
By \eqref{eq:uHnu}, \eqref{est:nu} and Theorem \ref{Riesz}, we have
\beq\label{eq:uFourier}
|\widehat{u^\prime}(k)|=|\widehat{\mathcal{H}[\nu]}(k)|=|\hat{\nu}(k)|\leq \pi \tilde{\mu}(\overline{D_1})\leq \frac{\pi(B-m)}{\log{R/{R_1}}}.
\eeq
In view of  \eqref{eq:vFourier}, \eqref{eq:hFourier} and \eqref{eq:uFourier} we infer that 
\begin{align*}
|\hat{v}(k)|
\leq \frac{1}{2\pi |k|}\left( B_1(R,R_1,R_2)(B-m)+ \frac{\pi(B-m)}{\log{R/{R_1}}}\right).
\end{align*}
as claimed. 
\e{proof}

\section{$\T^2$ Splitting Lemma}\label{sec:T2splitting}

Our applications to the skew-shift dynamics on $\T^{2}$ require a version of the splitting lemma in two variables. First, we formalize the class  of
plurisubharmonic functions that we will be working with. 

\begin{defn}\label{def:T2 v}
Let $v(z,w)$ be a continuous pluri-subharmonic function on $D_{R}\times D_{R}$, satisfying the following estimates for $R>1$,
\beq\label{eq:v 41}
\begin{aligned}
\left\lbrace
\begin{matrix}
v(z, w)\leq B_4(R)\ \text{for}\ \forall (z, w)\in D_R\times \partial D_1,\ \text{and}\ v(0,e(y))\ge m_4\ \text{for}\ \forall y\in \T,\\
\\
v(z, w)\leq B_5(R)\ \text{for}\ \forall (z, w)\in \partial D_1\times D_R,\ \text{and}\ v(e(x),0)\ge m_5\ \text{for}\ \forall x\in \T,\\
\\
|v(e(x), e(y))|\leq B_6\ \text{for}\ \forall (x,y)\in \T^2
\end{matrix}
\right.
\end{aligned}
\eeq 
For a function $f$ defined on a polydisk in $\C^2$ which contains $\T^{2}$, let us denote $f(e(x), e(y))$ by $f(x,y)$ for simplicity. In particular, we will write $v(x,y)$ on~$\T^{2}$. The average is denoted by  $\la f\ra_{\T^2}:=\int_{\T^2}f(x,y)\, \d x\d y$.
\end{defn}

Below, we will analyze a particular Schr\"{o}dinger cocycle over a skew shift base. At that opine,  we will specify the constants in Definitions~\ref{def:Poisson} and~\ref{def:T2 v}. But for now we develop more analytical machinery with these constants as parameters.  
Recall that for a Borel set $U$, the Lebesgue measure  will be written as  $|U|$.

\be{lm}\label{lm:T2Splitting}
Let $v$ be as in Definition~\ref{def:T2 v}, and assume \eqref{eq:epsklein}.   
Let $0<r<1$, $0<\varepsilon_1$ and $0<\varepsilon_0<\varepsilon_3<\varepsilon_2$.
Assume
$$v(x,y)=v_0(x,y)+v_1(x,y)+\la v\ra_{\T^2},$$
where $\|v_0\|_{L^\infty(\T^2)}<\varepsilon_0$, and $\|v_1\|_{L^1(\T^2)}<\varepsilon_1$.
Then
\begin{align*}
  \l|\left\lbrace (x,y)\in \T^2: |v(x,y)-\la v\ra_{\T^2}|>\varepsilon_2\right\rbrace\r|
 & < 2(2C_{0})^{\frac12} \exp\big( -\frac{\pi}{8\, \delta_{0}^{(1)} } \varepsilon_{3}\big)  \\
&\qquad + C_{0 }\exp\big( -\frac{\pi}{4 \, \delta_{0}^{(2)} } \varepsilon_{2}\big)
\end{align*}
in which $\frac{9}{2}$ and $B_3=B_3(R,R_1,R_2)$ are defined as in \eqref{def:constantsB3B4}, $C_{0}$ in~\eqref{eq:expintegrable},  and 
\beq\label{eq:delta def}
\begin{aligned}
\delta_{0}^{(1)} &:=   \frac{9}{2} \varepsilon_0+ 2 B_3\sqrt{\varepsilon_1^r (B_4-m_4)}  \\
\delta_{0}^{(2)} &:=   \frac{9}{2}\varepsilon_3+ 4  B_3\sqrt{B_6} \sqrt{\varepsilon_1^{1-r} (B_5-m_5)}
\end{aligned}
\eeq
\e{lm}
\be{proof}
Fix $0<r<1$. 
Let 
\beq\label{def:setA1}
A_1:=\left\lbrace y\in \T: \int_{\T}|v_1(x,y)|\ \mathrm{d}x<\varepsilon_1^{r}\right\rbrace.
\eeq
By Markov's inequality, we have 
\beq\label{est:A1mes}
|A_1^c|< \varepsilon_1^{1-r}.
\eeq
For any fixed $y\in A_{1}$, we have $\|v_0(\cdot,y)\|_{L^\infty(\T)}<\varepsilon_0$ and $\|v_1(\cdot, y)\|_{L^1(\T)}<\varepsilon_1^r$.
Applying Lemma \ref{lm:T1Splitting} in the $x$ variable, we then have
$$
\int_{\T} \exp\big( \frac{\pi}{4\,\delta_{0}^{(1)}} |v(x,y)-\la v\ra_{\T^2}|\big) \,\d x\leq C_{0}
$$
Integrating over $y\in A_1$, and interchanging the integrations, yields
\beq\label{T2:1}
\int_{\T}  \int_{A_{1}}\exp\big( \frac{\pi}{4\,\delta_{0}^{(1)}} |v(x,y)-\la v\ra_{\T^2}|\big) \,\d y\d x\leq C_{0}
\eeq
For $\gamma>0$, let us define
\beq\label{def:setA2}
A_2:=\left\lbrace x\in \T: \int_{A_1}\exp\big( \frac{\pi}{4\,\delta_{0}^{(1)}} |v(x,y)-\la v\ra_{\T^2}|\big) \,\d y \leq C_{0}\, \gamma^{-1} \right\rbrace.
\eeq
By Markov's inequality,
\beq\label{est:A2mes}
|A_2^c|<\gamma.
\eeq
For $x\in A_2$ and $\varepsilon_3>\varepsilon_0$, let us define
\beq\label{def:setA3}
A_3:=\left\lbrace y\in A_1: |v(x,y)-\la v\ra_{\T^2}|<\varepsilon_3\right\rbrace,
\eeq
Again, by Markov's inequality,
\beq\label{est:A3mes}
|A_3^c|<C_{0} \, \gamma^{-1} \exp\big(- \frac{\pi\varepsilon_3 }{4\,\delta_{0}^{(1)}}\big)
\eeq
Thus for $x\in A_2$, $\left\lbrace y\in \T: |v(x,y)-\la v\ra|>\varepsilon_3\right\rbrace\subseteq A_1^c\cup A_3^c$, with the following measure estimate
\beq\label{T2:2}
|\left\lbrace y\in \T: |v(x,y)-\la v\ra_{\T^2}|>\varepsilon_3\right\rbrace|\leq \varepsilon_1^{1-r}+C_{0} \, \gamma^{-1} \exp\big(- \frac{\pi\varepsilon_3 }{4\,\delta_{0}^{(1)}}\big).
\eeq
Here we divide into two different cases, depending on which term on the right-hand side dominates. 

\medskip

\subparagraph{{\bf Case 1:} $\varepsilon_1^{1-r}< C_{0} \, \gamma^{-1} \exp\big(- \frac{\pi\varepsilon_3 }{4\,\delta_{0}^{(1)}}\big) $.}
Then \eqref{T2:2} directly implies that for any $x\in A_2$,
\beq\label{T2:3}
\l|\left\lbrace y\in \T: |v(x,y)-\la v\ra_{\T^2}|>\varepsilon_3\right\rbrace\r|\leq 2C_{0} \, \gamma^{-1} \exp\big(- \frac{\pi\varepsilon_3 }{4\,\delta_{0}^{(1)}}\big)
\eeq
Together with \eqref{est:A2mes}, we conclude that
\beq\label{T2:4}
\l|\left\lbrace (x,y)\in \T^2: |v(x,y)-\la v\ra_{\T^2}|>\varepsilon_3 \right\rbrace \r|\leq \gamma+2C_{0} \, \gamma^{-1} \exp\big(- \frac{\pi\varepsilon_3 }{4\,\delta_{0}^{(1)}}\big)
\eeq

\subparagraph{{\bf Case 2:} $\varepsilon_1^{1-r}\ge  C_{0} \, \gamma^{-1} \exp\big(- \frac{\pi\varepsilon_3 }{4\,\delta_{0}^{(1)}}\big) $.}
Then for any $x\in A_2$, 
\beq\label{T2:5}
\l|\left\lbrace y\in \T: |v(x,y)-\la v\ra_{\T^2}|>\varepsilon_3\right\rbrace\r|\leq 2\varepsilon_1^{1-r}.
\eeq
For $x\in A_2$, let
\beq\label{T2def:v0v1tilde}
\begin{aligned}
\tilde{v}_{x,0}(y)&=(v(x,y)-\la v\ra_{\T^2})\ \one_{\{y\in \T: |v(x,y)-\la v\ra_{\T^2}|\leq \varepsilon_3\}} \\
\tilde{v}_{x,1}(y)&=(v(x,y)-\la v\ra_{\T^2})\ \one_{\{y\in \T: |v(x,y)-\la v\ra_{\T^2}|>\varepsilon_3\}}.
\end{aligned}
\eeq
Then \eqref{T2:5} implies, assuming $x\in A_2$, 
\beq\begin{aligned}\label{T2:6}
&v(x,y)=\tilde{v}_{x,0}(y)+\tilde{v}_{x,1}(y)+\la v\ra_{\T^2},\\
&\|\tilde{v}_{x,0}(\cdot)\|_{L^\infty(\T)}\leq \varepsilon_3,\\
&\|\tilde{v}_{x,1}(\cdot)\|_{L^1(\T)}\leq 2\,\varepsilon_1^{1-r}\|v(x, \cdot)-\la v\ra_{\T^2}\|_{L^\infty(\T)}\leq 4B_6\, \varepsilon_1^{1-r}.
\end{aligned}
\eeq
Applying Corollary \ref{cor:T1Splitting} in the $y$ variable, we obtain that for any $x\in A_2$ and any $\varepsilon_2>\varepsilon_3$, 
\beq\label{T2:7}
\l|\left\lbrace y\in \T: |v(x,y)-\la v\ra_{\T^2}|>\varepsilon_2\right\rbrace\r|\leq C_{0} \exp\big(- \frac{\pi\varepsilon_2 }{4\,\delta_{0}^{(2)}}\big)
\eeq
Together with \eqref{est:A2mes}, we then get
\beq\label{T2:8}
\l|\left\lbrace (x,y)\in \T^2: |v(x,y)-\la v\ra_{\T^2}|>\varepsilon_2\right\rbrace\r|\leq \gamma+  C_{0} \exp\big(- \frac{\pi\varepsilon_2 }{4\,\delta_{0}^{(2)}}\big).
\eeq

\medskip

Finally, we choose $\gamma$ to equalize the terms in~\eqref{T2:4}:  
\[
\gamma= (2C_{0})^{\frac12} \exp\big(- \frac{\pi\varepsilon_3 }{8\,\delta_{0}^{(1)}}\big)
\]
Then the estimate of Case 1, namely \eqref{T2:4}, yields
\beq\label{T2:4p}
\begin{aligned}
&\l|\left\lbrace (x,y)\in \T^2: |v(x,y)-\la v\ra_{\T^2}|>\varepsilon_2 \right\rbrace \r| \\
\leq &\l|\left\lbrace (x,y)\in \T^2: |v(x,y)-\la v\ra_{\T^2}|>\varepsilon_3 \right\rbrace \r|\\
\leq & 2(2C_{0})^{\frac12} \exp\big(- \frac{\pi\varepsilon_3 }{8\,\delta_{0}^{(1)}}\big)
\end{aligned}
\eeq
The estimate of Case 2, namely \eqref{T2:8}, becomes
\beq\label{T2:8p}
\begin{aligned}
&\l|\left\lbrace (x,y)\in \T^2: |v(x,y)-\la v\ra_{\T^2}|>\varepsilon_2\right\rbrace\r|\\
\leq &  (2C_{0})^{\frac12} \exp\big(- \frac{\pi\varepsilon_3 }{8\,\delta_{0}^{(1)}}\big) + C_{0} \exp\big(- \frac{\pi\varepsilon_2 }{4\,\delta_{0}^{(2)}}\big).
\end{aligned}
\eeq
Combining \eqref{T2:4p} with \eqref{T2:8p}, we conclude that 
\begin{align*}
&\l|\left\lbrace (x,y)\in \T^2: |v(x,y)-\la v\ra_{\T^2}|>\varepsilon_2\right\rbrace\r|\\
\leq & 2 (2C_{0})^{\frac12} \exp\big(- \frac{\pi\varepsilon_3 }{8\,\delta_{0}^{(1)}}\big) + C_{0} \exp\big(- \frac{\pi\varepsilon_2 }{4\,\delta_{0}^{(2)}}\big).
\end{align*}
as claimed. 
\end{proof}

\section{Avalanche Principle}

The Avalanche Principle (AP)  is a device to compare the logarithm of the norm of a long product $A_{n}A_{n-1}\ldots A_{2}A_{1}$ of matrices to the sum of the logarithms of the norms of shorter sections of the product. In the original formulation from~\cite{GS} for $\mathrm{SL}_2(\R)$ matrices the length of the chain was limited depending on the norms of the individual matrices $A_{j}$. The same restriction applied to the extension of the AP to $\mathrm{SL}_d(\R)$ matrices in~\cite{Sch}. Later, Duarte and Klein~\cite{DK} found a different proof of the AP which does not use impose any restriction on the length of the chain. Even though the older version~\cite{GS} would suffice for our purposes, we present the argument from~\cite{DK} with explicit constants. (These are not provided in \cite{DK}.)

Thus, this section is devoted to making the constants in Chapter~2 of \cite{DK2} effective (we mostly follow~\cite{DK2} instead of~\cite{DK} for the sake of simplicity). 
We use the same notation as \cite{DK2}, which we first recall.
Although we only need the results for $\mathrm{SL}_2(\R)$ matrices in this paper, we aim at proving more general results which are of independent interest.

Let $\mathrm{GL}_d(\R)$ be the general linear group of real $d\times d$ matrices. 
\begin{defn}
Given matrices $g_0, g_1,..., g_{n-1}\in \mathrm{GL}_d(\R)$, the expansion rift is the ratio
$$\rho(g_0, g_1,...,g_{n-1}):=\frac{\|g_{n-1}\cdots g_1 g_0\|}{\|g_{n-1}\|\cdots \|g_1\| \|g_0\|}\in (0,1].$$
\end{defn}

Given $g\in \mathrm{GL}_d(\R)$, let
$$s_1(g)\geq s_2(g)\geq ...\geq s_d(g)>0$$
denote the sorted singular values of $g$. The first singular value $s_1(g)$ is the operator norm
$$s_1(g)=\max_{x\in \R^d\setminus \{0\}}\frac{\|gx\|}{\|x\|}:=\|g\|.$$
The last singular value of $g$ is the least expansion factor of $f$, regarded as a linear transformation, and it can be characterized by 
$$s_d(g)=\min_{x\in \R^d\setminus \{0\}}\frac{\|gx\|}{\|x\|}=\|g^{-1}\|^{-1}.$$

\be{defn}
The gap (or the singular gap) of $g\in\mathrm{GL}_d(\R)$ is the ratio between its first and second singular values.
$$\mathrm{gr}(g):=\frac{s_1(g)}{s_2(g)}.$$
\e{defn}
\be{rmk}
If $g\in \mathrm{SL}_2(\R)$, then $\mathrm{gr}(g)=\|g\|^2$.
\e{rmk}

Let $\mathbb P(\R^d)$ denote the projective space.
Points in $\mathbb P(\R^d)$ are equivalence classes $\hat{x}$ of non-zero vectors $x\in \R^d$. We consider the projective distance $\delta:\mathbb{P}(\R^d)\times \mathbb{P}(\R^d)\rightarrow [0,1]$,
$$\delta(\hat{x}, \hat{y}):=\sin{(\angle (x,y))},$$
where $\angle$ is the length of the arc connecting $x$ and $y$.

\be{defn}
Given $g\in \mathrm{GL}_d(\R)$ such that $\mathrm{gr}(g)>1$, the most expanding direction of $g$ is the singular direction $\hat{\mf{v}}\in \mathbb{P}(\R^d)$ associated with the first singular value $s_1(g)$ of $g$. Let $\mf{v}(g)$ be any of the two unit vector representatives of the projective point $\hat{\mf{v}}(g)$. We set $\hat{\mf{v}}^*(g):=\hat{\mf{v}}(g^*)$ and $\mf{v}^*(g):=\mf{v}(g^*)$.
\e{defn}

Any matrix $g\in \mathrm{GL}_d(\R)$ maps the most expanding direction of $g$ to the most expanding direction of $g^*$, multiplying vectors by the factor $s_1(g)=\|g\|$.
$$g\mf{v}(g)=\pm s_1(g)\mf{v}^*(g).$$
The matrix $g$ also induces a projective map $\hat{g}:\pr (\R^d)\rightarrow \pr(\R^d)$, $\hat{g}(\hat{x}):=\hat{gx}$, for which one has
$$\hat{g}(\hat{\mf{v}}(g))=\hat{\mf{v}}^*(g)\ \ \text{and}\ \ \hat{g^*} (\hat{\mf{v}}^*(g))=\hat{\mf{v}}(g).$$

\be{thm}\label{thm:AP}
Let $n\geq 1$ and $0<\epsilon \leq\frac{1}{10}$. Given $0<\kappa\leq \frac{1}{10}\epsilon^2$ and $g_0,g_1,...,g_{n-1}\in \mathrm{GL}_d(\R)$, if
\begin{itemize}
\item (G) $\mathrm{gr}(g_i)\geq \kappa^{-1}$ for $j=0,1,...,n-1$,
\item (A) $\rho(g_{j-1}, g_j)\geq \epsilon$ for $j=1,2,...,n-1$,
\end{itemize}
then, writing $g^j:=g_{j-1}\cdots g_1g_0$, we have
\begin{itemize}
\item (i) $\max{\{ \delta(\hat{\mf{v}}(g^n), \hat{\mf{v}}(g_0)),\ \delta(\hat{\mf{v}}^*(g^n), \hat{\mf{v}}^*(g_{n-1})) \}}\leq 3\kappa\epsilon^{-1}$,
\item (ii) $e^{-5n\kappa/{\epsilon^2}}\leq \frac{\rho(g_0,g_1,...,g_{n-1})}{\rho(g_0,g_1)\cdots \rho(g_{n-2}, g_{n-1})}\leq e^{11n\kappa/{\epsilon^2}}$.
\end{itemize}
\e{thm}
The proof follows the general line of argumentation in \cite{DK2}, keeping track of the effective constants throughout.

\subsection*{Staging the proof}

The projective distance $\delta:\pr(\R^d)\times \pr(\R^d)\rightarrow [0,1]$ determines a complementary angle function $\alpha:\pr(\R^d)\times \pr(\R^d)\rightarrow [0,1]$, defined by
$$\alpha(\hat{x}, \hat{y}):=|\cos{(\angle (x,y))}|.$$

Let us also introduce the algebraic operation
$$a\oplus b:=a+b-ab.$$
For properties of $a\oplus b$, one may refer to Proposition 2.1 of \cite{DK2}.

\be{lm}\label{lm:DK22} 
Given $g\in \mathrm{GL}_d(\R)$ with $\mathrm{gr}(g)>1$, $\hat{x}\in \pr(\R^d)$ and a unit vector $x\in \hat{x}$, writing $\alpha=\alpha(\hat{x}, \hat{\mf{v}}(g))$ we have
\begin{itemize}
\item (a) $\alpha\leq \frac{\|gx\|}{\|x\|}\leq \sqrt{\alpha^2\oplus \mathrm{gr}(g)^{-2}}$,
\item (b) $\delta(\hat{g}(\hat{x}), \hat{\mf{v}}^*(g))\leq \alpha^{-1}\mathrm{gr}(g)^{-1}\delta(\hat{x}, \hat{\mf{v}}(g))$,
\item (c) The restriction of the map $\hat{g}:\pr(\R^d)\rightarrow \pr(\R^d)$ to the disk $\{\hat{x}\in \pr(\R^d):\ \delta(\hat{x}, \hat{\mf{v}}(g))\leq r\}$ has Lipschitz constant $\leq \frac{\pi}{2}\frac{r+\sqrt{1-r^2}}{\mathrm{gr}(g)(1-r^2)}$ with respect to the $\delta$-metric.
\end{itemize}
\e{lm}
\be{proof}
The factor $\frac{\pi}{2}$ in the Lipschitz constant is already explicit in the proof of Lemma 2.2 of \cite{DK2}.
\e{proof}

\be{cor}\label{cor:DK23}
Given $g\in \mathrm{GL}_d(\R)$ such that $\mathrm{gr}(g)\geq \kappa^{-1}$, define 
$$\Sigma_{\epsilon}:=\{\hat{x}\in \pr(\R^d):\ \alpha(\hat{x}, \hat{\mf{v}}(g))\geq \epsilon\}=B(\hat{\mf{v}}(g), \sqrt{1-\epsilon^2}).$$
Given a point $\hat{x}\in \Sigma_{\epsilon}$,
\begin{itemize}
\item (a) $\delta(\hat{g}(\hat{x}), \hat{g}(\hat{\mf{v}}(g)))\leq \kappa\epsilon^{-1}\delta(\hat{x}, \hat{\mf{v}}(g))$,
\item (b) The map $\hat{g}|_{\Sigma_{\epsilon}}\rightarrow \pr(\R^d)$ has Lipschitz constant $\leq \frac{\sqrt{2}\pi}{2} \kappa\epsilon^{-2}$.
\end{itemize}
\e{cor}
\be{proof}
(a) follows directly from (b) of Lemma \ref{lm:DK22}. 
(b) follows from (c) of Lemma \ref{lm:DK22} and the fact that $\epsilon+\sqrt{1-\epsilon^2}\leq \sqrt{2}$.
\e{proof}

\be{defn}
Given $g, g^\prime\in \mathrm{GL}_d(\R)$ with $\mathrm{gr}(g), \mathrm{g}(g^\prime)>1$, we define their lower angle as 
$$\alpha(g, g^\prime):=\alpha(\hat{\mf{v}}^*(g), \hat{\mf{v}}(g^\prime)).$$
The upper angle between $g$ and $g^\prime$ is 
$$\beta(g, g^\prime):=\sqrt{\mathrm{gr}(g)^{-2}\oplus \alpha(g, g^\prime)^2\oplus \mathrm{gr}(g^\prime)^{-2}}.$$
\e{defn}

\be{lm}\label{lm:DK24}
Given $g, g^\prime\in \mathrm{GL}_d(\R)$, if $\mathrm{gr}(g), \mathrm{gr}(g^\prime)>1$, then 
$$\alpha(g, g^\prime)\leq \rho(g, g^\prime)\leq \beta(g, g^\prime).$$
\e{lm}
This lemma has the following immediate corollary. It shows how the assumptions (G) and (A) in Theorem \ref{thm:AP} will be used.

\be{cor}\label{cor:AP(A)}
Given $g, g^\prime\in \mathrm{GL}_d(\R)$, if $\mathrm{gr}(g), \mathrm{gr}(g^\prime)\geq \kappa^{-1}$ and $\rho(g, g^\prime)\geq \epsilon$,
then 
$$\delta(\hat{\mf{v}}^*(g), \hat{\mf{v}}(g^\prime))\leq \sqrt{1-\frac{\epsilon^2}{1+2\frac{\kappa^2}{\epsilon^2}}}.$$
\e{cor}

We recall that $g^j:=g_{j-1}\cdots g_1g_0$.

\be{lm}\label{lm:DK25}  
If $\mathrm{gr}(g_j)>1$ for $j=0,1,...,n-1$, and $\mathrm{gr}(g^j)>1$ for $j=1,2,...,n$, then 
$$\prod_{j=1}^{n-1}\alpha(g^j, g_j)\leq \rho(g_0, g_1,..., g_{n-1})\leq \prod_{j=1}^{n-1}\beta(g^j, g_j).$$
\e{lm}

\subsection*{Proof of Theorem \ref{thm:AP}}
To simplify the notation, we will write $c_0=\frac{1}{10}$, $\hat{\mf{v}}_j:=\hat{\mf{v}}(g_j)$ and $\hat{\mf{v}}_j^*:=\hat{\mf{v}}^*(g_j)$ for $j=0,1,...,n-1$.
We also let
\begin{align*}
g_j=g_{2n-1-j}^*,\ \ \hat{\mf{v}}_j=\hat{\mf{v}}^*_{2n-1-j},\ \ \text{and}\ \ \hat{\mf{v}}_j^*=\hat{\mf{v}}_{2n-1-j}\ \ \text{for}\ j=n,n+1,...,2n-1.
\end{align*}
For each $i=0,1,...,2n-1$ and $j=0,1,...,2n-i$, set 
$$\hat{\mf{v}}_i^j:=\hat{g}_{i+j-1}\cdots \hat{g}_{i+1}\hat{g}_i \hat{\mf{v}}_i.$$
In terms of the notation above, we have $\widehat{(g^n)^* g^n}(\hat{\mf{v}}(g_0))=\hat{\mf{v}}_0^{2n}$ and $\hat{\mf{v}}_0=\hat{\mf{v}}_{2n-1}^1$.

By Assumption (A), we have $\rho(g_{j-1}, g_j)\geq \epsilon$ for $1\leq j\leq n-1$. Hence for $n+1\leq j\leq 2n-1$, 
\beq\label{eq:APA1}
\rho(g_{j-1}, g_j)=\rho(g_{2n-j}^*, g_{2n-j-1}^*)=\rho(g_{2n-j-1}, g_{2n-j})\geq \epsilon.
\eeq
Clearly, we also have 
\beq\label{eq:APA2}
\rho(g_{n-1}, g_n)=\rho(g_{n-1}, g_{n-1}^*)=\frac{\|g_{n-1}^*g_{n-1}\|}{\|g_{n-1}\|^2}=1.
\eeq
Therefore combining Assumption (A) with \eqref{eq:APA1} and \eqref{eq:APA2}, we have
\beq\label{eq:APA3}
\rho(g_{j-1}, g_j)\geq \epsilon\ \ \text{for}\ j=1,2,...,2n-1.
\eeq

We begin with the proof of statement (i). We will prove $\delta(\hat{\mf{v}}(g^n), \hat{\mf{v}}(g_0))\leq 3\kappa\epsilon^{-1}$. The other bound can be proved in exactly the same way.

First, we will show that for $\tilde{\epsilon}=t\epsilon$, $t=2/3$, we have
\be{lm}\label{APlm:1}
For any $1\leq j\leq 2n-1$,
$$\hat{g}_{j-1}(B(\hat{\mf{v}}_{j-1}, \sqrt{1-\tilde{\epsilon}^2}))\subseteq B(\hat{\mf{v}}_{j}, \sqrt{1-\tilde{\epsilon}^2}).$$
\e{lm}
\be{proof}
Taking any $\hat{x}\in B(\hat{\mf{v}}_{j-1}, \sqrt{1-\tilde{\epsilon}^2})$, we have 
$$\delta(\hat{x}, \hat{\mf{v}}_{j-1})=\sin(\angle (\hat{x}, \hat{\mf{v}}_{j-1}))\leq \sqrt{1-\tilde{\epsilon}^2}.$$
By (a) of Corollary \ref{cor:DK23}, 
\beq\label{AP:induction1}
\begin{aligned}
\delta(\hat{g}_{j-1}\hat{x}, \hat{g}_{j-1}\hat{\mf{v}}_{j-1})=\delta(\hat{g}_{j-1}\hat{x}, \hat{\mf{v}}_{j-1}^1)
=&\sin(\angle (\hat{g}_{j-1}\hat{x}, \hat{\mf{v}}_{j-1}^1))\\
\leq &\kappa\tilde{\epsilon}^{-1}\sqrt{1-\tilde{\epsilon}^2}\leq \frac{c_0\epsilon}{t}\sqrt{1-t^2\epsilon^2}.
\end{aligned}
\eeq
By \eqref{eq:APA3} and Corollary \ref{cor:AP(A)}, we have
\beq\label{AP:induction2}
\begin{aligned}
\delta(\hat{\mf{v}}^*(g_{j-1}), \hat{\mf{v}}(g_{j}))=\delta(\hat{\mf{v}}_{j-1}^{1}, \hat{\mf{v}}_{j})
=&\sin(\angle(\hat{\mf{v}}_{j-1}^{1}, \hat{\mf{v}}_{j})))\\
 \leq &\sqrt{1-\frac{\epsilon^2}{1+2\frac{\kappa^2}{\epsilon^2}}}\leq \sqrt{1-\frac{\epsilon^2}{1+2c_0^2\epsilon^2}}.
\end{aligned}
\eeq
Let $\theta_1=\angle (\hat{g}_{j-1}\hat{x}, \hat{\mf{v}}_{j-1}^1)$ and $\theta_2=\angle(\hat{\mf{v}}_{j-1}^{1}, \hat{\mf{v}}_{j}))$.
Then 
\beq\label{AP:induction3}
\begin{aligned}
\delta(\hat{g}_{j-1}\hat{x}, \hat{\mf{v}}_j)
\leq &|\cos{\theta_1}|\sin{\theta_2}+\sin{\theta_1}|\cos{\theta_2}|\\
=&\sqrt{1-\sin^2{\theta_1}}\sin{\theta_2}+\sin{\theta_1}\sqrt{1-\sin^2{\theta_2}}:=f(\sin{\theta_1}, \sin{\theta_2}).
\end{aligned}
\eeq
With $f(x,y)=y\sqrt{1-x^2}+x\sqrt{1-y^2}$, it is easy to see that both $\frac{\partial f}{\partial x}$ and $\frac{\partial f}{\partial y}$ have the same sign as 
$\sqrt{1-x^2}\sqrt{1-y^2}-xy$.
Thus both $\frac{\partial f}{\partial x}$ and $\frac{\partial f}{\partial y}$ are positive if $x^2+y^2<1$.

By \eqref{AP:induction1} and \eqref{AP:induction2}, we have
\beq\label{AP:induction4}
\sin^2{\theta_1}+\sin^2{\theta_2}\leq \frac{c_0^2\epsilon^2}{t^2}(1-t^2\epsilon^2)+1-\frac{\epsilon^2}{1+2c_0^2\epsilon^2}<1.
\eeq
Here it is enough to have that for $\tilde{\epsilon}=t\epsilon$, $$\frac{c_0^2}{t^2}+c_0^2\epsilon^2<1.$$

Then \eqref{AP:induction3} implies
\beq\label{AP:induction5}
\begin{aligned}
\delta(\hat{g}_{j-1}\hat{x}, \hat{\mf{v}}_j) 
\leq &f\left(\frac{c_0\epsilon}{t}\sqrt{1-t^2\epsilon^2}, \sqrt{1-\frac{\epsilon^2}{1+2c_0^2\epsilon^2}}\right)\\
<&\left(1-\frac{c_0^2\epsilon^2}{2t^2}(1-t^2\epsilon^2)\right)\left(1-\frac{\epsilon^2}{2+4c_0^2\epsilon^2}\right)+\frac{c_0\epsilon^2}{t}(1-\frac{c_0^2\epsilon^2}{2})\\
<&\sqrt{1-t^2\epsilon^2}=\sqrt{1-\tilde{\epsilon}^2}.
\end{aligned}
\eeq
(By our choice of  $c_0=\frac{1}{10}$ and $t=\frac{2}{3}$,  the $\epsilon^2$ coefficients of \eqref{AP:induction5} corresponds to $-\frac{9}{800}-\frac{1}{2}+\frac{3}{20}<-\frac{2}{9}$.)
\e{proof}

This lemma has the following intermediate corollary.
\be{cor}\label{APlm:1cor}
For any $1\leq j\leq 2n-1$ and $1\leq m\leq 2n-j-1$, we have
$$\hat{g}_{j+m-1}\cdots \hat{g}_{j}\hat{g}_{j-1}B(\hat{\mf{v}}_{j-1}, \sqrt{1-\tilde{\epsilon}^2})\subseteq B(\hat{\mf{v}}_{j+m}, \sqrt{1-\tilde{\epsilon}^2}).$$
\e{cor}

Next, let us show
\be{lm}\label{APlm:2}
For any $0\leq j\leq 2n-1$, for any $\hat{x}\in B(\hat{\mf{v}}_j, \sqrt{1-\tilde{\epsilon}^2})$,
$$
\delta(\hat{g}_{2n-1}\cdots \hat{g}_{j+2}\hat{g}_{j}\hat{x}, \hat{\mf{v}}_{j}^{2n-j})\leq \kappa \tilde{\epsilon}^{-1} \left(\frac{\sqrt{2}\pi}{2}\kappa \tilde{\epsilon}^{-2}\right)^{2n-j-1}
$$
\e{lm}
\be{proof}
By Corollary \ref{APlm:1cor}, for any $0\leq m\leq 2n-j-1$, we have that both the two elements $\hat{g}_{j+m-1}\cdots \hat{g}_{j+1}\hat{g}_{j}\hat{x}$ and $\hat{\mf{v}}_j^m$ belong to $B(\hat{\mf{v}}_{j+m}, \sqrt{1-\tilde{\epsilon}^2})$. 
Hence by (a) of \ref{cor:DK23}, we have that for $m=0$,
\beq\label{eq:AP2}
\delta(\hat{g}_{j}\hat{x}, \hat{\mf{v}}_j^{1})=\delta(\hat{g}_{j}\hat{x}, \hat{g}_{j}\hat{\mf{v}}_j)\leq \kappa\tilde{\epsilon}^{-1} \delta(\hat{x}, \hat{\mf{v}}_j)<\kappa\tilde{\epsilon}^{-1}.
\eeq
For $1\leq m\leq 2n-j-1$, by (b) of \ref{cor:DK23}, we have
\beq\label{eq:AP3}
\begin{aligned}
\delta(\hat{g}_{j+m}\hat{g}_{j+m-1}\cdots \hat{g}_{j+1}\hat{g}_{j}\hat{x}, \hat{\mf{v}}_j^{m+1})
=&\delta(\hat{g}_{j+m}\hat{g}_{j+m-1}\cdots \hat{g}_{j+1}\hat{g}_{j}\hat{x}, \hat{g}_{j+m}\hat{\mf{v}}_j^m)\\
\leq &\frac{\sqrt{2}\pi}{2}\kappa\tilde{\epsilon}^{-2}\delta(\hat{g}_{j+m-1}\cdots \hat{g}_{j+1}\hat{g}_{j}\hat{x}, \hat{\mf{v}}_j^m).
\end{aligned}
\eeq
\eqref{eq:AP2} and \eqref{eq:AP3} imply that
\beq\label{eq:AP4}
\delta(\hat{g}_{2n-1}\cdots \hat{g}_{j+2}\hat{g}_{j}\hat{x}, \hat{\mf{v}}_{j}^{2n-j})\leq \kappa \tilde{\epsilon}^{-1} \left(\frac{\sqrt{2}\pi}{2}\kappa \tilde{\epsilon}^{-2}\right)^{2n-j-1}
\eeq
as desired. 
\e{proof}

In particular, combining Corollary \ref{APlm:1cor} with Lemma \ref{APlm:2}, we have the following corollary
\be{cor}\label{APlm:2cor}
For any $1\leq j\leq 2n-1$,
$$
\delta(\hat{\mf{v}}_{j-1}^{2n-j+1}, \hat{\mf{v}}_{j}^{2n-j})\leq \kappa \tilde{\epsilon}^{-1} \left(\frac{\sqrt{2}\pi}{2}\kappa \tilde{\epsilon}^{-2}\right)^{2n-j-1}
$$
\e{cor}

Next, we will show
\be{lm}\label{APlm:3}
For any $\hat{x}\in B(\hat{\mf{v}}_0, \sqrt{1-\tilde{\epsilon}^2})$, we have
$$
\delta(\widehat{(g^n)^* g^n}\hat{x}, \hat{\mf{v}}_0)\leq 3\kappa\epsilon^{-1},\qquad 
\delta(\hat{\mf{v}}_0^{2n}, \hat{\mf{v}}_{2n-1}^1) \leq 3\kappa\epsilon^{-1}.
$$
\e{lm}
\be{proof}
By Corollary \ref{APlm:2cor}, we have
\beq\label{eq:AP5}
\delta(\hat{\mf{v}}_0^{2n}, \hat{\mf{v}}_{2n-1}^1)\leq \sum_{j=1}^{2n-1}\delta(\hat{\mf{v}}_{j-1}^{2n-j+1}, \hat{\mf{v}}_{j}^{2n-j})\leq \kappa\tilde{\epsilon}^{-1}\sum_{j=0}^{2n-2}\left(\frac{\sqrt{2}\pi}{2}\kappa\tilde{\epsilon}^{-2}\right)^{j}.
\eeq
By Lemma \ref{APlm:2},  
\beq\label{eq:AP6}
\delta(\hat{g}_{2n-1}\cdots \hat{g}_1\hat{g}_0\hat{x}, \hat{\mf{v}}_0^{2n})\leq \kappa\tilde{\epsilon}^{-1}\left(\frac{\sqrt{2}\pi}{2}\kappa\tilde{\epsilon}^{-2}\right)^{2n-1}.
\eeq
Hence, combining \eqref{eq:AP5} with \eqref{eq:AP6}, we conclude that 
$$\delta(\widehat{(g^n)^* g^n}\hat{x}, \hat{\mf{v}}(g_0))\leq \kappa\tilde{\epsilon}^{-1}\sum_{j=0}^{2n-1}\left(\frac{\sqrt{2}\pi}{2}\kappa\tilde{\epsilon}^{-2}\right)^{j} \leq \frac{\kappa\tilde{\epsilon}^{-1}}{1-\frac{\sqrt{2}\pi}{2}\kappa\tilde{\epsilon}^{-2}}\leq \frac{1}{t-\frac{\pi c_0}{\sqrt{2}t}}\kappa\epsilon^{-1}<3\kappa\epsilon^{-1}. $$
\e{proof}

We are now ready to give the

\dashuline{Proof of (i).}
Lemma \ref{APlm:3} shows that $\widehat{(g^n)^* g^n}$ maps the ball $B(\hat{\mf{v}}_0, \sqrt{1-\tilde{\epsilon}^2})$ into itself. By Corollary \ref{cor:DK23}, it has the contracting Lipschitz factor $\leq (\frac{\sqrt{2}\pi}{2}\kappa\tilde{\epsilon}^{-2})^{2n}\ll 1$. Therefore the map $\widehat{(g^n)^* g^n}$ has a unique fixed point in $B(\hat{\mf{v}}_0, \sqrt{1-\tilde{\epsilon}^2})$, call it $x_*$. Lemma \ref{APlm:3} implies that 
\beq\label{eq:x*close}
\delta(x_*,\hat{\mf{v}}(g_0))<3\kappa\epsilon^{-1}.
\eeq

The claim will follow once we prove that $x_*=\hat{\mf{v}}(g^n)$. Since $\hat{\mf{v}}(g^n)$ is a fixed point of $\widehat{(g^n)^* g^n}$, it suffices to prove that $\hat{\mf{v}}(g^n)\in B(\hat{\mf{v}}_0, \sqrt{1-\tilde{\epsilon}^2})$. (Notice that $\widehat{(g^n)^* g^n}$ has several fixed points, one for every eigenvalue of $(g^n)^* g^n$.)

Let $\de_*:=\de(\hat{\mf{v}}(g^n), x_*)$. We will show that $\de_*=0$. For any unit vector $v$, we have
$$
\l|\langle\mf{v}(g^n),v\rangle\r|
=\frac{1}{(s_1(g^n))^2}\l|\langle(g^n)^*g^n\mf{v}(g^n),v\rangle\r|
\leq \l|\langle\mf{v}(g^n),\frac{(g^n)^*g^nv}{|(g^n)^*g^n v|}\rangle\r|,
$$
where we used that $|(g^n)^*g^n v|\leq (s_1(g^n))^2$. This lifts to a relation on projective space:
\beq\label{eq:this}
\de(\hat{\mf{v}}(g^n),\widehat{(g^n)^* g^n} \hat{v})\leq \de(\hat{\mf{v}}(g^n),\hat{v}).
\eeq
We apply this with $\hat{v}=\hat{v}_*$ as the ``halfway point'' between $\hat{\mf{v}}(g^n)$ and $x_*$, i.e., $\hat{v}_*$ satisfies 
$$
\de(\hat{\mf{v}}(g^n),\hat{v}_*)=\de(x_*,\hat{v}_*)=\frac{\de_*}{2}.
$$
(This $\hat{v}_*$ can be constructed by following the arc that connects $\hat{\mf{v}}(g^n)$ with $x_*$, assuming that $\mf{v}(g^n)\neq x_*$.) 

Notice that $\hat{v}_*\in B(\hat{\mf{v}}_0, \sqrt{1-\tilde{\epsilon}^2})$ because \eqref{eq:x*close} gives
$$
\de(\hat{\mf{v}}_0,\hat{v}_*)\leq \de(\hat{\mf{v}}_0,x_*)+\frac{\de_*}{2}
\leq \frac{3\kappa}{\eps}+\frac{\de_*}{2}
\leq \frac{3\eps}{10}+\frac{1}{2}<\sqrt{1-\tilde\eps^2}.
$$
Recall that $\widehat{(g^n)^* g^n}$ maps the ball $B(\hat{\mf{v}}_0, \sqrt{1-\tilde{\epsilon}^2})$ into itself with Lipschitz factor $\leq L_0:=(\frac{\sqrt{2}\pi}{2}\kappa\tilde{\epsilon}^{-2})^{2n}\ll 1$. Since $\widehat{(g^n)^* g^n}x_*=x_*$, we have
$$
\de(\widehat{(g^n)^* g^n}\hat{v}_*,x_*)\leq L_0\de(\hat{v}_*,x_*)=L_0\frac{\de_*}{2}.
$$
We combine this bound and \eqref{eq:this}, with $\hat{v}=\hat{v}_*$, to conclude that
$$
\de_*
=\de(\hat{\mf{v}}(g^n), x_*)
\leq \de(\hat{\mf{v}}(g^n),\widehat{(g^n)^* g^n}\hat{v}_*)
+\de(\widehat{(g^n)^* g^n}\hat{v}_*,x_*)
\leq \l(1+\frac{L_0}{2}\r)\de_*.
$$
Since $L_0<1$, this implies $\de_*=0$, i.e., $x_*=\hat{\mf{v}}(g^n)$. Consequently, \eqref{eq:x*close} reads $\delta(\hat{\mf{v}}(g^n), \hat{\mf{v}}(g_0))\leq 3\kappa\epsilon^{-1}$ as claimed in (i) of Theorem \ref{thm:AP}. The other bound in (i) can be proved in exactly the same way.

\dashuline{Proof of (ii)}
By Lemma \ref{lm:DK25}, we have
\beq\label{eq1:AP2}
\prod_{j=1}^{n-1}\frac{\alpha(g^j,g_j)}{\beta(g_{j-1},g_j)}\leq \frac{\rho(g_0,...,g_{n-1})}{\prod_{j=1}^{n-1}\rho(g_{j-1}, g_j)}\leq \prod_{j=1}^{n-1}\frac{\beta(g^j,g_j)}{\alpha(g_{j-1},g_j)}.
\eeq
We will show that the factors
$$\frac{\alpha(g^j,g_j)}{\beta(g_{j-1},g_j)}\ \ \text{and}\ \ \frac{\beta(g^j,g_j)}{\alpha(g_{j-1},g_j)}$$
are very close to $1$, with logarithms of order $\kappa\epsilon^{-2}$. 
From conclusion (i), applied to the sequence of matrices $g_0,g_1,...,g_{j-1}$, we have
\beq\label{eq2:AP2}
\max\{\delta(\hat{\mf{v}}^*(g^j), \hat{\mf{v}}^*(g_{j-1})),\ \delta(\hat{\mf{v}}(g^j), \hat{\mf{v}}(g_0)) \}\leq 3\kappa\epsilon^{-1}
\eeq
for all $1\leq j\leq n$. From \eqref{eq2:AP2}, we deduce that 
\beq\label{eq3:AP2}
\begin{aligned}
\l|\log\frac{\alpha(g^j, g_j)}{\alpha(g_{j-1},g_j)}\r|
\leq &\frac{|\alpha(g^j,g_j)-\alpha(g_{j-1},g_j)|}{\min\{\alpha(g^j, g_j), \alpha(g_{j-1}, g_j)\}}\\
\leq &\frac{\delta(\hat{\mf{v}}^*(g^j), \hat{\mf{v}}^*(g_{j-1}))}{\min\{\alpha(g^j, g_j), \alpha(g_{j-1}, g_j)\}}
\leq \frac{3\kappa\epsilon^{-1}}{\min\{\alpha(g^j, g_j), \alpha(g_{j-1}, g_j)\}}.
\end{aligned}
\eeq
We estimate the minimum as follows, using \eqref{eq2:AP2} and Corollary \ref{cor:AP(A)},
\beq\label{eq:minal}
\begin{aligned}
\min\{\alpha(g^j, g_j), \alpha(g_{j-1}, g_j)\}
\geq& \alpha(g_{j-1}, g_j)-|\alpha(g^j,g_j)-\alpha(g_{j-1},g_j)|\\
\geq&
\alpha(g_{j-1}, g_j)-3\kappa\epsilon^{-1}\\
\geq& \frac{\eps}{\sqrt{1+2\frac{\kappa^2}{\eps^2}}}-3\frac{\kappa}{\epsilon}\\
\geq& \frac{2\eps}{3}.
\end{aligned}
\eeq
In the last step, we used $\frac{\kappa}{\eps}\leq \frac{\eps}{10}$. Returning to \eqref{eq3:AP2}, we have shown
\beq\label{eq:11/3}
\l|\log\frac{\alpha(g^j, g_j)}{\alpha(g_{j-1},g_j)}\r|\leq \frac{9}{2}\frac{\kappa}{\epsilon^{2}}.
\eeq
From the definition of the upper angle $\beta$ and Corollary \ref{cor:AP(A)}, we also have
\beq\label{eq4:AP2}
\l|\log{\frac{\beta(g_{j-1},g_j)}{\alpha(g_{j-1},g_j)}}\r|
\leq \log{\sqrt{1+2\frac{\kappa^2}{\alpha(g_{j-1},g_j)}}}
\leq \log{\sqrt{1+2\frac{\kappa^2}{\epsilon^2}}}\leq \frac{\kappa^2}{\epsilon^2}.
\eeq
Hence \eqref{eq:11/3} and \eqref{eq4:AP2} yield
$$\l|\log{\frac{\alpha(g^j,g_j)}{\beta(g_{j-1},g_j)}}\r|\leq \frac{9}{2}\frac{\kappa}{\epsilon^{2}}+\frac{\kappa^2}{\epsilon^2}<5\frac{\kappa}{\epsilon^{2}}.$$
Together with \eqref{eq1:AP2}, this implies the lower bound in (ii), i.e.,
$$ e^{-5n\kappa/{\epsilon^2}}\leq \frac{\rho(g_0,g_1,...,g_{n-1})}{\rho(g_0,g_1)\cdots \rho(g_{n-2}, g_{n-1})}.$$

For the upper bound, we argue similarly. The only difference occurs in the analog of \eqref{eq4:AP2}, i.e., the estimate
on
$$
\l|\log{\frac{\beta(g^j,g_j)}{\alpha(g^j,g_j)}}\r|.
$$
To bound this quantity, we need to control the gap ratio $(\mathrm{gr}(g^{j}))^{-1}$. This control is provided by the following lemma. 

\be{lm}\label{lm:gap}
We have $\l(\mathrm{gr}(g^{j})\r)^{-1}\leq \kappa':=20 \frac{\kappa}{\eps}$.
\e{lm}

We postpone the proof of the lemma for now. It gives
$$
\l|\log{\frac{\beta(g^j,g_j)}{\alpha(g^j,g_j)}}\r|
\leq \log{\sqrt{1+2\frac{(\kappa')^2}{\alpha(g^j,g_j)}}}
\leq 600\frac{\kappa^2}{\eps^3}.
$$
In the second step, we used that $\alpha(g^j,g_j)\geq 2\eps/3$ by \eqref{eq:minal}. Recalling \eqref{eq:11/3}, one has
$$\l|\log{\frac{\beta(g^j,g_j)}{\al(g_{j-1},g_j)}}\r|\leq \frac{9}{2}\frac{\kappa}{\epsilon^{2}}+600\frac{\kappa^2}{\eps^3}<11\frac{\kappa}{\epsilon^{2}}.$$
In the last estimate, we used that $600\kappa/\eps\leq 60\eps\leq 6$. 

By \eqref{eq1:AP2}, this proves the upper bound in (ii), i.e.,
$$ \frac{\rho(g_0,g_1,...,g_{n-1})}{\rho(g_0,g_1)\cdots \rho(g_{n-2}, g_{n-1})}\leq e^{11n\kappa/{\epsilon^2}}.$$

It remains to prove Lemma \ref{lm:gap}. For this part, we follow page 71 in \cite{DK} and make the constants precise. From Proposition 2.28 in \cite{DK}, we see that
$$
\l(\mathrm{gr}(g^j)\r)^{-1}=\|(D\hat{g^j})_{\hat{\mf{v}}(g^j)}\|.
$$ 
where $(D\hat{g^j})_{\hat{\mf{v}}(g^j)}$ is the derivative of $\hat{g^j}:\mathbb{P}(\R^d)\to \mathbb{P}(\R^d)$, evaluated at $\hat{\mf{v}}(g^j)$. The norm of this derivative is bounded by the Lipschitz constant in a neighborhood of $\hat{\mf{v}}(g^j)$. Since the Lipschitz constant is $\leq L^j$, with $L:=\frac{\sqrt{2}\pi}{2}\kappa\tilde{\epsilon}^{-2}$, everywhere on the ball $B(\hat{\mf{v}}_0, \sqrt{1-\tilde{\epsilon}^2})$, we immediately obtain the bound $(\mathrm{gr}(g^j))^{-1}\leq L^j$. However, this is not good enough for our purposes (note that $L$ is an order one quantity in general). 

We may improve the estimate as follows: Applying statement (i) of the theorem with $n=j$, we obtain that 
$$
\de(\hat{\mf{v}}(g^j),\hat{\mf{v}}(g_0))<3\frac{\kappa}{\epsilon}.
$$
Now, a calculation based on Proposition 2.28 in \cite{DK} shows that
$$
\|(D\hat{g_0})_{\hat{\mf{v}}(g^j)}-(D\hat{g_0})_{\hat{\mf{v}}(g_0)}\|\leq 12\pi \frac{\kappa}{\eps}
$$
and therefore
$$
\|(D\hat{g_0})_{\hat{\mf{v}}(g^j)}\|\leq \|(D\hat{g_0})_{\hat{\mf{v}}(g_0)}\|+12\pi \frac{\kappa}{\eps}
=\mathrm{gr}(g_0)^{-1}+12\pi \frac{\kappa}{\eps}
\leq \kappa +12\pi \frac{\kappa}{\eps}
\leq \l(12\pi+\frac{1}{10}\r) \frac{\kappa}{\eps}.
$$
Finally, we apply the chain rule and estimate the derivative of the product $g_{j-1}\ldots g_1$ by its Lipschitz constant $L^{j-1}$, which satisfies $L^{j-1}\leq L\leq \frac{9\pi}{40\sqrt{2}}$ for $j\geq 2$. Therefore, 
$$
\begin{aligned}
\l(\mathrm{gr}(g^j)\r)^{-1}
=&\|(D\hat{g^j})_{\hat{\mf{v}}(g^j)}\|
\leq \|(D\hat{g}_{j-1}\ldots \hat{g}_1)_{\hat{\mf{v}}(g^j)} g\| \|(D\hat{g}_0)_{\hat{\mf{v}}(g^j)}\|\\
\leq& L\l(12\pi+\frac{1}{10}\r) \frac{\kappa}{\eps}
<20 \frac{\kappa}{\eps}.
\end{aligned}
$$
This proves Lemma \ref{lm:gap} and hence completes the proof of Theorem \ref{thm:AP}.
\qed

\section{Herman's regularization}\label{sec:Herman}
\subsection{Monodromy matrices}

One has  $T_{\omega}^n(x,y)=(x+ny+\frac{n(n-1)}{2}\omega, y+n\omega)$ for any positive integer~$n$, where $T$  is the skew-shift with frequency $\omega$.
Denote   the projection of $\T^2$ onto the first coordinate by $\mathcal{P}$, viz.\ $\mathcal{P}(x,y)=x$.

We consider the Schr\"odinger operator 
$$(H_{\lambda, \omega, x,y}\psi)_n=\psi_{n+1}+\psi_{n-1}+2\lambda \cos{(2\pi \mathcal{P}(T_{\omega}^n(x,y)))}\psi_n$$
with $\lambda>0$. 
This equation has the following cocycle reformulation
\beq\label{def:transfer}
\begin{aligned}
\left(\begin{matrix}
\psi_{n+1}\\
\psi_n
\end{matrix}\right)
&=
\left(\begin{matrix}
E-2\lambda\cos{(2\pi \mathcal{P}(T_{\omega}^n(x,y)))}\ \ &-1\\
1\ \ &0
\end{matrix}\right)
\left(\begin{matrix}
\psi_{n}\\
\psi_{n-1}
\end{matrix}\right) \\
& =:M(\lambda, E; T^n_{\omega}(x,y))
\left(\begin{matrix}
\psi_{n}\\
\psi_{n-1}
\end{matrix}\right).
\end{aligned}
\eeq
Define the  transfer matrices $M_n(\lambda, E; x,y)$ to be  
\beq\label{def:n-transfer}
M_n(\lambda, E; x,y)=
\left\lbrace
\begin{matrix}
\prod_{j=n}^{1} M(\lambda, E; T^j_{\omega}(x,y)),\ \ &n\geq 1,\\
\\
\mathrm{Id},\ \ &n=0,\\
\\
(M_{-n}(\lambda, E; T_{\omega}^{n+1}(x,y)))^{-1},\ \ &n<0. 
\end{matrix}
\right.
\eeq
Then
\begin{align*}
\left(\begin{matrix}
\psi_{n+1}\\
\psi_{n}
\end{matrix}\right)=
M_n(E; x,y)
\left(\begin{matrix}
\psi_1\\
\psi_0
\end{matrix}\right),
\end{align*}
The following function on $\T^{2}$ plays a fundamental role in our analysis: 
\beq\label{def:un}
\begin{aligned}
& u_n(\lambda, E; x,y):
=\frac{1}{n}\log{\|M_n(\lambda, E; x,y)\|}\\
=&\frac{1}{n}\log{\left\lVert \prod_{j=n}^1 \left(\begin{matrix}
E-\lambda e^{2\pi i (x+jy+\frac{j(j-1)}{2}\omega)}-\lambda e^{-2\pi i (x+jy+\frac{j(j-1)}{2}\omega)}\ \ &-1\\
1\ \ &0
\end{matrix} \right)\right\rVert}.
\end{aligned}
\eeq 
Let $z=e^{2\pi ix}, w=e^{2\pi i y}, a=e^{\pi i \omega}$, as well as 
\beq\label{eq:Matrix A}
A_{\lambda}(\lambda, E, z, w, a) := \left(\begin{matrix}
Ez w -\lambda z^2 w^{2} a -\lambda \ol{a} \ \ &-z w\\
z w \ \ &0
\end{matrix} \right) . 
\eeq
Then for $(z,w)\in \partial D_1\times \partial D_1$,
\beq\label{def:un=vn}
\begin{aligned}
u_n(\lambda, E; x,y)
=&\frac{1}{n}\log{\left\lVert \prod_{j=n}^1 \left(\begin{matrix}
E-\lambda z w^{j} a^{j(j-1)} -\lambda z^{-1} {w}^{-j}\, \ol{a}^{j(j-1)} \ \ &-1\\
1\ \ &0
\end{matrix} \right)\right\rVert}\\
=&\frac{1}{n}\log{\left\lVert \prod_{j=n}^1 \left(\begin{matrix}
Ez w^{j}-\lambda z^2 w^{2j} a^{j(j-1)} -\lambda \ol{a}^{j(j-1)} \ \ &-z w^{j}\\
z w^{j}\ \ &0
\end{matrix} \right)\right\rVert}\\
=&\frac{1}{n} \log \Big \| \prod_{j=n}^1A_{\lambda}(E, z, w^{j}, a^{j(j-1)}) \Big \| \\
=&:v_n(\lambda, E; z,w).
\end{aligned}
\eeq
Note that $v_n(\lambda, E; z,w)$ is a pluri-subharmonic function on $\C^2$.  {\em Herman's regularization} refers to the transition from the first to the second line in~\eqref{def:un=vn}, which removes the singularities $z^{-1}$ and~$w^{-1}$.   Note that 
\beq\label{eq:origin}
v_n(\lambda, E; 0,w) = v_n(\lambda, E; z,0) = \log\lambda 
\eeq
For simplicity, we will write $A$ instead of $A_{\lambda}$, since $\lambda$ will be a fixed parameter within some range.  As a general rule for the arguments of the matrix function~$A$, the complex variables $z,w$ will belong to some disk $D_{R}$, whereas $|a|=1$ and $E$ will be real-valued within some range. We will also keep $0<\lambda<1$.

\subsection{Explicit bounds on the monodromy matrices}

As a first step towards obtaining the explicit constants in Definition~\ref{def:Poisson} and~\ref{def:T2 v} we prove the following bounds on $v_{n}$. 

\be{lm}\label{lm:Herman}
Let $0<\lambda<1$, and $R_3\geq 1$. Define
\beq\label{def:ClambdaR}
U(\lambda, R_3):=\frac{1}{2}\log{\left(\left(\lambda (1+\frac{1}{R_3^{2}})^2+\frac{2}{R_3} \right)^2+\frac{2}{R_3^{2}}\right)}.
\eeq
Then for any $E$ with $|E|\leq 2+2\lambda$,
$v_n(\lambda, E, z, w)$ from~\eqref{def:un=vn} satisfies the following estimates 
\begin{itemize}
\item for any $w\in \partial{D}_1$, 
\beq\label{eq:B5Herman}
v_n(\lambda, E; z, w)\leq 2\log{R_3}+U(\lambda, R_3)\ \text{for}\ \forall z\in \ol{D_{R_3}},\ \ \text{and}\ v_n(\lambda, E; 0,w)=\log{\lambda}.
\eeq
\item for any $(z,w)\in \partial{D}_1\times \ol{D_{R_3}}$, we have upper bound
\beq\label{eq:B6Herman}
v_n(\lambda, E; z, w)\leq 
(n+1)\log{R_3} + U(\lambda, R_3) 
\eeq
We also have $v_n(\lambda, E; z,0)=\log{\lambda}$ for any $z\in \partial{D}_1$.
\item for any $(z,w)\in \partial{D}_1\times \partial{D}_1$,
\beq\label{eq:B7Herman}
|v_n(\lambda, E; z, w)|\leq U(\lambda, 1).
\eeq
\end{itemize}
\e{lm}
\begin{rmk}
Let us note that 
\beq\label{eq:4U>1}
4U(\lambda, 1)\geq 2\log{6}>1.
\eeq
\end{rmk}

\be{proof}
Clearly \eqref{eq:B7Herman} follows from \eqref{eq:B5Herman} with $R_3=1$.
   We will use that  for any complex-valued matrix
   \[
   \|A\|^{2} = \| A^{*}A\| \le \Tr(A^{*}A) 
   \]
   For $A$ as in~\eqref{eq:Matrix A} this means that 
   \begin{align*}
   \|A(E, z, w, a)\|^{2} &\le |Ezw - \lambda z^{2}w^{2}a - \lambda\bar a|^{2}+ 2|z|^{2}|w|^{2} \\ 
   & \le (|E| |zw| + \lambda |zw|^{2} + \lambda)^{2} + 2|z|^{2}|w|^{2} 
   \end{align*}
whence
\beq\label{eq:vnleqtrace}
\begin{aligned}
|v_n(E; z,w)|
\leq &\frac{1}{n}\sum_{j=1}^{2n}\log{\|A(E, z, w^{j}, a^{j(j-1)})\|}\\
\leq &\frac{1}{2n}\sum_{j=1}^{n}\log{\left( \left( \lambda (|z|^2 |w|^{2j}+ 1)+|E| |z| |w|^{j} \right)^2+ 2|z|^2 |w|^{2j}\right)}.
\end{aligned}
\eeq
For $w\in \partial{D}_1$ and $|z|\leq R_3$, \eqref{eq:vnleqtrace} yields
\beq
\begin{aligned}
|v_n(E; z,w)|
\leq &\frac{1}{2n}\sum_{j=1}^{n}\log{\left( \left( \lambda (R_3^2+1)+|E| R_3 \right)^2+ 2R_3^2 \right)}\\
\leq &2\log{R_3}+ \frac{1}{2}\log{\left( \left( \lambda (1+\frac{1}{R_3^2})+\frac{2+2\lambda}{R_3} \right)^2+ \frac{2}{R_3^2} \right)}\\
=& 2\log{R_3}+U(\lambda, R_3).
\end{aligned}
\eeq
This proves \eqref{eq:B5Herman}.

Next, we turn to \eqref{eq:B6Herman}. 
For $z\in \partial{D}_1$ and $|w|\leq R_3$, \eqref{eq:vnleqtrace} yields
\beq\nn
\begin{aligned}
|v_n(E; z,w)|
\leq &\frac{1}{2n}\sum_{j=1}^{n}\log{\left( \left( \lambda (R_3^{2j}+1)+|E| R_3^{j} \right)^2+ 2R_3^{2j} \right)}\\
\leq &{(n+1)}\log{R_3}+\frac{1}{2n}\sum_{j=1}^{n}\log{\left( \left( \lambda (1+\frac{1}{R_3^{2j}})+(2+2\lambda) \frac{1}{R_3^{j}} \right)^2+ \frac{2}{R_3^{2j}} \right)}.
\end{aligned}
\eeq
Note that the summands are maximized at $j=1$, which gives us the constant $2U(\lambda, R_3)$. Hence, in total
\[
|v_n(E; z,w)| \le (n+1)\log{R_3} + U(\lambda, R_3)
\]
as claimed. 
\e{proof}

\section{Long Sums of Skew-Shift Functions}\label{sec:longsum}
In this section we establish a key large-deviation estimate on the ergodic averages of a pluri-subharmonic function, as defined above, over a long skew-shift orbit. 
The argument is based on~\cite[Lemma 2.6]{BGS}, but deviates from that reference in ways which are essential for our purposes. The precise dependence on all parameters is made explicit and effective. This leads to a somewhat cumbersome formulation which is, however, absolutely necessary for the main application.  We wish to point out  one technical feature of our version of this argument, namely that we only use a trivial bound on the number theoretic divisor function, see the constant $C^*$ below. We have found this to lead to the best constants. We also remark that significant gains in the following proposition would  lead to dramatic improvements in the inductive machinery that we use to control the Lyapunov exponent, cf.~the next two sections. At this point, however, it is not clear how to obtain such gains. 

Recall constants $B_1(R,R_1,R_2)$ is as in \eqref{def:constantsB0B1B2}, $B_3(R,R_1,R_2)$ is as in \eqref{def:constantsB3B4}, and $B_4(R), B_5(R), m_4, m_5$ are as in \eqref{eq:B56m56}. 
In the following we will write $B_1, B_3, B_4, B_5$ and omit the dependence on the radii.

\be{prop}\label{lm:longsumwithtrivialbdd}
Let $\omega=\frac{\sqrt{5}-1}{2}$ be the golden ratio.
Let $v$ be defined as in the beginning of Section~\ref{sec:T2splitting}, let $C=C(R,R_1,R_2)$ be the constant as in \eqref{def:constantCRR1R2}, and impose Definition~\ref{def:epssmall}. 
Let $\delta\in (0, 1/2)$ and $\delta_2,\delta_3>0$ be constants. 
Assuming 
\begin{enumerate}[label=(\roman*).]
\item $C(B_5-m_5) \leq K^{\delta}$,
\item $K\geq 38$,
\item $\exp{\Big(4(\log{K})^{\delta_2}\Big)}\geq K+1$,
\item $21 K^{-\frac{9}{10}+\frac{9}{5}\delta} {(\ln{K})^{\frac{9}{10}+\frac{9}{5}\delta_2}}+4C(B_4-m_4) \leq K^{\delta}(\log{K})^{\delta_2}$.
\end{enumerate}
Then for any positive parameter $C_2>0$, we have
\begin{align*}
&\l| \left\lbrace (x,y)\in \T^2:\ \l|\frac{1}{K}\sum_{k=1}^K v\circ T_{\omega}^k(x,y)-\la v\ra \r|>  \varepsilon_4 \right\rbrace \r|\\
\leq &2\sqrt{2} \, \exp \Big(\frac{\pi}{4} \big [\frac{17}{36}+ \frac{B_1}{4 B_3^2} -\varepsilon_5  \big]  \Big)\\
&\qquad +\sqrt{2} (C(B_4-m_4))^{-1} {K^{\frac{1}{5}-\frac{2\delta}{5}}}{(\ln{K})^{-\frac{1}{5}-\frac{2\delta_2}{5}}} \exp{\left(-2(\ln{K})^{\delta_2}\right)},
\end{align*}
where
\beq\label{def1:epsilon45}
\begin{aligned}
\varepsilon_4&=C_2 {K^{-\frac{1}{10}+\frac{\delta}{5}}}{(\ln{K})^{\frac{1}{10}+\frac{\delta_2}{5}+\delta_3}},\\
\varepsilon_5&=C_2 \left(472.5 +3.2 B_3(B_4-m_4) \sqrt{C} \right)^{-1} (\log{K})^{\delta_3}.
\end{aligned}
\eeq
\e{prop}

Before proving Proposition \ref{lm:longsumwithtrivialbdd}, we will review some background of continued fractions.
\subsection{Continued fractions}
Each $\omega\in [0,1)$ has the following unique expansion
\beq\label{def:conti}
\omega=\frac{1}{a_1+\frac{1}{a_2+\frac{1}{a_3+\frac{1}{\cdots}}}},
\eeq
where $a_i\in \N_+$. 
We will denote this expansion by $\omega=[a_1,a_2,...]$
If $\omega\in \Q$, the expansion is finite, while it is infinite for irrational $\omega$.

Let $\omega\in [0,1)\setminus \Q$, let
\beq\label{def:pnqn}
\frac{p_n}{q_n}=\frac{1}{a_1+\frac{1}{a_2+\frac{1}{\cdots+\frac{1}{a_n}}}}
\eeq
be the continued fraction approximants of $\omega$.
These approximants satisfy the following three properties:
\beq\label{eq:anqnqn+1}
q_{n+1}=a_{n+1}q_n+q_{n-1}\ \ \text{with}\ q_0:=1.
\eeq
\beq\label{eq:qnbest}
\|k\omega\|_{\T}\geq \|q_n\omega\|_{T}\ \ \text{for any}\ q_n\leq |k|<q_{n+1}.
\eeq
\beq\label{eq:qnqn+1}
\frac{1}{q_{n+1}+q_n}\leq \|q_n\omega\|_{\T}\leq \frac{1}{q_{n+1}}.
\eeq

A very special number we are interested in is the golden ration $\omega=\frac{\sqrt{5}-1}{2}$, it is well-known that $\omega$ has continued fraction expansion with $a_i\equiv 1$ for any $i\geq 1$.
Then by \eqref{eq:anqnqn+1}, we have $q_{n+1}=q_n+q_{n-1}$ with $q_0=q_1=1$. It is then easy to find out that for any $n\geq 0$,
Hence for any $n\geq 0$, $q_{n+1}\leq 2q_n$.
Then by \eqref{eq:qnbest} and \eqref{eq:qnqn+1}, we have the following property of the golden ratio:

\be{prop} The golden ratio satisfies 
\beq\label{eq:3}
\|k\omega\|_{\T}\geq \frac{1}{3|k|}\ \ \text{for any}\ k\neq 0.
\eeq
\e{prop}

The optimal bound here is  $\frac{\sqrt{5}+1}{2}+$, but the constant $3$ is sufficient. 
We will use the following corollary of~\eqref{eq:3} repeatedly. 

\be{cor}\label{cor:gold1} 
The golden ratio satisfies the following two properties: 
\begin{itemize}
\item Let $\ell$ be a positive integer such that $\| \ell\omega\|_{\T}\leq \sigma$, then $\ell\geq \frac{1}{3\sigma}$.
\item 
Let $\ell, \tilde{\ell}$ be two distinct positive integers such that $\max{(\| \ell\omega\|_{\T}, \|\tilde{\ell}\omega\|_{\T})}\leq \sigma$, then $|\ell-\tilde{\ell}|\geq \frac{1}{6\sigma}$.
\end{itemize}
\e{cor}

In order to control small divisors we will rely on the following two propositions.  

\be{lm}
For $\theta\in \R$, we have
\beq\label{eq:sinlb}
\left|\sin{\left(\frac{\pi}{2}\theta\right)}\right|\geq \|\theta\|_{\T}.
\eeq
\e{lm}
\be{proof}
If $\theta\in \Z+\frac{1}{2}$, $\left|\sin{\left(\frac{\pi}{2}\theta\right)}\right|=\sin{\left(\frac{\pi}{4}\right)}>\frac{1}{2}=\|\theta\|_{\T}$.

If $\theta\notin \Z+\frac{1}{2}$, then there exists a unique $k\in \Z$, such that $\theta=k+\|\theta\|_{\T}$ (if $\theta\in [k, k+\frac{1}{2})$), or $\theta=k-\|\theta\|_{\T}$ (if $\theta\in (k-\frac{1}{2},k)$). 
If $k$ is an even number, then $\left| \sin{\left(\frac{\pi}{2}\theta\right)}\right|=\sin{\left(\frac{\pi}{2}\|\theta\|_{\T}\right)}\geq \|\theta\|_{\T}$, in which we used $\sin{x}\geq \frac{2}{\pi}x$ for $0\leq x\leq \frac{\pi}{2}$.
If $k$ is odd, then $\left| \sin{\left(\frac{\pi}{2}\theta\right)}\right|=\cos{\left(\frac{\pi}{2}\|\theta\|_{\T}\right)}\geq \cos{\left(\frac{\pi}{4}\right)}>\frac{1}{2}\geq \|\theta\|_{\T}$.
\e{proof}

We will also use the following two estimates:
\be{lm}
For any positive integer $R$,
\beq\label{eq:sumofexp}
\l| \sum_{k=1}^{R} e(k\ell\omega)\r|\leq \min{\left(R, \frac{2}{2\sin{(\pi \| \ell\omega\|_{\T})}}\right)}  \leq \min{\left(R, \frac{1}{2\| \ell\omega\|_{\T}} \right)},
\eeq
and 
\beq\label{eq:sumofexp/2}
\l| \sum_{k=1}^{R} e\left(\frac{1}{2}k\ell \omega \right) \r|\leq \min{\left(R, \frac{1}{\| \ell\omega \|_{\T}} \right)}.
\eeq
\e{lm}
\be{proof}
For $\theta\notin \Z$, we have
\begin{align*}
\left| \sum_{k=1}^R e(k\theta)\right|=\min{\left(R, \left|\frac{e(\theta)-e((k+1)\theta)}{1-e(\theta)}\right| \right)}
\leq \min{\left(R, \frac{1}{\left| \sin{(\pi \theta)} \right|}\right)},
\end{align*}
then \eqref{eq:sumofexp} follows from taking $\theta=\ell\omega$, and using $\sin{(\pi x)}\geq {2}x$ for $0\leq x\leq \frac{1}{2}$.
\eqref{eq:sumofexp/2} follows from taking $\theta=\frac{1}{2}\ell \omega$ and employing \eqref{eq:sinlb}.
\e{proof}

\subsection{Proof of Proposition \ref{lm:longsumwithtrivialbdd}}
Let $\hat{v}(\ell,y)$ and $\hat{v}(x, \ell)$ denote the Fourier coefficients relative to the first and second variables,  respectively,
and by $\hat{v}(\ell_1, \ell_2)$ we mean the Fourier transform in both variables. 
For simplicity, let us omit the dependence of $C(R,R_1,R_2), B_3(R,R_1,R_2)$ on the radii.

First, we note the following estimate as a corollary of Lemma \ref{lm:Fourier}.
\be{cor}\label{cor:uFourier}
For any $\ell\neq 0$, we have
\beq\label{cor:uFourier1}
\begin{aligned}
\sup_{x\in \T}|\hat{v}(x, \ell)|\leq \frac{C}{|\ell |}(B_5-m_5),\\
\sup_{y\in \T}|\hat{v}(\ell, y)|\leq \frac{C}{|\ell |}(B_4-m_4),
\end{aligned}
\eeq
and
\beq\label{cor:uFourier2}
\begin{aligned}
(\sum_{\ell_1\in \Z} |\hat{v}(\ell_1, \ell_2)|^2)^{\frac{1}{2}}\leq \frac{C}{| \ell_2|}(B_5-m_5),\ \ \text{for any}\ \ell_2\neq 0,\\
(\sum_{\ell_2\in \Z} |\hat{v}(\ell_1, \ell_2)|^2)^{\frac{1}{2}}\leq \frac{C}{| \ell_1|}(B_4-m_4),\ \ \text{for any}\ \ell_1\neq 0,
\end{aligned}
\eeq
\e{cor}
\be{proof}
Note that \eqref{cor:uFourier1} follows directly from Lemma \ref{lm:Fourier}. On the other hand, $$(\sum_{\ell_1\in \Z} |\hat{v}(\ell_1, \ell_2)|^2)^{\frac{1}{2}}=\|\hat{v}(\cdot, \ell_2)\|_{L^2(\T)}\leq \sup_{x\in \T} |\hat{v}(x, \ell_2)|.$$
Hence, 
\eqref{cor:uFourier2} reduces to \eqref{cor:uFourier1}.
\e{proof}

With some positive integer $p_1$ to be determined, let
\beq\label{eq:uxlow-high}
\begin{aligned}
v(x,y)=&\sum_{|\ell_1|\leq p_1}\hat{v}(\ell_1,y)e(\ell_1 x)+\sum_{|\ell_1|>p_1}\hat{v}(\ell_1, y) e(\ell_1 x)\\
=:&v_1(x,y)+\tilde{v}_1(x,y),
\end{aligned}
\eeq
where $v_1$ and $\tilde{v}_1$ are the low and high frequency parts,  respectively. 

By Corollary \ref{cor:uFourier}, 
\beq\label{eq:tildeu1L2}
\begin{aligned}
\sup_{y\in \T}\|\tilde{v}_1(\cdot, y)\|_{L^1(\T)}
\leq &\left(\sum_{|\ell_1|>p_1} \sup_{y\in \T} |\hat{v}(\ell_1, y)|^2\right)^{\frac{1}{2}}\\
\leq &C_{5}\, \left(\sum_{|\ell_1|>p_1}\frac{1}{\ell_1^2}\right)^{\frac{1}{2}}(B_4-m_4)\\
\leq &\sqrt{2}C\, (B_4-m_4)  p_1^{-\frac{1}{2}}.
\end{aligned}
\eeq

Next we further decompose $v_1$ into low and high frequency parts in the $y$ variable. 
With some positive integer $p_2$ to be determined, let
\beq\label{eq:u1ylow-high}\begin{aligned}
v_1(x,y)&=\sum_{\substack{|\ell_1|\leq p_1\\ |\ell_2|>p_2}}\hat{v}(\ell_1, \ell_2) e(\ell_1 x+\ell_2 y)+\sum_{\substack{|\ell_1|\leq p_1\\ |\ell_2|\leq p_2}}\hat{v}(\ell_1, \ell_2) e(\ell_1 x+\ell_2 y) \\
&=:v_2(x,y)+v_3(x,y).
\end{aligned}
\eeq
By Corollary \ref{cor:uFourier}, we have
\beq\label{eq:u2L1}
\begin{aligned}
\|v_2(x,y)\|_{L^1(\T^2)} \leq & \|v_2(x,y)\|_{L^2(\T^2)}
=\left( \sum_{|\ell_1|\leq p_1,\ |\ell_2|>p_2}|\hat{v}(\ell_1, \ell_2)|^2\right)^{\frac{1}{2}}\\
\leq &\left( \sum_{\ell_1\in \Z,\ |\ell_2|>p_2}|\hat{v}(\ell_1, \ell_2)|^2\right)^{\frac{1}{2}} 
\leq \left( \sum_{|\ell_2|>p_2}\frac{C^2}{\ell_2^2}(B_5-m_5)^2\right)^{\frac{1}{2}}\\
<& \sqrt{2}C(B_5-m_5) p_2^{-\frac{1}{2}}.
\end{aligned}
\eeq
Hence, by Markov's inequality,
\beq\label{eq:lsbady}
\begin{aligned}
       &\l| \left\lbrace y\in \T:\ \frac{1}{K}\int_{\T}\l| \sum_{k=1}^K v_2\circ T_{\omega}^k(x,y)\r| \d x> t \right\rbrace \r| \\
 \leq &   \sqrt{2}\, C(B_5-m_5) p_2^{-\frac{1}{2}}\, t^{-1}. 
\end{aligned}
\eeq
We denote the set on the left-hand side of \ref{eq:lsbady} by $\mathcal{A}(t)$.

Now let us consider $v_3$, which will lead to small divisor problems.
By Corollary \ref{cor:uFourier} and the fact that $\hat{v}(0,0)=\la v\ra$, separating the cases $\ell_{1}=0$, $\ell_{2}=0$, and $\ell_{1}\ell_{2}\ne0$, yields 
\beq\label{def:lsS1S2}
\begin{aligned}
       &\sup_{(x,y)\in \T^2}\l| \frac{1}{K}\sum_{k=1}^K v_3\circ T^k_{\omega}(x,y)-\la v\ra \r|\\
\leq &\frac{1}{K}\sum_{\substack{|\ell_1|\leq p_1\\ |\ell_2|\leq p_2\\ |\ell_1|+|\ell_2|\neq 0}} |\hat{v}(\ell_1, \ell_2)| \l|\sum_{k=1}^K e\left(\ell_1\left(ky+\frac{k(k-1)\omega}{2}\right)+\ell_2k\omega\right)\r|\\
\leq & \frac{C(B_5-m_5)}{K}\sum_{1\leq |\ell_2|\leq p_2}\frac{1}{|\ell_2|} \l|\sum_{k=1}^{K} e(\ell_2 k\omega)\r|\\
+& \frac{C(B_4-m_4)}{K}\sum_{1\leq |\ell_1|\leq p_1}\frac{1}{|\ell_1|} \l|\sum_{k=1}^K e\left(\ell_1 \left(ky+\frac{k(k-1)\omega}{2}\right)\right)\r| \\
+& \frac{1}{K}\sum_{1\leq |\ell_2|\leq p_2}\ \sum_{1\leq |\ell_1|\leq p_1} |\hat{v}(\ell_1,\ell_2)| \l|\sum_{k=1}^K e\left(\ell_1 \left(ky+\frac{k(k-1)\omega}{2}\right)+\ell_2k\omega\right)\r|
\end{aligned}
\eeq
We now separately consider the sums appearing on the previous three lines. First, 
\[
S_{1}:=  \frac{1}{K}\sum_{1\leq |\ell_2|\leq p_2}\frac{1}{|\ell_2|} \l|\sum_{k=1}^{K} e(\ell_2 k\omega)\r|
\]
Second, by Cauchy-Schwarz, 
\beq
\begin{aligned}
& K^{-1}\sum_{1\leq |\ell_1|\leq p_1}\frac{1}{|\ell_1|} \l|\sum_{k=1}^K e\left(\ell_1 \left(ky+\frac{k(k-1)\omega}{2}\right)\right)\r|  \\
& \leq K^{-1} \left(\sum_{1\leq |\ell_1|\leq p_1}\frac{1}{\ell_1^2}\right)^{\frac{1}{2}} \left( \sum_{1\leq |\ell_1|\leq p_1} \l|\sum_{k=1}^K e\left(\ell_1 \left(ky+\frac{k(k-1)\omega}{2}\right) \right)\r|^2 \right)^{\frac{1}{2}} \\
& \le 2K^{-1}\left( \sum_{1\leq |\ell_1|\leq p_1} \l|\sum_{k=1}^K e\left(\ell_1 \left(ky+\frac{k(k-1)\omega}{2}\right) \right)\r|^2 \right)^{\frac{1}{2}} =: S_{2}
\end{aligned}
\eeq
And, finally, 
\begin{align*}
& K^{-1}\sum_{1\leq |\ell_2|\leq p_2}\ \sum_{1\leq |\ell_1|\leq p_1} |\hat{v}(\ell_1,\ell_2)| \l|\sum_{k=1}^K e\left(\ell_1 \left(ky+\frac{k(k-1)\omega}{2}\right)+\ell_2k\omega\right)\r|  \\
&\leq  K^{-1}\sum_{1\leq |\ell_2|\leq p_2}\ \left(\sum_{1\leq |\ell_1|\leq p_1} |\hat{v}(\ell_1,\ell_2)|^2\right)^{\frac{1}{2}} \left( \sum_{1\leq |\ell_1|\leq p_1} \l|\sum_{k=1}^K e\left(\ell_1 \left(ky+\frac{k(k-1)\omega}{2}\right)+\ell_2k\omega\right)\r|^2 \right)^{\frac{1}{2}}\\ 
&\leq C(B_5-m_5) K^{-1}\!\!\!\!\sum_{1\leq |\ell_2|\leq p_2}\!\! \frac{1}{|\ell_2|} \left( \sum_{1\leq |\ell_1|\leq p_1} \l|\sum_{k=1}^K e\left(\ell_1 \left(ky+\frac{k(k-1)\omega}{2}\right)+\ell_2k\omega\right)\r|^2 \right)^{\frac{1}{2}} \\
& =: C(B_5-m_5) S_{3}
\end{align*}
Returning to  \eqref{def:lsS1S2} we conclude that 
\beq\label{def:lsS1S2*}
\begin{aligned}
       &\sup_{(x,y)\in \T^2}\l| \frac{1}{K}\sum_{k=1}^K v_3\circ T^k_{\omega}(x,y)-\la v\ra \r|\\
\leq &  C(B_5-m_5)S_1+C(B_4-m_4)S_2+C(B_5-m_5) S_3.
\end{aligned}
\eeq

\subsubsection{Estimate of $S_1$}
Applying \eqref{eq:sumofexp} to $S_1$, we infer that 
\beq\label{eq:lsS1}
\begin{aligned}
S_1
< &2 \sum_{1\leq \ell_2\leq p_2}\frac{1}{\ell_2} \min{\left(1, \frac{1}{2K\| \ell_2\omega\|_{\T}}\right)}\\
=& 2 \sum_{1\leq \ell_2\leq p_2} \one_{\{\ell_2:\ \| \ell_2\omega\|\le \frac{1}{2K}\}} \frac{1}{\ell_2} \\ 
+& 2 \sum_{j=1}^{2^j\leq 2K} \sum_{1\leq \ell_2\leq p_2}\one_{\{\ell_2:\ \frac{2^{j-1}}{2K}< \| \ell_2\omega\|_{\T}\le \frac{2^j}{2K}\}} \frac{1}{\ell_2} \min{\left(1, \frac{1}{2K\| \ell_2\omega\|_{\T}}\right)}\\
=:&S_{1,1}+S_{1,2}.
\end{aligned}
\eeq

\subsubsection*{Estimate of $S_{1,1}$}
\beq\label{eq:lsS111}
S_{1,1}= 2 \sum_{1\leq \ell_2\leq p_2} \one_{\{\ell_2:\ \| \ell_2\omega\|\leq \frac{1}{2K}\}} \frac{1}{\ell_2}.
\eeq
By Corollary \ref{cor:gold1}, if for some $\ell\geq 1$, $\| \ell\omega\|_{\T}\leq \frac{1}{2K}$, then
\beq\label{eq:lsS112}
\ell\geq \frac{2}{3}K.
\eeq
Moreover, if for some distinct $\ell, \tilde{\ell}\geq 1$, $\max{(\| \ell\omega\|_{\T}, \|\tilde{\ell}\omega\|_{\T})}\leq \frac{1}{2K}$, then
\beq\label{eq:lsS113}
|\ell-\tilde{\ell}|\geq \frac{1}{3}K.
\eeq
By \eqref{eq:lsS111}, \eqref{eq:lsS112} and \eqref{eq:lsS113}, we have
\beq\label{eq:lsS114}
S_{1,1}\leq 2 \sum_{\ell=1}^{\lfloor 3 p_2/K\rfloor-1} \frac{3}{\ell+1}\frac{1}{K}<\frac{6}{K}\ln{\frac{3 p_2}{K}}.
\eeq
For this estimate, and from this point on, we assume that $p_{2}\ge K$. 

\subsubsection*{Estimate of $S_{1,2}$}
\beq\label{eq:lsS121}
\begin{aligned}
S_{1,2}
\leq &\frac{1}{K} \sum_{j=1}^{2^j\leq 2K} \sum_{1\leq \ell_2\leq p_2}\one_{\{\ell_2:\ \frac{2^{j-1}}{2K}< \| \ell_2\omega\|_{\T}\leq \frac{2^j}{2K}\}} \frac{1}{\ell_2 \| \ell_2\omega\|_{\T}}\\
\leq &2 \sum_{j=1}^{2^j\leq 2K} \sum_{1\leq \ell_2\leq p_2}\one_{\{\ell_2:\ \frac{2^{j-1}}{2K}< \| \ell_2\omega\|_{\T}\leq \frac{2^j}{2K}\}} \frac{1}{2^{j-1} \ell_2}
\end{aligned}
\eeq
By Corollary \ref{cor:gold1}, if for some $\ell\geq 1$, $\| \ell\omega\|_{\T}\leq \frac{2^j}{2K}$, then we have
\beq\label{eq:lsS122}
\ell\geq \frac{2}{2^j 3}K.
\eeq
Moreover, if for some distinct $\ell,\tilde{\ell}\geq 1$, $\max{(\| \ell\omega\|_{\T}, \| \tilde{\ell}\omega\|_{\T})}\leq \frac{2^j}{2K}$, then
\beq\label{eq:lsS123}
|\ell-\tilde{\ell}|\geq \frac{1}{2^j 3}K.
\eeq
By \eqref{eq:lsS121}, \eqref{eq:lsS122} and \eqref{eq:lsS123}, we have
\beq\label{eq:lsS124}
\begin{aligned}
S_{1,2}
\leq &2 \sum_{j=1}^{2^j \leq 2K}\ \ \sum_{\ell=1}^{\lfloor 2^j 3p_2/K\rfloor -1} \frac{6}{\ell+1}\frac{1}{K}\\
\leq &\frac{12}{K} \sum_{j=1}^{2^j\leq 2K} \ln{\frac{2^j 3 p_2}{K}}.\\
\leq &\frac{12}{K} \left( \frac{\ln{2}}{2} (\log_2{K}+1)^2 +(\log_2{K}+1) \ln{\frac{3 p_2}{K}}\right).
\end{aligned}
\eeq
Putting \eqref{eq:lsS1}, \eqref{eq:lsS114} and \eqref{eq:lsS124} together, we have
\beq\label{eq:lsS1final}
S_1\leq \frac{6}{K}\left(\ln{2} (\log_2{K}+2)^2 +(2\log_2{K}+3) \ln{\frac{3 p_2}{K}}\right).
\eeq
Henceforth we assume  $p_{2}\ge K>23$, which allows us to 
 simplify \eqref{eq:lsS1final} into 
\beq\label{eq:lsS1final2}
S_1\leq 26 \frac{\ln{K}}{K}\ln{p_2}.
\eeq

\subsubsection{Estimate of $S_2$ and $S_3$} 
In order to estimate $S_2$ and $S_3$, we use the well-known method of Weyl-differencing, cf.~for example~\cite{Mont}.  As a first step, 
\begin{align*}
  &\l|\sum_{k=1}^K e\left(\ell_1 \left(ky+\frac{k(k-1)\omega}{2} \right)+\ell_2k\omega \right)\r|^2\\ 
  =&\left(\sum_{k=1}^K e \left( \ell_1 \left(ky+\frac{k(k-1)\omega}{2} \right)+\ell_2k\omega\right) \right) \left(\sum_{j=1}^K e\left( -\ell_1\left(jy+\frac{j(j-1)\omega}{2}\right)-\ell_2j\omega \right) \right)\\
=&\sum_{j,k=1}^K e\left( \ell_1 \left( (y-\frac{\omega}{2})(k-j)+\frac{\omega}{2} (k^2-j^2)\right)+ \ell_2\omega (k-j) \right)
\end{align*}
Let $\ell=k+j$ and $m=k-j$, hence $\ell\equiv m$ (mod 2). 
Let us denote the sum over even integers $j$'s by $\evensum_j$, and odd integers by $\oddsum_j$.
\beq\label{eq:Weyl1}
\begin{aligned}
  &\l|\sum_{k=1}^K e\left(\ell_1 \left(ky+\frac{k(k-1)\omega}{2}\right)+\ell_2k\omega \right)\r|^2\\ 
=&\evensumm_{m=1-K}^{K-1}\ \ \evensumm_{\ell=2+|m|}^{2K-|m|} e\left(\ell_1m \left(y+\frac{(\ell-1)\omega}{2} \right)+\ell_2 m\omega\right)\\
  &\qquad\qquad\qquad\qquad\quad
  +\oddsumm_{m=1-K}^{K-1}\ \ \oddsumm_{\ell=2+|m|}^{2K-|m|} e\left(\ell_1m \left(y+\frac{(\ell-1)\omega}{2} \right)+\ell_2 m\omega\right),
\end{aligned}
\eeq
In which by \eqref{eq:sumofexp}, with $m=2\tilde{m}$ and $\ell=2\tilde{\ell}$,
\beq\label{eq:Weyl1even}
\begin{aligned}
  &\left|\ \evensumm_{m=1-K}^{K-1}\ \ \evensumm_{\ell=2+|m|}^{2K-|m|} e\left(\ell_1m \left(y+\frac{(\ell-1)\omega}{2} \right)+\ell_2 m\omega\right) \right|\\
=&\left| \sum_{\tilde{m}=-\lfloor\frac{K-1}{2}\rfloor}^{\lfloor \frac{K-1}{2}\rfloor}\ \sum_{\tilde{\ell}=1+|\tilde{m}|}^{K-|\tilde{m}|} 
e\left(\ell_1 \tilde{m} \left(2y+(2\tilde{\ell}-1)\omega \right)+2\ell_2 \tilde{m}\omega\right) \right|\\
\leq & K+2\sum_{\tilde{m}=1}^{\lfloor \frac{K-1}{2}\rfloor} \min{\left(K-2\tilde{m},\ \frac{1}{2\| 2\ell_1 \tilde{m} \omega\|_{\T}} \right)}\\
\leq & K + 2K \sum_{\tilde{m}=1}^{\lfloor \frac{K-1}{2}\rfloor} \min{\left(1,\ \frac{1}{2K\| 2\ell_1 \tilde{m} \omega\|_{\T}} \right)},
\end{aligned}
\eeq
and with $m=2\tilde{m}-1$, $\ell=2\tilde{\ell}-1$, 
\beq\label{eq:Weyl1odd}
\begin{aligned}
&\left|\ \oddsumm_{m=1-K}^{K-1}\ \ \oddsumm_{\ell=2+|m|}^{2K-|m|} e\left(\ell_1m \left(y+\frac{(\ell-1)\omega}{2} \right)+\ell_2 m\omega\right) \right|\\
\leq &\ 2\  \oddsumm_{m=1}^{K-1}\ \ \left|\ \oddsumm_{\ell=2+|m|}^{2K-|m|} e\left(\ell_1m \left(y+\frac{(\ell-1)\omega}{2} \right)+\ell_2 m\omega\right) \right|
\end{aligned}
\eeq
which is further equal to 
\beq\nn 
\begin{aligned}
&2\sum_{\tilde{m}=1}^{\lfloor \frac{K}{2}\rfloor}\ \ \left|\sum_{\tilde{\ell}=\tilde{m}+1}^{K-\tilde{m}+1} 
e\left(\ell_1 (2\tilde{m}-1) \left(y+\tilde{\ell}\omega \right)+\ell_2 (2\tilde{m}-1)\omega\right) \right|\\
\leq & 2\sum_{\tilde{m}=1}^{\lfloor \frac{K}{2}\rfloor} \min{\left(K-2\tilde{m}+1,\ \frac{1}{2\| 2\ell_1 \tilde{m} \omega\|_{\T}}\right)}\\
\leq & 2K \sum_{\tilde{m}=1}^{\lfloor \frac{K}{2}\rfloor} \min{\left(1,\  \frac{1}{2K \| 2\ell_1 \tilde{m} \omega\|_{\T}}\right)}.
\end{aligned}
\eeq
Plugging the estimates of \eqref{eq:Weyl1even} and \eqref{eq:Weyl1odd} into \eqref{eq:Weyl1}, yields
\begin{align*}
\l|\sum_{k=1}^K e\left(\ell_1 \left(ky+\frac{k(k-1)\omega}{2}\right)+\ell_2k\omega \right)\r|^2
\leq K+4K \sum_{m=1}^{\lfloor \frac{K}{2}\rfloor} \min{\left(1,\  \frac{1}{2K \| 2\ell_1 \tilde{m} \omega\|_{\T}}\right)}.
\end{align*}

Hence we have
\beq\label{eq:newS2}
\begin{aligned}
S_2
\leq &\frac{2}{K}\left(2p_1 K+4K \sum_{1\leq |\ell_1|\leq p_1}\ \sum_{m=1}^{\lfloor \frac{K}{2}\rfloor} \min{\left(1, \frac{1}{2K \| 2\ell_1m\omega\|_{\T}}\right)} \right)^{\frac{1}{2}}\\
\leq &\frac{2}{K}\left( 2p_1 K+8K \sum_{\ell_1=1}^{p_1}\sum_{m=1}^{\lfloor \frac{K}{2}\rfloor} \min{\left(1, \frac{1}{2K \| 2\ell_1m\omega\|_{\T}}\right)} \right)^{\frac{1}{2}}\\
\leq &\frac{2}{K}\left( 2p_1 K+8K C^*_{p_1 K} \sum_{j=1}^{p_1 K} \min{\left(1, \frac{1}{2K \| j\omega\|_{\T}}\right)} \right)^{\frac{1}{2}}.
\end{aligned}
\eeq
and similarly
\beq\label{eq:Weyl2}
\begin{aligned}
S_3
\leq &\frac{2}{K}(\ln{(p_2)}+1) \left(2p_1 K+8K C^*_{p_1 K} \sum_{j=1}^{p_1K} \min{\left(1, \frac{1}{2K \| j\omega\|_{\T}}\right)} \right)^{\frac{1}{2}}.
\end{aligned}
\eeq
The constant $C^*_{p_1K}$ comes from over-counting.

Next, we will need to bound $C^*_{p_1 K}$ and $\sum_{j=1}^{p_1K} \min{\left(1, \frac{1}{2K \| j\omega\|_{\T}}\right)}$ separately. 

\subsubsection*{Estimate of $C^*_{p_1K}$}
We first note the following simple bound on $C^*_{p_1K}$
\beq\label{eq:Cp1K}
C^*_{p_1K}\leq \min(p_1, \tau^*(p_1K)),
\eeq
where $\tau^*(p_1 K):=\max_{1\leq n\leq p_1K}\ \tau(n)$ with $\tau(n)$ be the divisor function of $n$.
Standard divisor bound yields the following estimates
\be{lm}\label{lm:boundJacobian}
\beq
\begin{aligned}
\tau^*(p_1K)\leq 
\left\lbrace
\begin{matrix}
(p_1K)^{\frac{1.06602}{\ln\ln(p_1K)}},\\
\\
C(\epsilon)(p_1K)^{\epsilon}.
\end{matrix}
\right.
\end{aligned}
\eeq
The second inequality above holds for any integer $p_1K\geq 1$ with explicit constants $C(\frac{1}{2})=2$ and $C(\frac{1}{8})=42000$.
It also holds for $p_1K\leq 327680000$ with constant $C(\frac{1}{50})=702$.
\e{lm}

Combining \eqref{eq:Cp1K} with Lemma \ref{lm:boundJacobian}, we have that for any $0\leq \alpha\leq 1$,
\beq\label{eq:Cp1K2}
C^*_{p_1K}\leq (C(\epsilon))^{\alpha} (p_1K)^{\alpha \epsilon} p_1^{1-\alpha}.
\eeq
We will only use the case when $\alpha=0$, but we keep this as a reference for the sake of completeness.

\subsubsection*{Estimate of $\sum_{\ell=1}^{p_1K}\min{\left(1, \frac{1}{2K \| \ell\omega\|_{\T}}\right)}$}\

In analogy with \eqref{eq:lsS1}, we will split the term $\sum_{\ell=1}^{p_1K}\min{\left(1, \frac{1}{2K \| \ell\omega\|_{\T}}\right)}$ appearing in~\eqref{eq:Weyl2} as follows: 
\beq\label{def:lsS21S22}
\begin{aligned}
  &\sum_{\ell=1}^{p_1K}\min{\left(1, \frac{1}{2K \| \ell\omega\|_{\T}}\right)}\\
=&\sum_{\ell=1}^{p_1K}\one_{\{\ell:\ \| \ell\omega\|_{\T}< \frac{1}{2K}\}}+\sum_{j=1}^{2^j<2K}\ \sum_{\ell=1}^{p_1K}\one_{\{\ell:\ \frac{2^{j-1}}{2K}
\leq \| \ell\omega\|_{\T}<\frac{2^j}{2K}\}}\ \frac{1}{2K \| \ell\omega\|_{\T}}\\
\leq &\sum_{\ell=1}^{p_1K}\one_{\{\ell:\ \| \ell\omega\|_{\T}< \frac{1}{2K}\}}+\sum_{j=1}^{2^j<2K}\ \sum_{\ell=1}^{p_1K}\one_{\{\ell:\ \frac{2^{j-1}}{2K}\leq \| \ell\omega\|_{\T}<\frac{2^j}{2K}\}}\ \frac{1}{2^{j-1}}\\
=:&S_4+S_5.
\end{aligned}
\eeq

By Corollary \ref{cor:gold1}, if for some $\ell\geq 1$, $\| \ell\omega\|_{\T}<\frac{2^j}{2K}$, then
\beq\label{eq:lsS211}
\ell\geq \frac{2}{2^j 3}K.
\eeq
By Corollary \ref{cor:gold1}, if for some distinct $\ell,\tilde{\ell}\geq 1$, $\max{(\| \ell\omega\|_{\T}, \|\tilde{\ell}\omega\|_{\T})}<\frac{2^j}{K}$, then
\beq\label{eq:lsS212}
|\ell-\tilde{\ell}|\geq \frac{1}{2^j 3}K.
\eeq
Combining \eqref{eq:lsS211} with \eqref{eq:lsS212}, we have
\beq\label{eq:lsS21}
S_4\leq 3 p_1,
\eeq
and
\beq\label{eq:lsS22}
S_5\leq \sum_{j=1}^{2^j<2K} 6 p_1\leq 6(\log_2{K}+1)p_1.
\eeq

In view of  \eqref{eq:newS2}, \eqref{eq:Weyl2}, \eqref{def:lsS21S22}, \eqref{eq:lsS21} and \eqref{eq:lsS22} together, one has 
\beq\label{eq:lsS3final}
\begin{aligned}
S_2
<&\frac{2}{K}\left(2p_1 K+24 p_1K\ C^*_{p_1K} (2\log_2{K}+3) \right)^{\frac{1}{2}}\\
<&8\sqrt{3}(C^*_{p_1K})^{\frac{1}{2}} p_1^{\frac{1}{2}}K^{-\frac{1}{2}}(\log_2{K}+2)^{\frac{1}{2}}\\
<&20 (C_{p_1K}^*)^{\frac{1}{2}}p_1^{\frac{1}{2}}K^{-\frac{1}{2}}(\ln{K)}^{\frac{1}{2}},\ \ \text{for}\ K\geq 23,
\end{aligned}
\eeq
and
\beq\label{eq:lsS2final}
\begin{aligned}
S_3
<&8\sqrt{3} (C^*_{p_1K})^{\frac{1}{2}} (\log{(p_2)}+1)p_1^{\frac{1}{2}}K^{-\frac{1}{2}}(\log_2{K}+2)^{\frac{1}{2}}\\
<& 25 (\ln{p_2}) (C^*_{p_1K})^{\frac{1}{2}} p_1^{\frac{1}{2}}K^{-\frac{1}{2}} (\ln K)^{\frac{1}{2}},\ \ \text{for}\ p_2\geq K\geq 38.
\end{aligned}
\eeq


\subsubsection{Combining $S_1, S_2, S_3$}
Taking $p_1=\lfloor K^{\delta_1} \rfloor$ and $p_2=\lfloor e^{4(\ln{K})^{\delta_2}} \rfloor$, the estimate of $S_1$, namely \eqref{eq:lsS1final2}, becomes
\beq\label{eq:lsS1final3}
S_1 <105 K^{-1} (\ln{K})^{\delta_2+1}.
\eeq
Recall from the preceding that we impose the conditions $K\ge 38$ and $\exp\big(4(\log K)^{\delta_2}\big)\ge K+1$, note that these are our assumptions (ii) and (iii).
The estimate of $S_2$, \eqref{eq:lsS3final}, becomes
\beq\label{eq:lsS3final2}
S_2<20 (C^*_{K^{1+\delta_1}})^{\frac{1}{2}} K^{-\frac{1-\delta_1}{2}} (\ln{K})^{\frac{1}{2}}.
\eeq
The estimate of $S_3$, \eqref{eq:lsS2final}, becomes
\beq\label{eq:lsS2final3}
S_3<100 (C^*_{K^{1+\delta_1}})^{\frac{1}{2}} K^{-\frac{1-\delta_1}{2}} (\ln{K})^{\delta_2+\frac{1}{2}}.
\eeq
Combining \eqref{eq:lsS1final3}, \eqref{eq:lsS3final2}, \eqref{eq:lsS2final3}, \eqref{def:lsS1S2*} with our assumption (i) that 
$C(B_5-m_5)\leq K^{\delta}$, yields
\beq\label{eq:lsS1S2final}
\begin{aligned}
  &\sup_{(x,y)\in \T^2}\l| \frac{1}{K}\sum_{k=1}^K v_3\circ T^k_{\omega}(x,y)-\la v\ra \r|\\
  <&105 K^{-1+\delta} (\ln{K})^{\delta_2+1} +20 C(B_4-m_4) (C^*_{K^{1+\delta_1}})^{\frac{1}{2}} K^{-\frac{1-\delta_1}{2}} (\ln{K})^{\frac{1}{2}}\\ 
  +&100 (C^*_{K^{1+\delta_1}})^{\frac{1}{2}} K^{-\frac{1-\delta_1}{2}+\delta}  (\ln{K})^{\delta_2+\frac{1}{2}}.
\end{aligned}
\eeq
By \eqref{eq:Cp1K2}, with $\alpha=0$, we have
\begin{align*}
  &\sup_{(x,y)\in \T^2}\l| \frac{1}{K}\sum_{k=1}^K v_3\circ T^k_{\omega}(x,y)-\la v\ra \r|\\
    <&105 K^{-1+\delta} (\ln{K})^{\delta_2+1} +20 C(B_4-m_4) K^{-\frac{1-2\delta_1}{2}}  (\ln{K})^{\frac{1}{2}}\\ 
  +&100 K^{-\frac{1-2\delta_1}{2}+\delta}  (\ln{K})^{\delta_2+\frac{1}{2}}.
\end{align*}
By condition (iv) in our statement of the proposition, we have
\begin{align*}
21 K^{-\frac{1}{2}+\delta-\delta_1}(\log{K})^{\delta_2+\frac{1}{2}}+4 C(B_4-m_4)\leq K^{\delta}(\log{K})^{\delta_2}.
\end{align*}
which implies
\beq\label{eq:lsS1S2final3}
\begin{aligned}
\sup_{(x,y)\in \T^2}\l| \frac{1}{K}\sum_{k=1}^K v_3\circ T^k_{\omega}(x,y)-\la v\ra \r| <105 K^{-\frac{1-2\delta_1}{2}+\delta}  (\ln{K})^{\delta_2+\frac{1}{2}}
=:\varepsilon_0.
\end{aligned}
\eeq

Let 
\beq\label{def:t}
t=C(B_4-m_4)  K^{-\frac{1}{2}\delta_1}
\eeq 
in \eqref{eq:lsbady}, then for any $y\notin \mathcal{A}(t)$, 
with \eqref{eq:tildeu1L2} we have,
\beq\label{eq:tildev1+v2L1norm}
\begin{aligned}
\|\frac{1}{K}\sum_{k=1}^K(\tilde{v}_1+v_2)\circ T_{\omega}^k(\cdot ,y)\|_{L^1(\T)}
\leq &\sqrt{2}C(B_4-m_4) K^{-\frac{1}{2}\delta_1}+t\\
=&(\sqrt{2}+1)C(B_4-m_4) K^{-\frac{1}{2}\delta_1}=:\varepsilon_1.
\end{aligned}
\eeq

Recall that $v=\tilde{v}_1+v_2+v_3$.
For any fixed $y\notin \mathcal{A}(t)$, consider the subharmonic function
$$v_y(z):=\frac{1}{K}\sum_{k=1}^K v\circ T_{\omega}^k(z,y) \ \ \text{with}\ \ z\in D_R.$$
This subharmonic function is going to satisfy the following bounds
$$v_y(z)\leq B_4\ \text{for}\ \forall z\in D_R,\ \ \text{and}\ \ v_y(0)\geq m_4.$$
By \eqref{eq:lsS1S2final3} and \eqref{eq:tildev1+v2L1norm}, we know $v_y(x)-\la v\ra $ can be decomposed into two parts, 
one with small $L^\infty$ norm $\varepsilon_0$, the other with small $L^1$ norm $\varepsilon_1$.  
We will choose $\delta_1$ such that $\varepsilon_0\sim \sqrt{\varepsilon_1}$, in the sense that
\beq
K^{-\frac{1-2\delta_1}{2}+\delta}(\log{K})^{\delta_2+\frac{1}{2}}=K^{-\frac{1}{4}\delta_1},
\eeq
which yields
\beq
K^{\delta_1}=\frac{K^{\frac{2}{5}-\frac{4\delta}{5}}}{(\ln{K})^{\frac{2}{5}+\frac{4\delta_2}{5}}},\  \text{ with } 0<\delta_1<\frac{2}{5}-\frac{4\delta}{5}.
\eeq
Then
\beq\label{eqrmk:1}
\varepsilon_0\leq 105  {K^{-\frac{1}{10}+\frac{\delta}{5}}}{(\ln{K})^{\frac{1}{10}+\frac{\delta_2}{5}}},
\eeq
and
\beq\label{eqrmk:3}
\varepsilon_1= (\sqrt{2}+1) C(B_4-m_4) {K^{-\frac{1}{5}+\frac{2\delta}{5}}}{(\ln{K})^{\frac{1}{5}+\frac{2\delta_2}{5}}}.
\eeq

Applying Corollary \ref{cor:T1Splitting} to $v_y-\la v\ra$, we obtain that for 
\beq\label{def:longsumepsilon2}
\varepsilon_4=C_2  {K^{-\frac{1}{10}+\frac{\delta}{5}}}{(\ln{K})^{\frac{1}{10}+\frac{\delta_2}{5}+\delta_3}}
\eeq 
with some constant $C_2>0$, 
\beq\label{eq:longsumlast}
\begin{aligned}
&\sup_{y\notin \mathcal{A}(t)} \l| \left\lbrace x\in\T:\ \l|\frac{1}{K}\sum_{k=1}^K v\circ T_{\omega}^k(x,y)-\la v\ra \r|>\varepsilon_4 \right\rbrace \r|\\
&\qquad\qquad\qquad\qquad \leq 2\sqrt{2} \, \exp \Big(\frac{\pi}{4} \big [\frac{17}{36}+\frac{B_1}{4B_3^2}-\varepsilon_{2}\delta_{0}^{-1}   \big]  \Big),
\end{aligned}
\eeq
where
\begin{align*}
\delta_{0}= \left(472.5 +2B_3(B_4-m_4)\sqrt{(\sqrt{2}+1) C}\right) {K^{-\frac{1}{10}+\frac{\delta}{5}}}{(\ln{K})^{\frac{1}{10}+\frac{\delta_2}{5}}},
\end{align*}
and hence
\begin{align*}
\varepsilon_4 \delta_0^{-1}
=&C_2 \left(472.5+2B_3 (B_4-m_4) \sqrt{(\sqrt{2}+1) C}\right)^{-1} (\log{K})^{\delta_3}\\
\geq &C_2 \left(472.5+3.2 B_3(B_4-m_4) \sqrt{C} \right)^{-1} (\log{K})^{\delta_3}=:\varepsilon_5.
\end{align*}

\eqref{eq:longsumlast} together \eqref{def:t} and \eqref{eq:lsbady} imply
\begin{align*}
&\l| \left\lbrace (x,y)\in \T^2:\ \l|\frac{1}{K}\sum_{k=1}^K v\circ T_{\omega}^k(x,y)-\la v\ra \r|>  \varepsilon_4 \right\rbrace \r|\\
\leq &2\sqrt{2} \, \exp \Big(\frac{\pi}{4} \big [\frac{17}{36}+\frac{B_1}{4B_3^2} -\varepsilon_5\big]  \Big)+|\mathcal{A}(t)|\\
\leq &2\sqrt{2} \, \exp \Big(\frac{\pi}{4} \big [\frac{17}{36}+\frac{B_1}{4B_3^2} -\varepsilon_5 \big]  \Big)\\
&\qquad +\sqrt{2} (C(B_4-m_4))^{-1} {K^{\frac{1}{5}-\frac{2\delta}{5}}}{(\ln{K})^{-\frac{1}{5}-\frac{2\delta_2}{5}}} \exp{\left(-2(\ln{K})^{\delta_2}\right)},
\end{align*}
as claimed.
$\hfill{} \Box$

\section{Multi-scale estimates}\label{sec:MS}

In this section we commence with the inductive arguments in our multi-scale Lyapunov exponent machinery. In analogy with~\cite{GS, B1} we proceed by combining the large deviation estimates with the avalanche principle. We begin with the basic induction step, which provides a lower bound for the Lyapunov exponent at a large scale from information on the Lyapunov exponents at smaller scales, in combination with level-set estimates. In Proposition \ref{prop:main8}, which is the main result of this section, we will also invoke  the quantitative control on the Birkhoff averages over the skew shift from the previous section in order to derive large deviation estimates at the larger scale. 

The following subsection will serve as an abstract multi-scale scheme to providing a lower bound of the (maximal) Lyapunov exponent, assuming large deviation estimates. 
In our application to the skew shift, the large deviation estimates will come from Proposition~\ref{lm:longsumwithtrivialbdd}, see Section \ref{sec:applytoskew}.

\subsection{Abstract multi-scale scheme}\

\subsubsection{Lyapunov exponent}
Let $(X, \mu, S)$ be an ergodic dynamical system. 
A linear cocycle over $(X, \mu, S)$ is a skew-product map 
\beq\nn
F_A: X\times \R^d \rightarrow X\times \R^d,
\eeq
given by 
\beq\nn 
X\times \R^d \ni (x, v)\rightarrow (Sx, A(x) v)\in X\times \R^d,
\eeq
where
\beq\nn
A: X\rightarrow \mathrm{SL}_d(\R)
\eeq
is a measurable function.

The forward iterates $F_A^n$ of a linear cocycle $F_A$ are given by $F_A^n(x, v)=(S^n x, M_n(x)v)$, where
\beq\nn
M_n(x):=A(S^{n-1}x)\cdots A(Sx)A(x)\ \ \text{(}n\in \N \text{)}.
\eeq

A linear cocycle $A$ is said to be $\mu$-integrable if 
\beq\nn
\int_{X} \log{\|A(x)\|}\ \mathrm{d}\mu<+\infty.
\eeq
Due to the fact that norms are sub-multiplicative with respect to matrix products, 
the sequence of functions $\log{\|A^{(n)}(x)\|}$ are subadditive. 
The F\"{u}rstenberg-Kesten theorem (or Kingman's ergodic theorem) implies that for a $\mu$-integrable linear cocycle, the following $\mu$-a.e. limit exists
\beq\nn
L(A):=\lim_{n\rightarrow\infty} \frac{1}{n}\log{\|M_n(x)\|},
\eeq
and it is called the (maximal) Lyapunov exponent of $A$.
Moreover,
\beq\nn
L(A):=\lim_{n\rightarrow\infty}\int_X \frac{1}{n}\log{\|M_n(x)\|}\, \mu(\d x)=\inf_{n\geq 1}\int_X \frac{1}{n}\log{\|M_n(x)\|}\, \mu(\d x).
\eeq
We point out the since $A\in \mathrm{SL}_d (\R)$, we have $\|M_n(x)\|\geq 1$, hence $L(A)\geq 0$.

\subsubsection{Inductive scheme}\

Let us denote 
\beq\nn
L_n(A):=\int_X \frac{1}{n}\log{\|M_n(x)\|}\ \mathrm{d}\mu(x).
\eeq
For simplicity, we may omit the dependence of $L(A), L_n(A)$ on $A$, and simply write $L$ and $L_n$.

Let us further assume that there exists a constant $C_3>0$, such that 
\beq\label{def:C1Linfty}
\frac{1}{n}\log{\|M_n(x)\|}\leq C_3<+\infty,
\eeq 
for $\mu$-a.e. $x$, uniformly in $n$.
We point out that in our application to the skew-shift model, $C_3$ can be taken as $U(\lambda, 1)$, see \eqref{def:un=vn} and \eqref{eq:B7Herman}.

\begin{defn}\label{def:Bn}
In our multi-scale scheme, we quantify the failure of the F\"{u}rstenberg-Kesten theorem via the following sets $\mathcal{B}_n$:
\beq\nn
\mathcal{B}_n:=\left\lbrace x\in X:\ \l| \frac{1}{n}\log{\|M_n(x)\|}-L_n \r|>\frac{1}{10} L_n \right\rbrace.
\eeq
\end{defn}

The lemma below shows how to inductively obtain estimates of $L_N$ at a larger scale $N$, based on information at a smaller scale $n$.
The key ingredient is the Avalanche Principle, Theorem \ref{thm:AP}.

\be{lm}\label{lm:induction'}
Let $n, N/n\in \N$ be positive integers, and $\delta\in (0,1/2)$.
Let $C_3$ be as in \eqref{def:C1Linfty}.
Assume the following three conditions:
\begin{enumerate}
\item[(a).]
\beq\nn
nL_n\geq 7,
\eeq
\item[(b).]
\beq\nn
L_n-L_{2n}\leq \frac{1}{8}L_n,
\eeq
\item[(c).] 
\beq\nn
\begin{aligned}
\max{(\mu(\mathcal{B}_n), \mu(\mathcal{B}_{2n}))}\leq N^{-\frac{12}{5}+\frac{4\delta}{5}}.
\end{aligned}
\eeq
\end{enumerate}
Then we have
\beq\label{conclusion:LN'}
L_N\geq  L_n-\left(2-\frac{2n}{N}\right)(L_n-L_{2n})-\frac{11}{n}e^{-\frac{1}{2} n L_n}- 8C_3 N^{-\frac{7}{5}+\frac{4\delta}{5}},
\eeq
and
\beq\label{conclusion:LNL2N'}
L_N-L_{2N}\leq \frac{n}{N} (L_n-L_{2n})+\frac{22}{n}e^{-\frac{1}{2} n L_n}+ 24 C_3 N^{-\frac{7}{5}+\frac{4\delta}{5}}.
\eeq
\e{lm}

\subsubsection{Multi-scale scheme}\ 

Lemma~\ref{lm:sequencescales}
below shows how information on a sequence of larger and larger scales determines the limit $L$.

\be{lm}\label{lm:sequencescales}
Let $\delta\in (0,1/2)$ be a constant, and $C_3$ be as in \eqref{def:C1Linfty}.
Let $\{N_m\}_{m=0}^{\infty} \in \N$ be a sequence of positive integers, such that $10\leq N_{m}/N_{m-1}\in \N$ for $1\leq m$.
Assume that the following hold for an integer $j\geq 0$ (note that (2)-(4) below are empty conditions for $j=0$):
\begin{enumerate}
\item \[N_0 L_{N_0}\geq 7, \ \ \text{and}\ \ L_{N_0}-L_{2N_0}\leq \frac{1}{8}L_{N_0},\]
\item \[\sum_{m=0}^{j-1}\frac{1}{N_m}e^{-\frac{1}{2}N_m L_{N_m}}<\frac{1}{512}L_{N_0},\]
\item \[\sum_{m=1}^{j} N_m^{-\frac{7}{5}+\frac{4\delta}{5}}<\frac{1}{1280 C_3}L_{N_0},\]
\item \[\max{(\mu(\mathcal{B}_{N_m}), \mu(\mathcal{B}_{2N_m}))}\leq N_{m+1}^{-\frac{12}{5}+\frac{4\delta}{5}},\ \ \text{for }0\leq m\leq j-1.\]
\end{enumerate}
Then we have the following four estimates for $j\geq 0$.

First,
\beq\label{eq:LNj}
\begin{aligned}
L_{N_j}
\geq &L_{N_0}-\left(2-\frac{2N_0}{N_j}\right)(L_{N_0}-L_{2N_0})\\
-&\sum_{m=1}^{j}\left(\frac{11}{N_{m-1}}e^{-\frac{1}{2}N_{m-1}L_{N_{m-1}}}+8C_3 N_{m}^{-\frac{7}{5}+\frac{4\delta}{5}}\right)\\
-&\sum_{m=1}^{j-1}\left(2-\frac{2N_{m}}{N_j}\right)\left(\frac{22}{N_{m-1}}e^{-\frac{1}{2}N_{m-1}L_{N_{m-1}}}+24C_3 N_{m}^{-\frac{7}{5}+\frac{4\delta}{5}}\right),
\end{aligned}
\eeq
in which $\sum_{m=1}^0=\sum_{m=1}^{-1}:\equiv 0$.

Second,
\beq\label{eq:LNjL2Nj2}
L_{N_j}-L_{2N_j}\leq \frac{N_0}{N_j}(L_{N_0}-L_{2N_0})+\sum_{m=1}^{j}\frac{N_m}{N_j} \left(\frac{22}{N_{m-1}}e^{-\frac{1}{2} N_{m-1} L_{N_{m-1}}}+24 C_3 N_{m}^{-\frac{7}{5}+\frac{4\delta}{5}} \right),
\eeq
in which $\sum_{m=1}^0:\equiv 0$.

Third,
\beq\label{eq:LNjL2Nj}
L_{N_j}-L_{2N_j}\leq \frac{1}{8}L_{N_j}
\eeq

Fourth,
\beq\label{eq:LNjLN0}
L_{2N_j}\geq \frac{1}{2}L_{N_0},\ \ \text{and}\ \ N_j L_{N_j}\geq 7.
\eeq
\e{lm}

\subsection{Application to the skew-shift model}\label{sec:applytoskew}
The two cornerstones of the abstract multi-scale scheme are:
\begin{itemize}
\item Initial scale $N_0$ estimates, including (1). $N_0 L_{N_0}\geq 7$, (2). $L_{N_0}-L_{2N_0}\leq \frac{1}{8}L_{N_0}$ and (3).  Large deviation estimates of $\mu(\mathcal{B}_{N_0})$ and $\mu(\mathcal{B}_{2N_0})$.
\item Large deviation estimates of $\mu(\mathcal{B}_{N_j})$ and $\mu(\mathcal{B}_{2N_j})$ for $j\geq 1$.
\end{itemize}
In this subsection, we will present a machinery that inductively provides large deviation estimates for scales $N_j$, $j\geq 1$,
thus reducing the problem to the initial scale only.
The key ingredients are the Avalanche Principle and the quantitative control of the ergodic averages of plurisubharmonic functions over a skew-shift orbit, Proposition~\ref{lm:longsumwithtrivialbdd}.

Let us recall some notations from Section \ref{sec:Herman}: 
$u_n(\lambda, E; x,y)=\frac{1}{n}\log{\|M_n(\lambda, E; x,y)\|}$, 
and $v_n(\lambda, E; z, w)$ be the complexification of $u_n$ from $\T^2$ to $\C^2$, as in \eqref{def:un=vn}.
The constant $U(\lambda, 1)$, as in \eqref{eq:B7Herman}, is a uniform (in $n$ and $E$) $L^{\infty}$ upper bound on $u_n(\lambda, E; x,y)$.
For simplicity, we will omit the dependence of $u_n(x,y)$, $v_n(z,w)$, $\log{\|M_n(x,y)\|}$ and $L_n$ on $\lambda, E$, since $\lambda$ will be fixed, and our estimates are uniform in $E\in [-2-2\lambda, 2+2\lambda]$.
Recall also from Lemma \ref{lm:Herman} with $R_3=R$ that the bounds with respect to $v_n$ satisfy
\beq\label{eq:B56m56}
\begin{aligned}
B_4-m_4&=2\log{R}+U(\lambda, R)-\log{\lambda},\\
B^{(n)}_5-m^{(n)}_5&=(n+1)\log{R}+U(\lambda, R)-\log{\lambda}.
\end{aligned}
\eeq
Let us finally also recall the constants $B_3(R, R_1, R_2)$ as in \eqref{def:constantsB3B4}, $C(R,R_1,R_2)$ as in \eqref{def:constantCRR1R2},
and $C_0(R,R_1,R_2)$ is as in \eqref{eq:expintegrable}.
In the following we will write $B_3, C, C_0$ for simplicity.

\be{prop}\label{prop:main8}
Let $\omega=\frac{\sqrt{5}-1}{2}$ be the golden ratio.
Let $\delta\in (0, 1/2)$ and \[\delta_2,\delta_3, \delta_4, C_2, C_4, C_5>0\] be constants. 
{{Let $n, N\in \N$ be positive integers and assume that $n$ divides $N$.}}
In addition to the conditions {{(a)-(c)}} in Lemma~\ref{lm:induction'} and Definition~\ref{def:epssmall}, assume further that
{{the following  properties hold for both $\tilde{N}=N$ and $2N$:}}
\begin{enumerate}[label=(\Roman*).]
\item $C\left((2n+1)\log{R}+U(\lambda, R)-\log{\lambda}\right) \leq \tilde{N}^{\delta}$,
\item $\tilde{N}\geq 38$,
\item $\exp{\Big(4(\log{\tilde{N}})^{\delta_2}\Big)}\geq \tilde{N}+1$,
\item $21 \tilde{N}^{-\frac{9}{10}+\frac{9}{5}\delta} {(\ln{\tilde{N}})^{\frac{9}{10}+\frac{9}{5}\delta_2}}+4C(B_4-m_4) \leq \tilde{N}^{\delta}(\log{\tilde{N}})^{\delta_2}$,
\item $2n \tilde{N}^{-1}(L_n-L_{2n})+8U(\lambda,1)\tilde{N}^{-\frac{7}{5}+\frac{4\delta}{5}}+5U(\lambda, 1)n \tilde{N}^{-1}<C_2 \tilde{N}^{-\frac{1}{10}+\frac{\delta}{5}}{(\ln{\tilde{N}})^{\frac{1}{10}+\frac{\delta_2}{5}+\delta_3}}$,
\item $22 n^{-1} \exp{\left(-nL_n/2\right)}<C_2 \tilde{N}^{-\frac{1}{10}+\frac{\delta}{5}}{(\ln{\tilde{N}})^{\frac{1}{10}+\frac{\delta_2}{5}+\delta_3}}$,
\item $4\sqrt{2} \, \exp \Big(\frac{\pi}{4} \big [\frac{17}{36}+\frac{B_1}{4B_3^2} -C_2 (472.5 +3.2 B_3(B_4-m_4) \sqrt{C})^{-1} (\log{\tilde{N}})^{\delta_3}  \big]  \Big)\leq \tilde{N}^{-\frac{7}{5}+\frac{4\delta}{5}}$,
\item $2\sqrt{2} (C(B_4-m_4))^{-1} {\tilde{N}^{\frac{1}{5}-\frac{2\delta}{5}}}{(\ln{\tilde{N}})^{-\frac{1}{5}-\frac{2\delta_2}{5}}} \exp{\left(-2(\ln{\tilde{N}})^{\delta_2}\right)}\leq \tilde{N}^{-\frac{7}{5}+\frac{4\delta}{5}}$,
\item $\tilde{N}>(\log{R}+U(\lambda,R)-\log{\lambda})(\log{R})^{-1}$,
\item $C_4 (\log{\tilde{N}})^{\delta_4}>4$,
\item $C_5 (\log{\tilde{N}})^{\delta_4}>C_4$.
\end{enumerate}
Then  {{the following holds for both $\tilde{N}=N$ and $2N$:}}
\beq\nn
\begin{aligned}
&\l| \left\lbrace (x,y)\in \T^2:\ |v_{\tilde{N}}(x,y)-L_{\tilde{N}} |>C_2 C_5  {\tilde{N}^{-\frac{1}{10}+\frac{\delta}{5}}}{(\ln{\tilde{N}})^{\frac{1}{10}+\frac{\delta_2}{5}+\delta_3+2\delta_4}} \right\rbrace \r|\\
\leq &2(2C_0)^{\frac{1}{2}} \exp{\left(\frac{-\pi C_2 C_4 (\ln{\tilde{N}})^{\delta_4}}{144 C_2 +{{48 B_3\sqrt{2U(\lambda,1) (B_4-m_4)}}} (\log{\tilde{N}})^{-\frac{1}{10}-\frac{\delta_2}{5}-\delta_3}}\right)}\\
&+C_0\exp{\left(\frac{-\pi C_2 C_5  {(\ln{\tilde{N}})^{\delta_4}}}{18 C_2 C_4 +{{96 B_3 U(\lambda, 1)\sqrt{\log{R}}}} {(\ln{\tilde{N}})^{-\frac{1}{10}-\frac{\delta_2}{5}-\delta_3-\delta_4}}} \right)}.
\end{aligned}
\eeq
\e{prop}
\begin{rmk}
Note that our conditions ($\mathrm{I}$)-($\mathrm{IV}$) correspond to (i)-(iv) of Proposition~\ref{lm:longsumwithtrivialbdd}.
In particular, ($\mathrm{I}$) is (i) of Proposition \ref{lm:longsumwithtrivialbdd} with $B_5^{(2n)}-m_5^{(2n)}$ given in~\eqref{eq:B56m56}.
\end{rmk}

\subsection{Proofs}\

Before proving Lemma \ref{lm:induction'} and Proposition \ref{prop:main8}, we will first give a quick proof of Lemma~\ref{lm:sequencescales} based on Lemma \ref{lm:induction'}.
\subsection*{Proof of Lemma \ref{lm:sequencescales}}
For $m\geq 1$, let us denote 
\beq\nn
\begin{aligned}
\alpha_m:=&{11}{N_{m-1}^{-1}}e^{-\frac{1}{2} N_{m-1} L_{N_{m-1}}}+ 8C_3 N_m^{-\frac{7}{5}+\frac{4\delta}{5}},\\
\beta_m:=&{22}{N_{m-1}^{-1}}e^{-\frac{1}{2} N_{m-1} L_{N_{m-1}}}+ 24 C_3 N_m^{-\frac{7}{5}+\frac{4\delta}{5}}.
\end{aligned}
\eeq
Note that in terms of $\alpha$ and $\beta$, our conditions~(2) and (3) in the statement of the lemma become
\beq\label{eq:alphaupper}
\sum_{m=1}^{j}\alpha_m\leq \frac{11}{512}L_{N_0}+\frac{8}{1280}L_{N_0}=\frac{71}{2560}L_{N_0},
\eeq
and 
\beq\label{eq:betaupper}
\sum_{m=1}^{j} \beta_m\leq \frac{22}{512}L_{N_0}+\frac{24}{1280}L_{N_0} =\frac{79}{1280}L_{N_0}.
\eeq

Our proof is based on induction on~$j$.
Note that for the induction base case $j=0$: 
\eqref{eq:LNjL2Nj2}, \eqref{eq:LNjL2Nj} and \eqref{eq:LNjLN0} follow directly from condition (1). 
\eqref{eq:LNj} follows from the fact that $L_{N_0}-L_{2N_0}\geq 0$.

Now let us suppose Lemma \ref{lm:sequencescales} holds for $j=J$ for some $J\geq 0$.
Note that conditions (2)-(4) with $j=J+1$ already imply those with $j=J$.
Hence by our inductive assumption, \eqref{eq:LNj}, \eqref{eq:LNjL2Nj2}, \eqref{eq:LNjL2Nj} and \eqref{eq:LNjLN0} hold for $j=J$, whence
\beq\label{eq:LNJ}
\begin{aligned}
L_{N_J}\geq L_{N_0}-\left(2-\frac{2N_0}{N_J}\right)(L_{N_0}-L_{2N_0})-\sum_{m=1}^{J} \alpha_{m} -\sum_{m=1}^{J-1}\left(2-\frac{2N_{m}}{N_J}\right) \beta_{m},
\end{aligned}
\eeq
\beq\label{eq:LNJL2NJ}
L_{N_{J}}-L_{2N_{J}}\leq \frac{N_0}{N_J}(L_{N_0}-L_{2N_0})+\sum_{m=1}^J \frac{N_m}{N_J}\beta_m,
\eeq
\beq\label{eq:LNJL2NJ2}
L_{N_J}-L_{2N_J}\leq \frac{1}{8}L_{N_J},
\eeq
and 
\beq\label{eq:LNJLN0}
N_J L_{N_J}\geq 7.
\eeq

Note that \eqref{eq:LNJL2NJ2}, \eqref{eq:LNJLN0} and our condition (4) in the statement of the lemma with $m=J$ verify the conditions of Lemma~\ref{lm:induction'} for $n=N_J$ and $N=N_{J+1}$.
Therefore Lemma~\ref{lm:induction'} implies
\beq\label{eq:LNJ1}
L_{N_{J+1}}\geq  L_{N_{J}}-\left(2-\frac{2N_{J}}{N_{J+1}}\right)(L_{N_{J}}-L_{2N_{J}}) -\alpha_{J+1},
\eeq
and
\beq\label{eq:LNJ+1L2NJ+1}
\begin{aligned}
L_{N_{J+1}}-L_{2N_{J+1}}
\leq &\frac{N_{J}}{N_{J+1}} (L_{N_{J}}-L_{2N_{J}})+\beta_{J+1}
\end{aligned}
\eeq

Plugging \eqref{eq:LNJL2NJ} into \eqref{eq:LNJ+1L2NJ+1}, we obtain
\beq\label{eq:J+1difference}
\begin{aligned}
L_{N_{J+1}}-L_{2N_{J+1}} \leq &\frac{N_0}{N_{J+1}}(L_{N_0}-L_{2N_0})+\frac{N_{J}}{N_{J+1}}\left(\sum_{m=1}^{J}\frac{N_m}{N_{J}}\beta_m\right)+\beta_{J+1}\\
=& \frac{N_0}{N_{J+1}}(L_{N_0}-L_{2N_0})+\sum_{m=1}^{J+1}\frac{N_m}{N_{J+1}}\beta_{m},
\end{aligned}
\eeq
this proves \eqref{eq:LNjL2Nj2} for $j=J+1$.

Plugging \eqref{eq:LNJ} and \eqref{eq:LNJL2NJ} with $j=J+1$ into \eqref{eq:LNJ1}, we have
\beq\label{eq:J+1lb}
\begin{aligned}
L_{N_{J+1}}\geq &L_{N_0}
-\left(2-\frac{2N_0}{N_J}\right)(L_{N_0}-L_{2N_0})-\sum_{m=1}^{J}\alpha_{m}-\sum_{m=1}^{J-1}\left(2-\frac{2N_{m}}{N_J}\right)\beta_{m}\\
&\ \ \ \ \ -\left(2-\frac{2N_J}{N_{J+1}}\right)\left(\frac{N_0}{N_J}(L_{N_0}-2L_{2N_0})+\sum_{m=1}^{J}\frac{N_m}{N_J} \beta_m\right)-\alpha_{J+1}\\
=&L_{N_0}-\left(2-\frac{2N_0}{N_{J+1}}\right)(L_{N_0}-L_{2N_0})-\sum_{m=1}^{J+1}\alpha_{m}-\sum_{m=1}^{J}\left(2-\frac{2N_{m}}{N_{J+1}}\right)\beta_{m}.
\end{aligned}
\eeq
This proves \eqref{eq:LNj} for $j=J+1$.

Combining \eqref{eq:J+1difference}, \eqref{eq:J+1lb} with the fact that $10\leq N_{j+1}/N_j$ for any $j\geq 0$, yields 
\beq\nn
\begin{aligned}
&8(L_{N_{J+1}}-L_{2N_{J+1}})-L_{N_{J+1}}\\
\leq &-L_{N_0}+\left(2+\frac{6N_0}{N_{J+1}}\right)(L_{N_0}-L_{2N_0})+\sum_{m=1}^{J+1}\alpha_{m}+\sum_{m=1}^{J+1}\left(2+\frac{6N_m}{N_{J+1}}\right)\beta_m\\
\leq &-L_{N_0}+\left(2+\frac{6}{10}\right)(L_{N_0}-L_{2N_0})+\sum_{m=1}^{J+1}\alpha_{m}+8\sum_{m=1}^{J+1}\beta_m.
\end{aligned}
\eeq
Using \eqref{eq:alphaupper} and \eqref{eq:betaupper},  and the fact that $L_{N_0}-L_{2N_0}\leq \frac{1}{8}L_{N_0}$, we conclude that 
\beq\label{eq:LJ+1L2J+1/8}
8(L_{N_{J+1}}-L_{2N_{J+1}})-L_{N_{J+1}}\leq -\frac{393}{2560}L_{N_0}<0.
\eeq
This proves \eqref{eq:LNjL2Nj} for $j=J+1$.

By \eqref{eq:J+1lb} and the fact that $N_{m+1}\geq 10N_m$  for any $m\geq 0$, we have
\beq\label{eq:LJ+1lower}
L_{N_{J+1}}\geq L_{N_0}-\left(2-\frac{1}{5}\right)(L_{N_0}-L_{2N_0})-\sum_{m=1}^{J+1}\alpha_m -\left(2-\frac{1}{5}\right)\sum_{m=1}^{J}\beta_m.
\eeq
Plugging \eqref{eq:alphaupper} and \eqref{eq:betaupper} with $j=J+1$ into \eqref{eq:LJ+1lower}, and using that $L_{N_0}-L_{2N_0}\leq \frac{1}{8}L_{N_0}$, yields
\beq\label{eq:LJ+1lower'}
L_{N_{J+1}}\geq \frac{8143}{12800}L_{N_0},
\eeq
which also implies $N_J L_{N_J}\geq 7$.
\eqref{eq:LNjLN0} with $j=J+1$ then follows from \eqref{eq:LJ+1L2J+1/8} and \eqref{eq:LJ+1lower'}, indeed,
\beq\nn
L_{2N_{J+1}}\geq \frac{7}{8}L_{N_{J+1}}\geq \frac{7}{8} \times \frac{8143}{12800}L_{N_0}\geq \frac{1}{2}L_{N_0},
\eeq
as desired.
$\hfill{} \Box$

\subsection*{Proof of Lemma \ref{lm:induction'}}
Let ${{\tilde{N}}}=N$ or $2N$.
Let us define
\beq\label{def:setfailsAP'}
\mathcal{B}^{({{\tilde{N}}})}:=\left(\bigcup_{j=0}^{{{\tilde{N}}}-1} S^{-j}\mathcal{B}_n\right) \bigcup \left(\bigcup_{j=0}^{{{\tilde{N}}}-n-1} S^{-j}\mathcal{B}_{2n}\right).
\eeq
We have the following measure estimates for $\mathcal{B}^{({{\tilde{N}}})}$ by condition~(c) of the statement of the lemma,
\beq\label{est:setfailsAP'}
\mu(\mathcal{B}^{(N)})\leq \frac{2N-n}{N^{\frac{12}{5}-\frac{4\delta}{5}}}\leq 2N^{-\frac{7}{5}+\frac{4\delta}{5}},\ \ \text{and}\ \ \mu(\mathcal{B}^{(2N)})\leq 4N^{-\frac{7}{5}+\frac{4\delta}{5}}.
\eeq
Taking any $x\notin \mathcal{B}^{({{\tilde{N}}})}$, by our definitions of $\mathcal{B}_n$ and $\mathcal{B}_{2n}$, we have
\beq\label{eq:induction1}
e^{\frac{11}{10} n L_n}\geq \|M_n(S^j x)\|\geq e^{\frac{9}{10} n L_n}=:\kappa^{-\frac{1}{2}},\ \ \text{for any }0\leq j\leq {{\tilde{N}}}-1, 
\eeq
and 
\beq\label{eq:induction2}
\|M_{2n}(S^j x)\|\geq e^{\frac{9}{5} n L_{2n}},\ \ \text{for any }0\leq j\leq {{\tilde{N}}}-n-1.
\eeq
Hence for any $0\leq j\leq {{\tilde{N}}}-n-1$, \eqref{eq:induction1} and \eqref{eq:induction2} imply that 
\beq\nn
\frac{\|M_{2n}(S^j x)\|}{\|M_n(S^{j+n} x)\|\ \|M_n(S^j x)\|}\geq \exp{\left(2n(L_{2n}-L_n)-\frac{1}{5} n(L_n+L_{2n})\right)}=:\epsilon.
\eeq
We now need to verify the assumptions of the Avalanche Principle, Theorem \ref{thm:AP}.
First, by sub-additivity of $\log{\|M_n(x)\|}$, we have $L_{2n}\leq L_n$. 
This together with our assumptions~(a) and (b) yield
\beq
\epsilon= \exp{\left(\frac{9}{5} n L_{2n}- \frac{11}{5} nL_n\right)}\leq e^{-\frac{2}{5}nL_n}\leq e^{-\frac{14}{5}}<\frac{1}{10},
\eeq
and 
\beq\nn
\kappa\epsilon^{-2}=\exp{\left(-2nL_n+\frac{3}{5} nL_n+\frac{2}{5} nL_{2n}+4n(L_n-L_{2n})\right)}\leq  e^{-\frac{1}{2} nL_n}\leq e^{-\frac{7}{2}}<\frac{1}{10}.
\eeq
Applying Theorem \ref{thm:AP} to $x\notin \mathcal{B}^{({{\tilde{N}}})}$, we conclude that for each $0\leq k\leq n-1$, 
\beq\nn
\begin{aligned}
&\frac{1}{{{\tilde{N}}}} \l| \log{\|M_{{\tilde{N}}}(S^{k} x)\|}+
\sum_{j=1}^{\frac{{{\tilde{N}}}-2n}{n}} \log{\|M_n(S^{jn+k} x)\|} \right.\\
&\qquad\qquad\qquad \left. - \sum_{j=0}^{\frac{{{\tilde{N}}}-2n}{n}}\log{\|M_{2n}(S^{jn+k} x)\|} \r| \leq \frac{11}{n}\kappa \epsilon^{-2}\leq \frac{11}{n} e^{-\frac{1}{2} nL_n}.
\end{aligned}
\eeq
Summing over $k\in [0,n-1]$ and dividing by $n$, and finally applying  the  triangle inequality yields
\beq\label{eq:applyAP2'}
\begin{aligned}
\l| \frac{1}{n} \sum_{k=0}^{n-1} \frac{1}{{{\tilde{N}}}} \log{\|M_{{\tilde{N}}}(S^k x)\|}+\frac{1}{{{\tilde{N}}}}\sum_{j=n}^{{{\tilde{N}}}-n-1} \frac{1}{n}\log{\|M_n(S^j x)\|} \right.\\
\left. -\frac{2}{{{\tilde{N}}}}\sum_{j=0}^{{{\tilde{N}}}-n-1}\frac{1}{2n} \log{\|M_{2n}(S^j x)\|} \r|\leq \frac{11}{n} e^{-\frac{1}{2} nL_n}.
\end{aligned}
\eeq
Integrating over $x\in X$, and using our definition of $C_3$ \eqref{def:C1Linfty}, we infer the following due to \eqref{est:setfailsAP'}:
\beq\label{eq:LNLnL2n'} 
\begin{aligned}
\left| L_N+\frac{N-2n}{N} L_n-2\frac{N-n}{N} L_{2n} \right|
\leq &\frac{11}{n}e^{-\frac{1}{2} nL_n}+4C_3  \mu(\mathcal{B}^{(N)})\\
<& \frac{11}{n}e^{-\frac{1}{2} nL_n}+ 8C_3 N^{-\frac{7}{5}+\frac{4\delta}{5}},
\end{aligned}
\eeq
and
\beq\label{eq:L2NLnL2n'}
\begin{aligned}
\left| L_{2N}+\frac{N-n}{N} L_n-\frac{2N-n}{N} L_{2n} \right|\leq &\frac{11}{n}e^{-\frac{1}{2} nL_n}+4C_3 \mu(\mathcal{B}^{(2N)})\\
<&\frac{11}{n}e^{-\frac{1}{2} nL_n}+ 16C_3 N^{-\frac{7}{5}+\frac{4\delta}{5}}.
\end{aligned}
\eeq
From \eqref{eq:LNLnL2n'}, we conclude that 
\begin{align*}
L_N\geq L_n-\left(2-\frac{2n}{N}\right)(L_n-L_{2n}) -\frac{11}{n}e^{-\frac{1}{2} nL_n}- 8C_3 N^{-\frac{7}{5}+\frac{4\delta}{5}}.
\end{align*}
This proves \eqref{conclusion:LN'}.

Taking the difference between \eqref{eq:LNLnL2n'} and \eqref{eq:L2NLnL2n'}, we obtain
\begin{align*}
L_N-L_{2N}\leq \frac{22}{n}e^{-\frac{1}{2} nL_n}+ 24C_3 N^{-\frac{7}{5}+\frac{4\delta}{5}} +\frac{n}{N} (L_n-L_{2n}).
\end{align*}
This proves \eqref{conclusion:LNL2N'}.

\subsection*{Proof of Proposition \ref{prop:main8}}
This will be a continuation of the proof of Lemma \ref{lm:induction'}.
Note that all the constants $C_3$'s will be replaced by $U(\lambda, 1)$.
{{Let ${{\tilde{N}}}$ be either $N$ or $2N$.}}
Let us consider the first term in \eqref{eq:applyAP2'}: 
\beq
\begin{aligned}
       &\l| \frac{1}{n} \sum_{k=0}^{n-1} \frac{1}{{{\tilde{N}}}} \log{\|M_{{{\tilde{N}}}}(T^k_{\omega}(x,y))\|}-\frac{1}{{{\tilde{N}}}}\log{\|M_{{{\tilde{N}}}}(x,y)\|}\r| \\
\leq &\frac{1}{n{{\tilde{N}}}} \sum_{k=0}^{n-1} \l| \log{\|M_{{{\tilde{N}}}}(T^k_{\omega}(x,y))\|}-\log{\|M_{{{\tilde{N}}}}(x,y)\|}\r|\\
\leq &\frac{1}{n{{\tilde{N}}}} \sum_{k=0}^{n-1} \left( \log{\|M_k(x,y)}+\log{\|M_k(T_{\omega}^{{{\tilde{N}}}}(x,y))\|} \right)\\
\leq &\frac{1}{n{{\tilde{N}}}} \sum_{k=0}^{n-1} 2k U(\lambda,1)\\
\leq &\frac{U(\lambda, 1) n}{{{\tilde{N}}}}.
\end{aligned}
\eeq
Hence \eqref{eq:applyAP2'} leads to
\beq\label{eq:applyAP3'}
\begin{aligned}
\l| \frac{1}{{{\tilde{N}}}}\log{\|M_{{{\tilde{N}}}}(x,y)\|}+\frac{1}{{{\tilde{N}}}}\sum_{j=n}^{{{\tilde{N}}}-n-1} \frac{1}{n}\log{\|M_n(T^j_{\omega}(x,y))\|} \right.\\
\left. -\frac{2}{{{\tilde{N}}}}\sum_{j=0}^{{{\tilde{N}}}-n-1}\frac{1}{2n} \log{\|M_{2n}(T^j_{\omega}(x,y))\|} \r|\leq \frac{11}{n}e^{-\frac{1}{2} nL_n}+\frac{U(\lambda, 1)n}{{{\tilde{N}}}},
\end{aligned}
\eeq
holds for $(x,y)\notin \mathcal{B}^{({{\tilde{N}}})}$.  This implies
\beq\label{eq:applyAP4'}
\begin{aligned}
\l| \frac{1}{{{\tilde{N}}}}\log{\|M_{{{\tilde{N}}}}(x,y)\|}+\frac{1}{{{\tilde{N}}}}\sum_{j=0}^{{{\tilde{N}}}-1} \frac{1}{n}\log{\|M_n(T^j_{\omega}(x,y))\|} \right.\\
\left. -\frac{2}{{{\tilde{N}}}}\sum_{j=0}^{{{\tilde{N}}}-1}\frac{1}{2n} \log{\|M_{2n}(T^j_{\omega}(x,y))\|} \r|\leq \frac{11}{n}e^{-\frac{1}{2} nL_n}+\frac{5U(\lambda, 1)n}{{{\tilde{N}}}}.
\end{aligned}
\eeq
Now we will apply Proposition \ref{lm:longsumwithtrivialbdd} to $v_n$ and $v_{2n}$ {{with $K={{\tilde{N}}}$}}. 
Note that conditions ($\mathrm{I}$)-($\mathrm{IV}$) ensure the applicability of that proposition.
Therefore following \eqref{def1:epsilon45}, we define
\beq\label{def:epsilon45}
\begin{aligned}
\varepsilon_4&=C_2  {{{\tilde{N}}}^{-\frac{1}{10}+\frac{\delta}{5}}}{(\ln{{{\tilde{N}}}})^{\frac{1}{10}+\frac{\delta_2}{5}+\delta_3}},\\
\varepsilon_5&=C_2 \left(472.5 +3.2 B_3(B_4-m_4) \sqrt{C} \right)^{-1} (\log{{{\tilde{N}}}})^{\delta_3}.
\end{aligned}
\eeq
For $\tilde{n}=n$ or $2n$, denote
\begin{align*}
\mathcal{C}_{\tilde{n}}:=\left\lbrace (x,y)\in \T^2:\ \l|\frac{1}{{{\tilde{N}}}}\sum_{j=0}^{{{\tilde{N}}}-1} v_{\tilde{n}}\circ T_{\omega}^j(x,y)-L_{\tilde{n}} \r|> \varepsilon_4 \right\rbrace .
\end{align*}
Then Proposition \ref{lm:longsumwithtrivialbdd} implies that 
\beq\label{meas:Cn'}
\begin{aligned}
\max{(|\mathcal{C}_n|, |\mathcal{C}_{2n}|)} \leq &2\sqrt{2} \, \exp \Big(\frac{\pi}{4} \big [\frac{17}{36}+\frac{B_1}{4B_3^2} -\varepsilon_5  \big]  \Big)\\
&\qquad +\sqrt{2} (C(B_4-m_4))^{-1} {{{\tilde{N}}}^{\frac{1}{5}-\frac{2\delta}{5}}}{(\ln{{{\tilde{N}}}})^{-\frac{1}{5}-\frac{2\delta_2}{5}}} \exp{\left(-2(\ln{{{\tilde{N}}}})^{\delta_2}\right)}.
\end{aligned}
\eeq
Let \[\mathcal{E}:=\mathcal{C}_n\cup \mathcal{C}_{2n}\cup\mathcal{B}^{({{\tilde{N}}})}.\]
For $(x,y)\notin \mathcal{E}$, by \eqref{eq:applyAP4'} we have that,
\beq\label{eq:applyAP5'}
\begin{aligned}
\l| \frac{1}{{{\tilde{N}}}}\log{\|M_{{{\tilde{N}}}}(x,y)\|}+L_n-2L_{2n}\r| \leq \frac{11}{n} e^{-\frac{1}{2} nL_n}+\frac{5U(\lambda, 1)n}{{{\tilde{N}}}}+2\varepsilon_4.
\end{aligned}
\eeq
Together with \eqref{eq:LNLnL2n'}, this implies  that  for any $(x,y)\notin \mathcal{E}$,
\beq\label{eq:applyA6'}
\begin{aligned}
       &\big\| \ind_{\mathcal{E}^c}(x,y) \big(\frac{1}{{{\tilde{N}}}}\log{\|M_{{{\tilde{N}}}}(x,y)\|}-L_{{{\tilde{N}}}} \big) \big\|_{L^{\infty}(\T^2)}\\
\leq &|L_{{{\tilde{N}}}}+L_n-2L_{2n}|+\frac{11}{n} e^{-\frac{1}{2} nL_n}+\frac{5U(\lambda, 1)n}{{{\tilde{N}}}}+2\varepsilon_4\\
\leq &\frac{2n}{{{\tilde{N}}}}(L_n-L_{2n})+\frac{22}{n}e^{-\frac{1}{2} nL_n}+8U(\lambda,1){{{\tilde{N}}}}^{-\frac{7}{5}+\frac{4\delta}{5}}+\frac{5U(\lambda, 1)n}{{{\tilde{N}}}}+2\varepsilon_4\\
=:&\varepsilon_0.
\end{aligned}
\eeq
{{Recall \eqref{est:setfailsAP'} states that}
\beq\nn
|\mathcal{B}^{(N)}|<2N^{-\frac{7}{5}+\frac{4\delta}{5}},\ \ \text{and}\ \ |\mathcal{B}^{(2N)}|\leq 4N^{-\frac{7}{5}+\frac{4\delta}{5}}.
\eeq
Since $\delta>0$, this clearly leads to 
\beq\label{est:setfailsAP}
|\mathcal{B}^{(\tilde{N})}|\leq 2^{\frac{17}{5}-\frac{4\delta}{5}} \tilde{N}^{-\frac{7}{5}+\frac{4\delta}{5}}<16 \tilde{N}^{-\frac{7}{5}+\frac{4\delta}{5}}.
\eeq
}
Combining \eqref{est:setfailsAP} with \eqref{meas:Cn'}, we obtain
\beq\label{meas:BCnC2n}
\begin{aligned}
&\big\| \ind_{\mathcal{E}}(x,y) \big(\frac{1}{{{\tilde{N}}}}\log{\|M_{{{\tilde{N}}}}(x,y)\|}-L_{{{\tilde{N}}}} \big) \big\|_{L^1(\T^2)}\\
\leq &U(\lambda, 1) |\mathcal{C}_n\cup \mathcal{C}_{2n}\cup \mathcal{B}^{({{\tilde{N}}})}|\\
\leq &U(\lambda, 1)\Bigg\{ 4\sqrt{2} \, \exp \Big(\frac{\pi}{4} \big [\frac{17}{36}+\frac{B_1}{4B_3^2}-\varepsilon_5  \big]  \Big)+{{16}}{{\tilde{N}}}^{-\frac{7}{5}+\frac{4\delta}{5}}\\
&\qquad\qquad+2\sqrt{2} (C(B_4-m_4))^{-1} {{{\tilde{N}}}^{\frac{1}{5}-\frac{2\delta}{5}}}{(\ln{{{\tilde{N}}}})^{-\frac{1}{5}-\frac{2\delta_2}{5}}} \exp{\left(-2(\ln{{{\tilde{N}}}})^{\delta_2}\right)}\Bigg\}=:\varepsilon_1.
\end{aligned}
\eeq
By our conditions ($\mathrm{V}$)-($\mathrm{VIII}$), we have
\beq\label{eq:epsilon01upper}
\begin{aligned}
\varepsilon_0& \leq 4\varepsilon_4,\\
\varepsilon_1& \leq 18 U(\lambda, 1) {{\tilde{N}}}^{-\frac{7}{5}+\frac{4\delta}{5}}=:18U(\lambda,1) {{\tilde{N}}}^{\eta},
\end{aligned}
\eeq
Indeed, note that the right-hand-sides of ($\mathrm{V}$) and ($\mathrm{VI}$) are precisely $\varepsilon_4$, which allow us to bound $\varepsilon_0$ by $4\varepsilon_4$. On the other hand,  ($\mathrm{VII}$) and ($\mathrm{VIII}$) simply state that the sum of the first two terms in the braces defining $\varepsilon_1$
are  bounded by the third term, viz.~$2{{\tilde{N}}}^{-\frac{7}{5}+\frac{4\delta}{5}}$. 

Let
\beq\label{def:repsilon23}
\begin{aligned}
\varepsilon_3:&=C_4 \varepsilon_4 {(\ln{{{\tilde{N}}}})^{\delta_4}},\\
\varepsilon_2:&=C_5 \varepsilon_4 {(\ln{{{\tilde{N}}}})^{2\delta_4}},\\
r&=\frac{1-2\delta}{7-4\delta}\ \in (0,1).
\end{aligned}
\eeq
Our condition ($\mathrm{X}$) and ($\mathrm{XI}$) ensure that $\varepsilon_0\leq 4\varepsilon_4<\varepsilon_3<\varepsilon_2$.
Recall that $B_4-m_4$ and $B_5^{({{\tilde{N}}})}-m_5^{({{\tilde{N}}})}$ are as in \eqref{eq:B56m56}. Therefore,  by our condition ($\mathrm{IX}$) we have
\beq\label{eq:B6m6upper}
\left(B_5^{({{\tilde{N}}})}-m_5^{({{\tilde{N}}})}\right){{\tilde{N}}}^{-1}<2\log{R}.
\eeq
Applying Lemma \ref{lm:T2Splitting} to $v_{{{\tilde{N}}}}$, and 
taking \eqref{eq:B6m6upper} into account,
we obtain
\beq\label{eq:vNdev}
\begin{aligned}
&\l| \left\lbrace (x,y)\in \T^2:\ |v_{{{\tilde{N}}}}(x,y)-L_{{{\tilde{N}}}} |>\varepsilon_2\right\rbrace \r|\\
\leq &2(2C_0)^{\frac{1}{2}} \exp{\left(-\pi\left(36 \varepsilon_0+16 B_3\sqrt{\varepsilon_1^r (B_4-m_4)}\right)^{-1}\varepsilon_3\right)}\\
&+C_0\exp{\left(-\pi\left(18 \varepsilon_3+16 B_3\sqrt{U(\lambda, 1)} \sqrt{\varepsilon_1^{1-r} (B_5^{({{\tilde{N}}})}-m_5^{({{\tilde{N}}})})}\right)^{-1}\varepsilon_2\right)}\\
\leq &2(2C_0)^{\frac{1}{2}} \exp{\left(-\pi\left(36 \varepsilon_0+16 B_3\sqrt{\varepsilon_1^r (B_4-m_4)}\right)^{-1}\varepsilon_3\right)}\\
&+C_0\exp{\left(-\pi\left(18 \varepsilon_3+16 B_3\sqrt{U(\lambda, 1)} \sqrt{2 \varepsilon_1^{1-r} {{\tilde{N}}} \log{R} }\right)^{-1}\varepsilon_2\right)}.
\end{aligned}
\eeq
Note that we changed $B_6$ into $U(\lambda,1)$ in this expression.
Inserting  our estimates of $\varepsilon_0, \varepsilon_1$, see \eqref{eq:epsilon01upper}, and choices of $\varepsilon_2, \varepsilon_3, r$, see \eqref{def:repsilon23}, into \eqref{eq:vNdev}, we arrive at
\beq\nn
\begin{aligned}
&\l| \left\lbrace (x,y)\in \T^2:\ |v_{{{\tilde{N}}}}(x,y)-L_{{{\tilde{N}}}} |>C_5 \varepsilon_4 {(\ln{{{\tilde{N}}}})^{2\delta_4}} \right\rbrace \r|\\
\leq &2(2C_0)^{\frac{1}{2}} \exp{\left(\frac{-\pi C_4 \varepsilon_4 {(\ln{{{\tilde{N}}}})^{\delta_4}}}{144 \varepsilon_4+16 B_3\sqrt{(18U(\lambda, 1))^r {{{\tilde{N}}}}^{\eta r} (B_4-m_4)}}\right)}\\
&+C_0\exp{\left(\frac{-\pi C_5 \varepsilon_4 {(\ln{{{\tilde{N}}}})^{2\delta_4}}}{18 C_4 \varepsilon_4 {(\ln{{{\tilde{N}}}})^{\delta_4}}+16 B_3\sqrt{U(\lambda, 1)} \sqrt{2 (18U(\lambda,1))^{1-r} {{{\tilde{N}}}}^{\eta(1-r)+1} \log{R}}}\right)}.
\end{aligned}
\eeq
Using \eqref{eq:4U>1}, and $0\leq r\leq 1$, we  estimate
\[(18U(\lambda,1))^r\leq 18U(\lambda,1),\ \ \text{and}\ \ (18U(\lambda,1))^{1-r}\leq 18U(\lambda,1),\]
respectively. 
Hence we have
\beq\nn
\begin{aligned}
&\l| \left\lbrace (x,y)\in \T^2:\ |v_{{{\tilde{N}}}}(x,y)-L_{{{\tilde{N}}}} |>C_5 \varepsilon_4 {(\ln{{{\tilde{N}}}})^{2\delta_4}} \right\rbrace \r|\\
\leq &2(2C_0)^{\frac{1}{2}} \exp{\left(\frac{-\pi C_4 \varepsilon_4 {(\ln{{{\tilde{N}}}})^{\delta_4}}}{144 \varepsilon_4+ 48 B_3 \sqrt{2 U(\lambda, 1) {{{\tilde{N}}}}^{\eta r} (B_4-m_4)}}\right)}\\
&+C_0\exp{\left(\frac{-\pi C_5 \varepsilon_4 {(\ln{{{\tilde{N}}}})^{2\delta_4}}}{18 C_4 \varepsilon_4 {(\ln{{{\tilde{N}}}})^{\delta_4}}+96 B_3 U(\lambda, 1) \sqrt{{{{\tilde{N}}}}^{\eta(1-r)+1} \log{R}}}\right)}.
\end{aligned}
\eeq

Plugging in our choice of $\varepsilon_4$, see \eqref{def:epsilon45}, and noting that the powers of ${{\tilde{N}}}$ in numerators and 
denominators  cancel out due to our choice of $r$, we infer that 
\beq\nn
\begin{aligned}
&\l| \left\lbrace (x,y)\in \T^2:\ |v_{{{\tilde{N}}}}(x,y)-L_{{{\tilde{N}}}} |>C_5 \varepsilon_4 {(\ln{{{\tilde{N}}}})^{2\delta_4}} \right\rbrace \r|\\
\leq &2(2C_0)^{\frac{1}{2}} \exp{\left(\frac{-\pi C_2 C_4 (\ln{{{\tilde{N}}}})^{\frac{1}{10}+\frac{\delta_2}{5}+\delta_3+\delta_4}}{144 C_2 (\ln{{{\tilde{N}}}})^{\frac{1}{10}+\frac{\delta_2}{5}+\delta_3}+48 B_3\sqrt{2U(\lambda, 1)(B_4-m_4)}}\right)}\\
&+C_0\exp{\left(\frac{-\pi C_2 C_5 {(\ln{{{\tilde{N}}}})^{\frac{1}{10}+\frac{\delta_2}{5}+\delta_3+2\delta_4}}}{18 C_2 C_4  {(\ln{{{\tilde{N}}}})^{\frac{1}{10}+\frac{\delta_2}{5}+\delta_3+\delta_4}}+96 B_3 U(\lambda, 1)\sqrt{\log{R}}} \right)}\\
=&2(2C_0)^{\frac{1}{2}} \exp{\left(\frac{-\pi C_2 C_4 (\ln{{{\tilde{N}}}})^{\delta_4}}{144 C_2 +48 B_3\sqrt{2U(\lambda,1) (B_4-m_4)} (\log{{{\tilde{N}}}})^{-\frac{1}{10}-\frac{\delta_2}{5}-\delta_3}}\right)}\\
&+C_0\exp{\left(\frac{-\pi C_2 C_5  {(\ln{{{\tilde{N}}}})^{\delta_4}}}{18 C_2 C_4 +96 B_3 U(\lambda, 1)\sqrt{\log{R}} {(\ln{{{\tilde{N}}}})^{-\frac{1}{10}-\frac{\delta_2}{5}-\delta_3-\delta_4}}} \right)},
\end{aligned}
\eeq
as desired.
\qed

\bigskip

The readers will note that the constants were chosen in such a way that in the final steps of the proof only powers of $\log N$ remained inside of the exponential. We have found this to be more efficient over intermediate scales. The following,  final,  section of this paper will show how our work up to this point allows for such concrete estimates with specific numbers.

\section{Explicit numbers and proof of Theorem \ref{thm:main}}
\label{sec:numbers}

Our goal here  is to make concrete choices of our parameters so as to arrive at an actual multi-scale scheme for the skew-shift operator from Section~\ref{sec:Herman}.   Let $\mathcal{B}_n$ be as in Definition~\ref{def:Bn}.  The values below were found to be convenient ones, but clearly many other choices could have been made.  

\begin{defn}
\label{def:values} 
Set
\[
R:=4,\; R_{1}:=3,\; R_{2}:= 2,
\]
in Definition~\ref{def:Poisson}. The coupling constant in \eqref{def:transfer} is required to obey $\lambda\in [\frac12,1]$.  Further, in Proposition~\ref{prop:main8}  set  $$\delta:=\frac18,\;  \delta_{2}:=1, \; \delta_{3}:=2,\; \delta_{4}:=\frac32.$$
as well as
\[
C_{2}:= 203, \; C_{4}:=\frac{145}{\pi},\; C_{5}:= \frac{850}{\pi}.
\]
\end{defn}

By an explicit computation, the condition in Definition~\ref{def:epssmall} is satisfied. In fact, one has
\[
B_3^{2}-{289}\big(B_{0}+\frac{13}{20\log(R/R_{1})}\big) > {61}>0
\]

\be{prop}\label{prop:main9}
Let $\omega=\frac{\sqrt{5}-1}{2}$ be the golden ratio, and consider the model~\eqref{def:transfer} with  $\lambda\in [\frac12,1]$ arbitrary but fixed. Let $a\geq 7$ and let $n,N$ be positive integers such that $N\geq 10^{12}$, $n$ divides $N$,  and 
\beq\label{eq:nN}
10^{13} (n+1)^{8} \le N,
\qquad 
\frac{N}{(\log N)^{\frac{92}{3}}}< \frac{1}{2}\l(\frac{203}{22} e^{a/2} n\r)^{\frac{40}{3}}.
\eeq
Impose the conditions 
\begin{enumerate}
\item[(a).]
\beq\nn
nL_n\geq a,
\eeq
\item[(b).]
\beq\nn
L_n-L_{2n}\leq \frac{1}{8}L_n,
\eeq
\item[(c).] 
\beq\nn
\begin{aligned}
\max( |\mathcal{B}_n|, |\mathcal{B}_{2n}|) \leq N^{-\frac{23}{10}}.
\end{aligned}
\eeq
\end{enumerate}
Then we have
\beq\label{eq:LDTN9}
\l| \Big\{  (x,y)\in \T^2:\ |v_{{\tilde{N}}}(x,y)-L_{{\tilde{N}}} |> {5.5\times 10^4} {{\tilde{N}}}^{-\frac{3}{40}} (\log {{\tilde{N}}})^{\frac{53}{10}} \Big\}  \r|\\
\leq  10\exp\left(-(\log {{\tilde{N}}})^{\frac32}\right), 
\eeq
holds for ${{\tilde{N}}}=N$ and $2N$.
\e{prop}

\be{rmk}
We will choose the constant $a=7$ along the inductive multiscale procedure. The only exception is the first step of the induction, which goes from the scale $N_0$ to $N_1$, where for some of our main results we use a larger value of $a$. This is made possible by assumption (i) on the Lyapunov exponent at the initial scale and it is the reason behind the relatively small values of $N_0$ in Theorems \ref{thm:main2} and \ref{thm:main3}.
\e{rmk}

\begin{proof}
We need to check the hypotheses of Proposition~\ref{prop:main8}. We already verified~\eqref{eq:epsklein}, and the conditions of Lemma~\ref{lm:induction'} hold by assumption. 
Let ${{\tilde{N}}}=N$ or $2N$.
The function $$[0.5,1]\to\R:\lambda\mapsto U(\lambda, 4)-\log\lambda$$ is decreasing and positive. 
Hence 
\beq\label{eq:Ulambda4}
0.5<U(1,4)\leq U(\lambda, 4)-\log\lambda\leq U(\frac{1}{2}, 4)-\log{\frac{1}{2}}<1.
\eeq
Further, the constant $C$ in $(\mathrm{I})$ satisfies $C<11.97$. So that condition is implied by the stronger one
\[
12^{8}(4\log(2)n+2\log(2)+1)^{8}\le {{\tilde{N}}}
\] 
which we may further strengthen to 
\[
36^{8}(n+1)^{8}<10^{13}(n+1)^{8}\le N,
\]
which is the left-hand side of~\eqref{eq:nN}. Condition~$(\mathrm{II})$  holds, as does~$(\mathrm{III})$ since $\exp(4(\log {{\tilde{N}}})^{\delta_{2}})={{\tilde{N}}}^{4}\ge {{\tilde{N}}}+1$. Condition~$(\mathrm{IX})$ is implied by the stronger one
\[
{{\tilde{N}}}> \frac{1+\log R}{\log R} = \frac{2\log(2)+1}{2\log(2)} \sim 1.721
\]
which clearly holds. In view of~\eqref{eq:B56m56} and \eqref{eq:Ulambda4}, we have
\beq\label{eq:B5-m5number}
4\log(2)+0.5\leq B_4-m_4\le 4\log(2)+1.
\eeq
Condition~$(\mathrm{IV})$ will therefore hold provided 
\[
{{\tilde{N}}}^{\frac18} \log({{\tilde{N}}}) - 21 {{\tilde{N}}}^{-\frac{27}{40}} (\log({{\tilde{N}}}))^{\frac{27}{10}} - 181>0
\]
The left-hand side is increasing in ${{\tilde{N}}}$, and one checks by explicit computation that it is positive if $N\ge 10^{8}$. So this condition holds as well.  Condition~$(\mathrm{VIII})$ is implied by the following  one
\[
{{\tilde{N}}}^{-\frac{13}{10}} > \frac{2\sqrt{2}}{C(4\log(2)+0.5)} {{\tilde{N}}}^{-\frac{37}{20}} (\log({{\tilde{N}}}))^{-\frac35}
\]
Simplifying this, one obtains the stronger condition
\[
{{\tilde{N}}}^{\frac15}(\log({{\tilde{N}}}))^{\frac35}>0.08
\]
which holds provided $N\ge 2$. So condition~$(\mathrm{VIII})$ holds. 

Next, we look at condition~$(\mathrm{VI})$. Using the assumed lower bound $nL_{n}\ge a$ we find the condition 
\[
22 e^{-a/2}< 203 \, n\, {{\tilde{N}}}^{-\frac{3}{40}}  (\log({{\tilde{N}}}))^{\frac{23}{10}.}
\]
We recall that $\tilde N\in \{N,2N\}$ and estimate $\log (\tilde N)\geq \log (N)$. This inequality follows from the upper bound in \eqref{eq:nN}. 
For condition~$(\mathrm{VII})$, one checks that it follows from the slightly stronger
\[
{{\tilde{N}}}^{-\frac{13}{10}} - 5.66\exp\big( 0.374 - 0.05(\log({{\tilde{N}}}))^{2}\big)>0
\]
which holds for ${{\tilde{N}}}\ge 10^{12}$ (but it fails for $10^{11}$). Hence we impose the second lower bound in~\eqref{eq:nN}.
For condition~$(\mathrm{V})$,  we use
\[
L_{n}-L_{2n}\le \frac18L_{n}\le \frac{U(\lambda,1)}{8} \le \frac14
\]
 and so it suffices to check that
 \[
 \frac{21n}{2{{\tilde{N}}}} + 16 {{\tilde{N}}}^{-\frac{13}{10}} < 203 {{\tilde{N}}}^{-\frac{3}{40}} (\log({{\tilde{N}}}))^{\frac{23}{10}}
 \]
Bounding $n$ in terms of $N$ via~\eqref{eq:nN} and discarding the $\log {{\tilde{N}}}$ on the right-hand side reduces us to 
\[
 \frac{21}{2\cdot 10^{\frac{13}{8}}}\cdot N^{-\frac78} + 16 {{\tilde{N}}}^{-\frac{13}{10}} < 203 {{\tilde{N}}}^{-\frac{3}{40}} 
\]
This holds for all $N\ge1$ so we are done with~$(\mathrm{V})$.  Finally, we turn to~$(\mathrm{X})$ and~$(\mathrm{XI})$. Using $N\ge 10^{12}$ they hold provided 
\[
C_{4}\ge 0.03,\qquad C_{5}\ge 0.007\cdot C_{4}
\]
Our actual values assigned to these constants satisfy 
\[
{46 < C_{4} < 47,\qquad C_{5}>270}
\]
and so all conditions of Proposition~\ref{prop:main8} hold. 

As for the conclusion of that proposition, we first compute $C_{2}C_{5}<{5.5\times 10^4}$. Thus, the size of the deviations satisfy
\[
C_2 C_5  {{{\tilde{N}}}^{-\frac{1}{10}+\frac{\delta}{5}}}{(\ln{{{\tilde{N}}}})^{\frac{1}{10}+\frac{\delta_2}{5}+\delta_3+2\delta_4}}< 
{5.5\times 10^4} {{\tilde{N}}}^{-\frac{3}{40}} (\log {{\tilde{N}}})^{\frac{53}{10}}
\]
as stated in~\eqref{eq:LDTN9}.  As for the measure bound, we calculate that 
\begin{align*}
&2(2C_0)^{\frac{1}{2}} + C_{0}<10,\\
&U(\lambda,1)=\frac{1}{2}\log{\left(\left(4\lambda+2 \right)^2+2\right)}\leq \frac{1}{2}\log{38}.
\end{align*}
Thus, in view of \eqref{eq:B5-m5number}, one has, $48B_3\sqrt{2U(\lambda, 1)(B_4-m_4)}<11518$, and 
\[
{144}+11518 C_{2}^{-1}(\log {{\tilde{N}}})^{-\frac{23}{10}}\le {144}+\frac{11518}{203}(12\log 10)^{-\frac{23}{10}}<{145}.
\]
Hence the first exponential in the measure bound of Proposition~\ref{prop:main8} contributes less than 
\[
\exp\left(-\pi C_{4}(\log {{\tilde{N}}})^{\frac32}/{{145}}\right) < \exp\left(-(\log {{\tilde{N}}})^{\frac32}\right).
\]
For the second exponential, we have $18 C_4<831$, $96B_3 U(\lambda,1)\sqrt{2\log{2}}<13317$,
and 
\[
{831}+13317 C_{2}^{-1}(\log {{\tilde{N}}})^{-\frac{38}{10}}\le {831}+\frac{13317}{203} (12\log 10)^{-\frac{38}{10}}<{850}.
\]
Hence, the second exponential contributes less than
\[
\exp\left(-\pi C_{5}(\log {{\tilde{N}}})^{\frac32}/850\right)<\exp\left(-(\log {{\tilde{N}}})^{\frac32}\right),
\]
and we are done. 
\end{proof}

\subsection{Proof of Theorem \ref{thm:main}}
Let $N_0:=2\times 10^{37}$. We define a sequence of scales $N_j:=N_{j-1}^9$ for $j\geq 1$. In particular, $N_1> 5\times 10^{335}$. The proof is based on an induction on scales, where at every step we first apply Lemma \ref{lm:sequencescales} to control the Lyapunov exponent at the next scale. Afterwards, we apply Proposition \ref{prop:main9} to obtain the large deviation estimate at the next scale and then we continue the induction.

For later purposes, we note some properties of this choice of scales. The last inequality is the main reason why we need to choose the scale so that $N_1$ is large.

\be{lm}\label{lm:scales}
Recall that we defined $N_{j+1}:=N_j^9$ with $N_0=2\times 10^{37}$. For all $j\geq 1$, we have the following bounds:
\beq\label{eq:NjNj+1'}
10^{13} (N_{j-1}+1)^8\leq N_{j},\qquad \frac{N_{j}}{(\log N_j)^{\frac{92}{3}}}< \frac{1}{2}\l(\frac{203}{22} e^{7/2} N_{j-1}\r)^{\frac{40}{3}}
\eeq
as well as
\beq\label{eq:NjNj+1}
10 \exp\left(-(\log {{N_j}})^{\frac32}\right)\leq  (N_j^9)^{-2.3}=N_{j+1}^{-2.3},
\eeq
and 
\beq\label{eq:NjNj+1''}
5.5\times 10^4 {N_j}^{-\frac{3}{40}} (\log {N_j})^{\frac{53}{10}}\leq \frac{1}{20} L_{N_0}.
\eeq
\e{lm}

\be{proof}
From the definition of the $N_j$, we have
\beq\label{eq:followsfrom0}
10^{13} (N_{j-1}+1)^{8} \le N_j< N_{j-1}^{13},
\eeq
and this implies \eqref{eq:NjNj+1'}. Notice that we have $N_j\geq N_1>5\times 10^{335}$ for all $j\geq 1$. Then \eqref{eq:NjNj+1} follows from the inequality
\beq\label{eq:followsfrom}
10 \exp\left(-(\log {{x}})^{\frac32}\right)\leq (x^9)^{-2.3},
\eeq
which holds for all $x\geq 2.06\times 10^{186}$. For \eqref{eq:NjNj+1''}, we note 
$$
5.5\times 10^4 {x}^{-\frac{3}{40}} (\log {x})^{\frac{53}{10}}\leq 10^{-5}=\frac{1}{20}\times 2\times 10^{-4}\leq \frac{1}{20} L_{N_0},
$$
where the first inequality holds for all $x\geq 10^{334}$. This proves the lemma.
\e{proof}

We will inductively apply Lemma \ref{lm:sequencescales} to $j=1,2,3,...$. We begin with $j=1$. Condition (1) of Lemma \ref{lm:sequencescales} follows from our assumptions (i) and (ii) and the fact that $N_0=2\times 10^{37}$. Condition (2) with $j=1$ is fulfilled since
\beq\nn
\frac{1}{N_0}\exp{\left(-\frac{1}{2}N_0L_{N_0}\right)}\leq \frac{1}{2\times 10^{37}} <\frac{2}{512} 10^{-4}\leq \frac{1}{512}L_{N_0}.
\eeq
For condition (3), recall that $N_1>5\times 10^{335}$, $\delta=\frac{1}{8}$ and $C_3=U(\lambda, 1)\leq \frac{1}{2}\log{38}$.
Then we have
\beq\nn
\begin{aligned}
N_1^{-\frac{13}{10}}<(5\times 10^{335})^{-1.3}<\frac{1}{320 \log{38}} 10^{-4}\leq \frac{1}{1280 U(\lambda,1)} L_{N_0}.
\end{aligned}
\eeq
For condition (4),
\beq\label{eq:BN0}
\max{(|\mathcal{B}_{N_0}|, |\mathcal{B}_{2N_0}|)}\leq N_0^{-21}=(N_1)^{-7/3}\leq N_1^{-2.3}.
\eeq
Hence Lemma \ref{lm:sequencescales} applies to $j=1$, and yields
\beq\label{eq:LN1}
N_1 L_{N_1}\geq 7,\qquad  L_{N_1}-L_{2N_1}\leq \frac{1}{8}L_{N_1},
\eeq
and
\beq\label{eq:LN1'}
L_{2N_1}\geq \frac{1}{2}L_{N_0}.
\eeq
We would like to apply Lemma \ref{lm:sequencescales} for $j=2$. This requires measure estimates for $\mathcal{B}_{N_1}$ and $\mathcal{B}_{2N_1}$. To this end, we invoke Proposition \ref{prop:main9} with $n=N_0$, $N=N_1$ and $a=7$. Let us check that its conditions are satisfied. First, we have \eqref{eq:nN} by applying \eqref{eq:NjNj+1'} with $j=0$. Moreover, conditions (a)-(c) hold by assumptions (i)-(iii) and \eqref{eq:BN0}. Hence, we can apply Proposition \ref{prop:main9} and obtain that, for $\tilde{N}=N_1$ and $2N_1$,
$$
\begin{aligned}
&\l| \Big\{  (x,y)\in \T^2:\ |v_{{\tilde{N}}}(x,y)-L_{{\tilde{N}}} |>  5.5 \times 10^4 {{\tilde{N}}}^{-\frac{3}{40}} (\log {{\tilde{N}}})^{\frac{53}{10}} \Big\}  \r|
\\
\leq& 10\exp{\left(-(\log{N_1})^{2.3}\right)}
\leq N_2^{-2.3}.
\end{aligned}
$$
In the second step, we used \eqref{eq:NjNj+1} with $j=1$. To turn this into measure estimates for $\mathcal{B}_{\tilde N}$, notice that \eqref{eq:NjNj+1''} and \eqref{eq:LN1'} imply 
\beq\nn
5.5 \times 10^4 {{\tilde{N}}}^{-\frac{3}{40}} (\log {{\tilde{N}}})^{\frac{53}{10}}\leq \frac{1}{20}L_{N_0}\leq \frac{1}{10}L_{\tilde{N}}.
\eeq
Therefore,
\beq\label{eq:BN1'}
\max{(|\mathcal{B}_{N_1}|, |\mathcal{B}_{2N_1}|)}\leq N_2^{-2.3}.
\eeq
We have shown how to pass from scale $N_0$ to $N_1$ via Lemma \ref{lm:sequencescales} and Proposition \ref{prop:main9}, by using the properties \eqref{eq:NjNj+1'}-\eqref{eq:NjNj+1''}.

 We can now iterate this procedure: We apply Lemma \ref{lm:sequencescales} with $j=2$. The main input is the measure estimate \eqref{eq:BN1'}, which verifies condition (4). The remaining conditions hold by our choice of scales, \eqref{eq:LN1}, \eqref{eq:LN1'} and assumptions (i) and (ii). (Notice that the sums in conditions (2) and (3) are rapidly convergent.) From Lemma \ref{lm:sequencescales}, we obtain estimates of $L_{N_2}$ and $L_{2N_2}$, in particular $L_{2N_2}\geq\frac{1}{2}L_{N_0}$. Then Proposition \ref{prop:main9} yields the measure estimates for $\mathcal{B}_{N_2}$ and $\mathcal{B}_{2N_2}$, which is the key input for Lemma \ref{lm:sequencescales} with $j=3$, etc. We conclude that, after $k$ steps of this procedure, we have
\beq\nn
L_{2N_k}\geq \frac{1}{2}L_{N_0}.
\eeq
This yields
\beq\nn
L\geq \frac{1}{2}L_{N_0},
\eeq
by taking $k\rightarrow\infty$, and we have proved Theorem \ref{thm:main}.
\qed

\subsection{Proof of Theorems \ref{thm:main2}}
We follow the general line of argumentation of Theorem \ref{thm:main}. The only difference is that in the sequence of scales $N_j$, we take the first step to be very large. Namely, while $N_0=3\times 10^5$, we define
\beq\label{eq:samescales}
N_1:=3\times 10^{334},\qquad N_{j+1}:=N_{j}^9,\,\,\forall j\geq 1.
\eeq
Notice that for $j\geq 1$, the scales $N_j$ are essentially the ones used in the proof of Theorem \ref{thm:main} above. Therefore we have the following analog of Lemma \ref{lm:scales}, in which \eqref{eq:NjNj+1} for $j=1$ is replaced 

\be{lm}\label{lm:scales'}
We have
\beq\label{eq:newNjNj+1'}
10^{13} (N_{0}+1)^8\leq N_{1},\qquad \frac{N_{1}}{(\log N_1)^{\frac{92}{3}}}< \frac{1}{2}\l(\frac{203}{22} e^{30} N_{0}\r)^{\frac{40}{3}}.
\eeq
Moreover, for all $j\geq 1$, we have the bounds
\beq\label{eq:NjNj+1'analog}
10^{13} (N_{j}+1)^8\leq N_{j+1},\qquad \frac{N_{j+1}}{(\log N_{j+1})^{\frac{92}{3}}}< \frac{1}{2}\l(\frac{203}{22} e^{7/2} N_{j}\r)^{\frac{40}{3}},
\eeq
as well as \eqref{eq:NjNj+1} and \eqref{eq:NjNj+1''}.
\e{lm}

Except for \eqref{eq:newNjNj+1'}, the bounds are only concerned with $N_j$, $j\geq 1$ and therefore follow in the same way as for Lemma \ref{lm:scales'}. The new bound \eqref{eq:newNjNj+1'} follows from
$$
\l(2\frac{N_{1}}{(\log N_1)^{\frac{92}{3}}}\r)^{\frac{3}{40}} \frac{1}{\frac{203}{22} e^{30}}<29974<N_0.
$$
This establishes Lemma \ref{lm:scales'}. As before, we will successively apply Lemma \ref{lm:sequencescales} and Proposition \ref{prop:main9} and iterate. We begin by applying Lemma \ref{lm:sequencescales} with $j=1$. Condition (1) is immediate from assumption (i) and $N_0=3\times 10^5$; indeed:
\beq\label{eq:initialnLn}
N_0L_{N_0}\geq 2N_0 10^{-4}=60.
\eeq
We can use this inequality to verify condition (2) as well:
\beq\label{eq:firstterm}
\frac{1}{N_0} e^{-\frac{N_0L_{N_0}}{2}}\leq \frac{1}{3}10^{-5} e^{-30}<10^{-18}< \frac{2}{512} 10^{-4}< \frac{L_{N_0}}{512}.
\eeq
Condition (3) holds by our choice of $N_1$. Finally, condition (4) holds by assumption (iii):
\beq\label{eq:waschecked}
\max{(|\mathcal{B}_{N_0}|, |\mathcal{B}_{2N_0}|)}\leq N_0^{-141}< (3\times 10^{334})^{-2.3}=N_1^{-2.3}.
\eeq
Hence, Lemma \ref{lm:sequencescales} applies and yields \eqref{eq:LN1} and \eqref{eq:LN1'} as before. Next, we verify the assumption of Proposition \ref{prop:main9} with $n=N_0$ and $N=N_1$. The key difference is that we now take $a=60$. 
This is made possible by \eqref{eq:initialnLn}, since it verifies condition (a) of Proposition \ref{prop:main9}. Condition (b) is immediate from assumption (ii) and condition (c) was checked in \eqref{eq:waschecked}. The bounds \eqref{eq:nN} hold by \eqref{eq:newNjNj+1'}. Therefore, we can apply Proposition \ref{prop:main9}. Combining the resulting estimate with \eqref{eq:NjNj+1''} for $j=1$,  \eqref{eq:LN1'} and \eqref{eq:NjNj+1} for $j=1$, we obtain the measure estimate
$$
\max{(|\mathcal{B}_{N_1}|, |\mathcal{B}_{2N_1}|)}\leq N_2^{-2.3}.
$$
At this point, we have moved completely from scale $N_0$ to scale $N_1$ and can follow the argument from Theorem \ref{thm:main} verbatim. In particular, we take $a=7$ in every subsequent application of Proposition \ref{prop:main9}. The only difference is the $m=0$ term in condition (2) of Lemma \ref{lm:sequencescales}, which now involves $N_0=3\times 10^{-5}$. By \eqref{eq:firstterm}, we can replace condition (2) by the stronger bound
$$
\sum_{m=1}^{j-1} N_m^{-1} e^{-\frac{1}{2} N_mL_{N_m}}<\frac{10^{-4}}{256}-10^{-18}
$$
and this holds by our choice of scales and the estimates \eqref{eq:LN1'} along the induction (notice again the rapid convergence of the series). We conclude that
$$
L\geq \frac{1}{2}L_{N_0}
$$
and this proves Theorem \ref{thm:main2}.
\qed

\subsection{Proof of Theorem \ref{thm:main3}}
Again, we follow the same steps for a different sequence of scales. We have $N_0=3\times 10^4$. We define the sequence of scales $N_j$, $j\geq 1$ by 
$$
N_1:=3\times 10^{320}, \qquad N_{j+1}:=N_{j}^9,\,\,\forall j\geq 1.
$$
We still have Lemma \ref{lm:scales'} for this choice of scales. Indeed, \eqref{eq:NjNj+1'analog} and \eqref{eq:NjNj+1} still follow from the inequalities  \eqref{eq:followsfrom0} and \eqref{eq:followsfrom} given in the proof of Lemma \ref{lm:scales}. For \eqref{eq:NjNj+1''}, we now use assumption (i) to find
$$
5.5\times 10^4 {x}^{-\frac{3}{40}} (\log {x})^{\frac{53}{10}}\leq 10^{-4}=\frac{1}{20}\times 2\times 10^{-3}\leq \frac{1}{20} L_{N_0},
$$
where the first inequality holds for all $x\geq 10^{320}$, so in particular for all $N_j$ with $j\geq 1$. Finally, \eqref{eq:newNjNj+1'} follows from
$$
\l(2\frac{N_{1}}{(\log N_1)^{\frac{92}{3}}}\r)^{\frac{3}{40}} \frac{1}{\frac{203}{22} e^{30}}<2938<N_0.
$$
This establishes Lemma \ref{lm:scales'} for the new choice of scales.

Next we check the hypotheses for Lemma \ref{lm:sequencescales} with $j=1$. Condition (1) is immediate from assumption (i) and $N_0=3\times 10^4$:
\beq\label{eq:initialLn'}
N_0L_{N_0}\geq 2N_0 10^{-3}=60
\eeq
(Compare this to \eqref{eq:initialnLn}.) Condition (2) holds by
\beq\label{eq:firstterm'}
\frac{1}{N_0} e^{-\frac{N_0L_{N_0}}{2}}\leq \frac{1}{3}10^{-4} e^{-30}<10^{-17}< \frac{2}{512} 10^{-3}< \frac{L_{N_0}}{512},
\eeq
where we used \eqref{eq:initialLn'}. Condition (3) holds by our choice of scales $N_j$, $j\geq 1$, and condition (4) holds by assumption (iii):
\beq\label{eq:waschecked'}
\max{(|\mathcal{B}_{N_0}|, |\mathcal{B}_{2N_0}|)}\leq N_0^{-165}< (3\times 10^{320})^{-2.3}=N_1^{-2.3}.
\eeq
Therefore we can apply Lemma \ref{lm:sequencescales} and obtain \eqref{eq:LN1} and \eqref{eq:LN1'}. As in the proof of Theorem \ref{thm:main2}, the first application of Proposition \ref{prop:main9} utilizes $a=60$. This is made possible by \eqref{eq:initialnLn}, since it verifies condition (a) of Proposition \ref{prop:main9}. Now we iterate the argument in the same way as was done for Theorem \ref{thm:main} and \ref{thm:main2}. (Notice that the series in conditions (2) and (3) of Lemma \ref{lm:sequencescales} are still rapidly convergent.) The end result is the lower bound
$$
L\geq \frac{1}{2}L_{N_0}
$$
and Theorem \ref{thm:main3} is proved.
\qed

\end{document}